\documentclass[reqno, 12 pt]{amsart}

\usepackage{amsthm}
\usepackage{amssymb}
\usepackage{latexsym}
\usepackage{cite}
\usepackage{url}
\usepackage{amsmath}
\usepackage{verbatim}
\usepackage{mathrsfs}
\usepackage[norelsize,ruled,vlined,boxed,commentsnumbered, linesnumbered]{algorithm2e}
\usepackage{graphicx}
\usepackage{color}
\usepackage{epsf}
\usepackage{enumerate}
\usepackage[width=6.8cm]{caption}
\usepackage{geometry}
\usepackage{hyperref}
\definecolor{dark-blue}{rgb}{0,0,0.6}
\definecolor{Purple}{rgb}{0.2,0,0.25}
\hypersetup{breaklinks=true,
pagecolor=white,
colorlinks=true,
citecolor=dark-blue,
filecolor=black,%
linkcolor=Purple,%
urlcolor=blue
}
\newtheorem{thm}{Theorem}[section]

\newtheorem{lem}[thm]{Lemma}
\newtheorem{prop}[thm]{Proposition}

\newtheorem{defin}[thm]{Definition}

\theoremstyle{definition}
\newtheorem{expl}[thm]{Example}
\newtheorem{remark}[thm]{Remark}

\newtheorem{method}[thm]{Method}

\newcommand{\R}{\mathbb{R}}

\newcommand{\N}{\mathbb{N}}

\newcommand{\M}{\mathscr{M}}

\newcommand{\wt}{\widetilde}

\numberwithin{equation}{section}

\newcommand{\bref}[1]{\textbf{\ref{#1}}} 
\newcommand{\beqref}[1]{\textbf{(\ref{#1})}} 

\title[The projector algorithm]{The projector algorithm: a simple parallel algorithm for computing Voronoi diagrams and Delaunay graphs}
\author{Daniel Reem}

\address{The Center for Mathematics and Scientific Computation (CMSC), University of Haifa, Mt. Carmel, Haifa, 3498838,  Israel.} 
\email{dream@math.haifa.ac.il}


\date{July 10, 2023}

\begin{document}
\maketitle

\begin{abstract}
The Voronoi diagram is a certain geometric data structure which has found numerous applications in various scientific and technological fields. The theory of algorithms for computing 2D Euclidean 
Voronoi diagrams of point sites is rich and useful, with several different and important algorithms. However, this theory has been quite steady during the last few decades in the sense that no essentially new algorithms have entered the game. In addition, most of the known algorithms are serial in 
nature and hence cast inherent difficulties on  the possibility to compute the diagram in parallel. In this paper we present the projector algorithm: a new and simple algorithm which enables the (combinatorial) computation of 2D Voronoi diagrams. The algorithm is significantly different from 
previous ones and some of the involved concepts in it are in the spirit of linear programming and optics. Parallel implementation is naturally supported since each Voronoi cell (actually, even just portions of one cell) can be computed independently of the other cells. A new combinatorial structure for  representing the cells (and any convex polytope) is described along the way and the computation of the 
induced Delaunay graph is obtained almost automatically.  
\end{abstract}

\tableofcontents
\section{Introduction}\label{sec:Introduction}
\subsection{Background} 
 The Voronoi diagram is a certain geometric data structure which appears in many fields in  science and technology and has found numerous applications, among them in computer graphics, geographic information  systems (GIS), molecular biology, data analysis, astrophysics, signal processing, mesh generation, coding, computational geometry, material engineering, pure mathematics and many more areas. See, for instance, \cite{Aurenhammer,CSKM2013,ConwaySloane,VoronoiCVD_Review,VoronoiWeb,GruberLek,OBSC} for some reviews and illustrations. In its simplest and widespread form, this diagram is a certain decomposition of the Euclidean plane, or a region $X$ in the plane, into cells induced by the Euclidean distance and by a collection of $n\in\N$ distinct points $p_1,\ldots,p_n$ (called ``sites'' or ``generators'' or ``particles''). More precisely, the Voronoi cell $R_k$ associated with the site $p_k$ is the set of all the points in $X$ whose distance to $p_k$ is not greater than their distance to the other sites $p_j$, $j\neq k$.  

Because it is widely used, the Voronoi diagram, which has other names as well (e.g., the Voronoi tessellation, the Voronoi decomposition,  Dirichlet tessellation), has attracted a lot of attention during the last four decades and actually much before, especially in the case of 2-dimensional Euclidean setting  of point sites. In particular, many algorithms  for computing these diagrams in the above-mentioned setting have been published. Among them we mention the naive method \cite[pp. 230-233]{OBSC}, the divide-and-conquer method \cite{AAAHJPR2009}, \cite[pp. 251-257]{OBSC}, \cite{ShamosHoey1975}, the incremental method \cite{GreenSibson1977},  \cite{GuibasKnuthSharir}, \cite{OhyaIriMurota1984}, \cite[pp. 242-251]{OBSC},    methods based on sweep \cite{Fortune1987}, \cite[pp. 257-264]{OBSC}, \cite{XinWangXiaMueller-WittigWangHe2013jour}, methods based on geometric transforms such as convex hulls \cite{Qhull, Brown1979,Brown1980,Edelsbrunner-book-1987} (see also  \cite{Aurenhammer,ChazelleMatousek1995,ClarksonShor1989} and the references therein) or Delaunay  triangulations  \cite[Chapter 3]{ChengDeyShewchuk2012}, \cite{GuibasStolfi1985}, \cite[pp. 275-80]{OBSC}, methods based on lower envelopes \cite{SetterSharirHalperin2010},\cite[p. 241]{SharirAgarwal}  and methods for very specific  configurations   \cite{AggarwalGuibasSaxeShor}.

It can be seen that the theory of algorithms for computing 2D Euclidean Voronoi diagrams of point sites is rich and useful, with many different and important algorithms and analyses. However, this theory has been quite steady during the last decades in the sense that no essentially new algorithms have entered the game (though several valuable improvements or variations of known algorithms have appeared).  

Another property of this theory is that most of the known algorithms are sequential (serial) in nature and they cannot compute each of the Voronoi cells independently of the other ones. Instead, they consider the diagram as a combinatorial structure and compute it as a whole in a sequential way. This fact casts inherent difficulties on any attempt to implement these algorithms in a parallel computing environment. It is therefore not surprising to see claims such as ``Parallelizing algorithms in computational geometry usually is a complicated task since many of the techniques used (incremental insertion or plane sweep, for instance) seem inherently sequential'' \cite[p. 367]{Aurenhammer}, or ``It  is  seldom  obvious  how  to  generate  parallel  algorithms  in  this  area [computational geometry] since popular  techniques  such  as  contour  tracing,  plane  sweeping,  or  gift  wrapping involve  an  explicitly  sequential   (iterative)  approach''  \cite[p. 293]{ACGOY1988}, or ``our technique, like all previous deterministic parallel algorithms, is based on the serial algorithm due to Shamos and Hoey''  \cite[p. 570]{ColeGoodrichODunlaing1996}. 

Nevertheless, starting from Chow \cite{Chow1980} a corresponding theory for parallelizing the computation of Voronoi diagrams has been developed  \cite{ACGOY1988, BermanLingas1997, ColeGoodrichODunlaing1996, EvansStojmenovic1989, GoodrichODunlaingYap1993, HagerapKatajainen1993, 
HALH2005, LeeJou1995, LevcopoulosKatajainenLingas1988, MacKenzieStout1998,   
RajasekaranRamaswami2002, Roos1994conf} and has been extended to related geometric 
structures such as the Delaunay triangulation and convex hulls \cite{AmatoGoodrichRamos1994, BHMT1999, DadounKirkpatrick, DFR-C1996, FragakisOnate2008, Goodrich1987jour, Meyerhenke2005, ReifSen1992, Schwarzkopf1989, SpielmanTengUngor2007, TrefftzSzakas2003, VemuriVaradarajanMayya1992}. See \cite{AtallahChen2000, Goodrich2004, Ramaswami1998, Sen1989} and \cite[pp. 367-369]{Aurenhammer} for a few surveys.  

In the works mentioned above the idea is to somehow share the work between the many processing units (namely processors, cores, etc.), under certain assumptions on the computational model.  Unfortunately, the above-mentioned sequential nature of the involved algorithms complicates the implementation of many of these parallel algorithms. In addition, a common assumption in the above-mentioned works is that there are many processing units, e.g., $O(n)$ or $O(n/\log(n))$, where $n$ is the number of sites. This assumption seems to cast difficulties on a practical implementation when the number of sites is much larger than the number of processing units, as frequently happens in real world applications. Additional assumptions which are often imposed are that the sites do no form degenerate configurations (e.g., no four sites are located on the same circle,  no two sites have the same first or second coordinate, etc.), that $n$ is a power of two, and so on. Any such an assumption complicates the implementation of the corresponding algorithms and/or limits their usage. The case of parallel algorithms for geometric structures related to Voronoi diagrams is somewhat similar: these (often theoretical) algorithms are either based on the former works or vice versa, or they use somewhat similar techniques and impose somewhat similar assumptions. 

In addition to the algorithms mentioned above, it is possible to mention other parallel algorithms of a somewhat different nature of either Voronoi diagrams or closely related geometric structures, for instance  \cite{Gonzalez2016jour,  StarinshakOwenJohnson2014jour, WangCuiRuiChengYingxiaWuYuan2014jour, WuRuiSuChengWang2014jour, XinWangXiaMueller-WittigWangHe2013jour}. However, these algorithms  also contain an inherent serial  component and they cannot compute each Voronoi cell independently of the other cells. Indeed, essentially  the parallelization is done by first performing a preprocessing stage in which the world $X$ is divided into sub-domains, then applying a known serial algorithm (such as divide-and-conquer or sweep) to each sub-domain  simultaneously, and then somehow sharing the information obtained in each sub-domain (possibly after a  merging step between neighbor sub-domains) so that the global Voronoi cells will be obtained correctly from the Voronoi cells which were obtained in the sub-domains. These algorithms are mostly suitable for certain distribution of sites (e.g., uniform distribution or distributions which are not far from being uniform in the sense that, e.g., after the original domain is divided into sub-domains, the distribution of sites in each sub-domain is roughly uniform), but they are less suitable and more complicated (and slower) for other types of distributions, for instance distributions in which the sites form highly degenerate configurations.

The motivation for developing parallel-in-nature algorithms for computational tasks stems from several natural reasons. One important reason is the ability to compute in a fast manner much larger inputs than the ones computed nowadays,  in numerous fields, or to perform in a fast way computations which require many iterations, such as Centroidal Voronoi Diagrams (CVD) \cite{DuEmelianenkoJu2006,VoronoiCVD_Review}. Another reason is that in recent years most of the computing devices (various types of computers, cell phones, graphics processors, etc.) arrive with several (sometimes with hundreds or even thousands) processing units which are just waiting to be used. Large networks of such computing devices can also be  used for parallel computing tasks.

By taking into account all of the above-mentioned considerations, it is natural to ask whether there exists an algorithm which can compute each of the Voronoi cells independently of the other ones, and hence can provide a simple way to compute the Voronoi diagram in parallel. To the best of our knowledge, only one such algorithm has been discussed in the literature, namely the naive one which computes each of the cells by intersecting corresponding halfplanes \cite[pp. 230-233]{OBSC} (see also the recent paper \cite{LuLazarRycroft2022prep} which describes an implementation of this method based on an adaptation to the 2D case of the 3D Voro++ package \cite{Rycroft2009jour}). This (very) veteran algorithm is rather simple, at least from a high-level perspective, but its time complexity is relatively slow: $O(n^2\log(n))$ for the whole diagram of $n$ sites in the worst case, assuming one processing unit is involved.  As claimed in \cite{BentleyWeideYao1980}, on the average (under the assumption of uniform distribution) its time complexity should behave as $O(n)$, but we have not seen any mathematical proof of this assertion. Similarly, no theoretical justification is given to the $O(n)$ average case time complexity of the  variation of the naive algorithm for the Delaunay triangulation which appears in \cite{ChenGotsman2013jour}.

In 2009, a new algorithm which allows the approximate computation of Voronoi diagrams in a general setting (general sites, general norms, general dimension) was published in \cite{ReemISVD2009proc}. This algorithm is based on the possibility to represent each Voronoi  cell as a union of rays (line segments), and it approximates the cells by considering a plurality of approximating rays which are shot in various directions. See Figures \bref{fig:Voronoi}--\bref{fig:VoronoiRays} for an illustration, \cite{Reem2018jour} for an application, and \cite{Vdream2017web} for an online implementation. This ``ray-shooting algorithm'' allows the approximate computation of each cell independently of the other ones. However, although in principle this algorithm may compute the combinatorial structure of the cell, to the best of our knowledge, no details regarding how to do so have been published so far; besides, it seems that even if somehow the algorithm mentioned in \cite{ReemISVD2009proc} or its output can be used to compute the combinatorial structure of the cell, then this will probably be done in a non-immediate and non-efficient way. The reason for this expected inefficiency is because one needs to detect somehow the corresponding combinatorial components, and for achieving this task many rays should be considered, and the information obtained from them should be analyzed correctly. Unfortunately, it is not clear in advance in which direction to shoot a ray such that it will hit a vertex exactly, and another issue which complicates the situation is the fact that the endpoints of the rays are computed only approximately, and hence accumulating errors can lead to a wrong conclusion regarding the  combinatorial structure. It is therefore natural 
to ask whether one can modify somehow the algorithm of \cite{ReemISVD2009proc}  in such a way that the modified algorithm will allow a simple and efficient computation of the combinatorial structure of the Voronoi cells.

\begin{figure}[t]
\begin{minipage}[t]{0.43\textwidth}
\begin{center}{\includegraphics[trim=2 2 2 2, clip=true, scale=1.1]{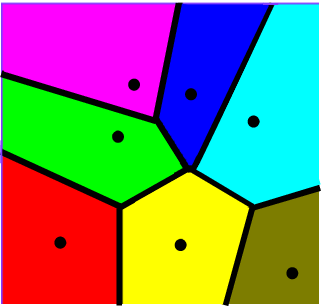}}
\end{center}
 \caption{A Voronoi diagram of 7 point sites in a square in the Euclidean plane.}
\label{fig:Voronoi}
\end{minipage}
\hfill
\begin{minipage}[t]{0.43\textwidth}
\begin{center}
{\includegraphics[trim=1 1 1 1, clip=true, scale=1.11]{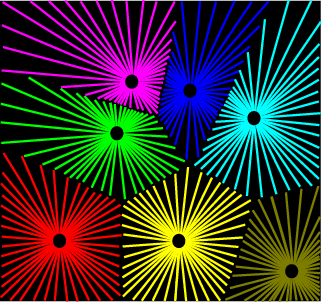}}
\end{center}
 \caption{Each of the cells of Figure  \bref{fig:Voronoi} is approximated using 44 rays.}
\label{fig:VoronoiRays}
\end{minipage}
\end{figure}

\subsection{Contribution of this work}
This paper presents and analyzes a new algorithm  which enables the combinatorial computation of 2D Euclidean Voronoi diagrams of point sites, where each cell is computed independently of the other ones. This algorithm, which is called ``the projector algorithm'', can in fact compute  even portions of the same cell independently of other portions. Parallel implementation is therefore naturally supported. The algorithm is significantly different from previous ones and some of the involved concepts in it are in the spirit of linear programming and optics. In contrast to many known algorithms in the theory of Voronoi diagrams, the sites given as an input to our algorithm do not need to be in a ``general position'' in order to avoid degenerate cases. A new combinatorial structure for representing the cells - and actually any convex polytope -  is described along the way, and the computation of the corresponding Delaunay graph (Delaunay triangulation) is obtained almost automatically. The time complexity of the algorithm, as a serial one (one processing unit) for the whole diagram, is $O(n^2)$. This upper bound on the time complexity is better than the one of the naive algorithm and it is not proved to be tight, i.e., to be $\Theta(n^2)$ (see Subsections  \bref{subsec:TimeComplexityComparison}--\bref{subsec:CommonScenarious} for further discussion on this issue). The actual behavior is in fact more or less linear when the sites are distributed uniformly (see also Theorem \bref{thm:CorrectnessOfTheAlg} below). It should be emphasized that this paper is theoretical. Issues related to implementation and experimental results are planned to be discussed elsewhere. 

\subsection{Paper layout} 
In Section  \bref{sec:Preliminaries} the notation, terminology and several tools are introduced. In Section  \bref{sec:Schematic} a schematic description of the algorithm is given. A detailed description of the algorithm is given in Section  \bref{sec:DetailedDescription}. A method for finding endpoints in an exact way is described in Section  \bref{sec:Endpoint}, and we also discuss briefly in this section an improvement of this method. The method for storing the data, as  well as a discussion on some combinatorial issues related to it, are described in Section   \bref{sec:CombinatorialInformation}. In Section  \bref{sec:Delaunay} we briefly discuss how the Delaunay graph can be extracted almost automatically from the stored data. In Section \bref{sec:HigherDimensions} we present by-products of possible independent interest, mainly a new combinatorial representation for the cells (in any dimension). The main theoretical result (Theorem \bref{thm:CorrectnessOfTheAlg}) is presented in Section \bref{sec:TimeComplextiyTheorem} and some ideas behind its proof are briefly discussed. We conclude the paper in Section  \bref{sec:ConcludeRemarks} with a  discussion on implementation issues and possible extensions of the algorithm and related issues to other settings. In order to increase the readability of the main body of the text, many relevant but rather technical issues were moved to three appendices: the first appendix (Section \bref{app:ImprovedEndpoint}) elaborates further on the methods presented in Section \bref{sec:Endpoint}, the second appendix (Section \bref{app:TheoryPractice}) discusses various  theoretical and practical issues related to Theorem \bref{thm:CorrectnessOfTheAlg}, and the third appendix (Section \bref{app:Proofs}) presents the proofs of Theorem \bref{thm:CorrectnessOfTheAlg} and related claims.

We note that large parts of the discussion below are rather detailed and full proofs are provided. We decided to do so in order to avoid the possibility of missing certain delicate points and also  in order to make the discussion as self-contained as possible. In addition, sometimes many details appear simply because several statements and methods  require a lot of case analysis. 

\section{Preliminaries}\label{sec:Preliminaries}
In this section we present the notation and basic definitions used later, as well as some helpful tools. Our world $X$ is a convex and compact polygon 
in the Euclidean plane $(\R^2,|\cdot|)$ with a nonempty interior, say a rectangle. Of course, $X$ is obtained from the intersection of finitely many half-planes. We  denote by $d(x,y)$ or $|x-y|$ the distance between the points $x\in \R^2$ and $y\in \R^2$. We denote by $[p,x]$ and $[p,x)$ the closed and half-open line segments connecting $p$ and $x$, respectively, i.e., the sets $\{p+t(x-p): t\in [0,1]\}$ and $\{p+t(x-p): t\in [0,1)\}$, respectively. The inner product between the points (vectors) $x=(x_1,x_2)$ and $y=(y_1,y_2)$ is $\langle x,y\rangle :=x_1y_1+x_2y_2$. A nonnegative linear wedge emanating from a point $p\in \R^2$ and generated by the vectors $v\in \R^2$ and $w\in\R^2$ is the set $\{p+\lambda v+\mu w: \lambda\geq 0, \mu\geq 0\}$. Of course, such a wedge is two-dimensional whenever there does not exists $\lambda\geq 0$ such that $v=\lambda w$ or $w=\lambda v$, namely whenever $v$  and $w$ are not located on the same ray emanating from the origin. We denote lines by $L$, $M$, etc. An edge (namely, a side) of a  convex polygon located on a corresponding  line $L$ is denoted by $\wt{L}$. In the context of Voronoi diagrams (Definition \bref{def:Voronoi} below and elsewhere), we refer to the inducing points as sites or generators, and denote them by $p_k$, $k\in K:=\{1,\ldots,n\}$, where $n\in\N$ is given as an input. These sites are assumed to be contained in the interior of $X$ and also to be distinct, namely $p_j\neq p_k$ for all $j,k\in K$ satisfying $j\neq k$. 

Here is the definition of the Voronoi diagram. This definition can easily be generalized to other settings, e.g., to spaces of higher dimension, to sites located anywhere in $X$ (possibly on the boundary of $X$), to the case of sites having a more general form than just point sites, to sites which have non-trivial intersection between themselves (possibly even identical sites), to various distance functions and so on, but in this paper we focus on the 2D Euclidean setting with distinct point sites located in the interior of the world $X$.  
\begin{defin}\label{def:Voronoi}
Given $n\in\N$ and a tuple of distinct point sites $(p_k)_{k=1}^n$ in our world $X$, the Voronoi diagram induced by these sites is the tuple $(R_k)_{k=1}^n$, where $R_k\subseteq X$ and for each $k\in K:=\{1,\ldots,n\}$,
\begin{equation*}
R_k:=\{x\in X: d(x,p_k)\leq d(x,p_j)\,\,\,\,\,\forall j\in K,\, j\neq k \}.
\end{equation*}
 In other words,  the Voronoi cell $R_k$ associated with the site $p_k$ is the set of all points 
 $x\in X$ whose distance to $p_k$ is not greater than their distance to the other sites $p_j$, $j\neq k$, $j,k\in K$.
\end{defin} 
The definition of the Voronoi diagram is analytic. However, it can be easily seen that each cell $R_k$ is 
the intersection of the world $X$ with halfplanes: the halfplanes  $\{x\in \R^2: d(x,p_k)\leq d(x,p_j\},\,j\neq k$, $j,k\in K$. 
Thus each cell is a closed and convex set 
which can be represented using its combinatorial structure, namely its vertices and edges (sides). 
Because of this property, the traditional approach to Voronoi diagrams is combinatorial.

In \cite{ReemISVD2009proc}, a different representation of the cells was introduced, suggesting 
to consider each of the cells as a union of rays (lines segments). An illustration of this representation is given in Figure \bref{fig:VoronoiRays} and is formulated mathematically in Theorem \bref{thm:domInterval} below. This representation is related 
to, but different from, the fact that the Voronoi cells are star-shaped. The theory of 
Voronoi diagrams in general and of algorithms for computing Voronoi diagrams in particular, 
is very diverse, with plenty of interesting and important facts and ideas. In particular, the star-shaped 
property of the cells is a well-known fact. However, to the best of our knowledge, and this is said after an extensive search that we have made in the literature (for many years) and after conversations with, or in front of, many experts, in various scientific and technological  fields, there has been no attempt to use any kind of 
ray-shooting techniques to compute (possibly approximately) the Voronoi cells. 
\begin{thm}\label{thm:domInterval}
The Voronoi cell $R_k$ of a site $p=p_k$ is a union of rays emanating from $p$ in 
various directions. More precisely, denote $A:=\bigcup_{j\neq k}\{p_j\}$. Given a unit vector $\theta$, let 
\begin{equation}\label{eq:Tdef}
T(p,\theta):=\sup\{t\in [0,\infty): p+t\theta\in X\,\,\mathrm{and}\,\ 
 d(p+t\theta,p)\leq d(p+t\theta,A)\},
\end{equation}
where $d(x,A):=\inf \{d(x,a): a\in A\}$. We refer to the point $p+T(p,\theta)\theta$ as the endpoint corresponding to the ray 
emanating from $p$ in the direction of $\theta$.  Then 
\begin{equation*}\label{eq:dom}
R_k=\bigcup_{|\theta|=1}[p,p+T(p,\theta)\theta].
\end{equation*}
(Note: in \cite{ReemISVD2009proc} a slightly different notation was used for the endpoint corresponding to the ray emanating from $p$ in the direction of $\theta$: instead of $p+T(p,\theta)\theta$, as in \beqref{eq:Tdef}, the notation was $p+T(\theta,p)\theta$.)
\end{thm}

 This representation actually holds (after simple modifications) in a more general setting 
 (any norm, any dimension, sites  of a general form, etc.) and it shows that 
by ``shooting'' enough rays in plurality of directions, one can obtain a fairly good approximation of the cells, as is illustrated in  Figures \bref{fig:Voronoi}--\bref{fig:VoronoiRays}. It will be shown in later sections how the idea of shooting rays can be used for obtaining the combinatorial structure of the cells.

\section{A schematic description of the projector algorithm}\label{sec:Schematic}
The projector algorithm is rather simple, at least from a high level point of view, as we show below (Method \bref{method:projectorEuclidean}). Of course, a detailed description of the algorithm (later in this section,  Sections \bref{sec:DetailedDescription}--\bref{sec:Endpoint}, Sections \bref{app:ImprovedEndpoint}--\bref{app:TheoryPractice}) is more involved, but one can say something in this spirit regarding many other algorithms as well.

Anyway, the method is based on the fact that the cell of some point site $p=p_k$ is a convex polygon whose boundary consists of vertices and edges. Some of the involved concepts in this algorithm are in the spirit of optics and linear programming. Recall again that for a unit vector $\theta$, the point $p+T(p,\theta)\theta$ is the endpoint corresponding to the ray emanating from $p$ in the direction of $\theta$ (see Figure  \bref{fig:VoronoiRays}). Recall also that we assume that all the sites are 
distinct points which are given as an input, and no site is located on the boundary of the world $X$. 

\begin{method} {\bf \it The Projector Algorithm: a high level description} \label{method:projectorEuclidean}
\begin{itemize}
\item {\bf  Input: } A site $p$;
\item {\bf  Output: } The (combinatorial) Voronoi cell of $p$; 
\end{itemize}
\begin{enumerate}
\item Think of $p$ as being a light source;
\item emanate a (linear wedge-like) beam of light from $p$ using a projector;
\item \label{item:StepVertex} detect iteratively (by possibly dividing the wedge into subwedges) 
all possible vertices (and additional related combinatorial information) inside this 
beam using corresponding endpoints and an associated system of equations; 
\item continue the process with other beams until the entire world around $p$ is covered; 
\end{enumerate}
\end{method}

\begin{figure}[t]
\begin{minipage}[t]{1\textwidth}
\begin{center}
{\includegraphics[scale=0.7]{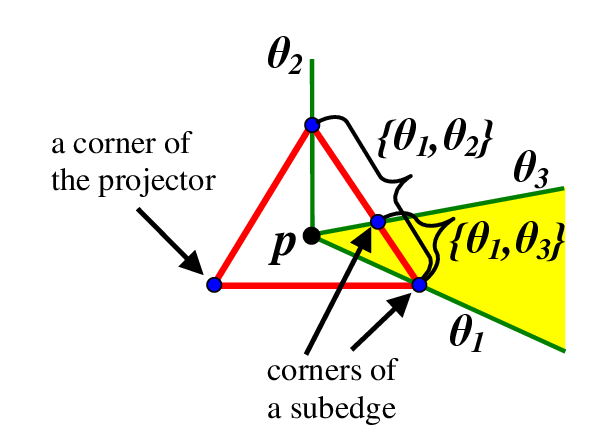}}
\end{center}
 \caption{The projector, some of its subedges, and the beam emanating from $p$ and corresponding to 
 the subfedge $\{\theta_1,\theta_3\}$.}
\label{fig:SubedgeSubwedge}
\end{minipage}
\end{figure}

The actual generation and handling of the (sub)wedges is done using a projector (triangle) 
located around $p$. See Figure \bref{fig:SubedgeSubwedge}. The boundary of this 
projector is initially composed of edges, and later these edges are composed of subedges,  
when we narrow the search to subwedges (sub-beams). 
Each such a subedge induces a wedge:  the wedge generated 
by the rays which pass via the corners of the subedge. 
Each such a corner induces a unit vector (denoted by $\theta$) which points in 
its direction (from $p$). Once the corresponding unit vector is known, then so is 
the ray in its direction. We use a simple data structure 
called $EdgeQueue$ which maintains a dynamic list of subedges so that the 
whole implementation is based on loops instead of recursive programming.
Initially there are three subedges (the ones of the projector) and later subedges 
are added or omitted. We represent a subwedge using its corners (vertices); for instance, in Figure \bref{fig:SubedgeSubwedge} the right side of the projector is represented by $\{\theta_1,\theta_2\}$, and later it is subdivided into the subwedges $\{\theta_1,\theta_3\}$ and $\{\theta_3,\theta_2\}$.

The system of equations mentioned above (Step \beqref{item:StepVertex}) is 
\begin{equation}\label{eq:B_lambda}
B \lambda = H.
\end{equation}
Here the vector of unknowns is $\lambda=(\lambda_1,\lambda_2)$;   
$B$ is the $2$ by $2$ matrix with entries  $B_{ij}=\langle N_i,T_j\rangle$ and       
 $H$ is a 2D vector with entries $H_i=\langle N_i,T_i\rangle$,  $i,j \in \{1,2\}$; here 
 $T_i:=T(p,\theta_i)\theta_i$, namely, $T_i$ is the vector in the direction of $\theta_i$ whose length is  the distance between $p$ and the endpoint $p+T_i$, $i\in \{1,2\}$; the $2$-dimensional vector $N_i$ is a normal to the 
line  $L_i:=\{x\in\R^2: \,\langle N_i,x\rangle=\langle N_i,p+T_i\rangle\}$ on which 
the endpoint $p+T_i$ is located. Of course, solving the linear 2 by 2 system of 
equations \beqref{eq:B_lambda} is a simple task, either 
in an exact way (exact arithmetic) or using floating point arithmetic. 

Equation \beqref{eq:B_lambda} has a simple geometric meaning: the point $u:=p+\sum_{i=1}^2\lambda_i T_i$ is in the intersection of the lines  $L_1$ and $L_2$ if and only if $\lambda$ solves \beqref{eq:B_lambda}. If we want to  restrict ourselves to the wedge generated by the corresponding rays, then we consider only the nonnegative solutions of \beqref{eq:B_lambda}, i.e., $\lambda_i\geq 0$ for $i=1, 2$. If equation  \beqref{eq:B_lambda} has a unique nonnegative solution $\lambda$, then this means that $u$ is a point in the wedge which is a candidate  to be a vertex of  the cell, since it may be (but is not necessary) in the intersection of the corresponding two different edges located on the lines $L_i$. If, in addition, $u$ is known to be in the cell, then it is indeed a vertex.

\begin{remark}
To the best of our knowledge, Method  \bref{method:projectorEuclidean} above, as well as Algorithm 1 below, are new. There are many differences between 
it and existing algorithms, e.g., its ability to compute each cell or parts of a cell independently of the other cells, its use of rays and wedges, the fact that it 
first detects edges of a cell and later vertices, etc. 

In this connection, we want to say a few words regarding a technique called  ``geometric  probing'', which, although it does not  consider Voronoi diagrams, it still uses rays in order to detect boundary of geometric objects  
 (see, e.g.,  \cite{ABY1987,ColeYap1987}). At first glance one may think that the papers which discuss this technique do have some relation to our method. However, it can be verified quickly  that the settings and methods described 
there are significantly different from the method described in our paper. Moreover, the actual relation of these papers to our paper is quite weak and besides, our paper has not been inspired by them (we have become aware of them years after developing the ideas described in our paper).

For instance, the probing done in \cite[Fig. 1, page 162]{ABY1987} 
is performed by a robot, which goes, from the outside, 
around the boundary of the given 2D object, and use its arms (or an optical 
device) in order to touch (probe) the boundary of the object. By this way it 
obtains a collection of points from the boundary and denotes this collection by $P$. 
Denote by $L$ the collection of rays emanating from the 
various locations of the robot when it goes around the object and 
probes the geometric object. 
Then, as is written in \cite[p. 163]{ABY1987}: ``the aim is to join the points of $P$ without intersecting 
the rays of $L$, in order to find a polygonal approximation of the object boundary''. 
Similar things can be said regarding \cite{ColeYap1987}. 
In comparison, in our method all the rays emanate from a unique point (the site), 
this point is located inside the geometric object (the Voronoi cell), 
the computation of the endpoints is not trivial (the boundary of the 
object is not known, in contrast to the geometric probing case), and the 
goal is not to connect the endpoints in order to obtain a 
polygonal approximation of the boundary of the Voronoi cell, but rather 
to use these endpoints in order to detect edges and later vertices of the cell. 
\end{remark}

\section{A detailed description of the projector algorithm}\label{sec:DetailedDescription}

\begin{figure}[t]
\begin{minipage}[t]{0.48\textwidth}
\begin{center}
{\includegraphics[scale=0.55]{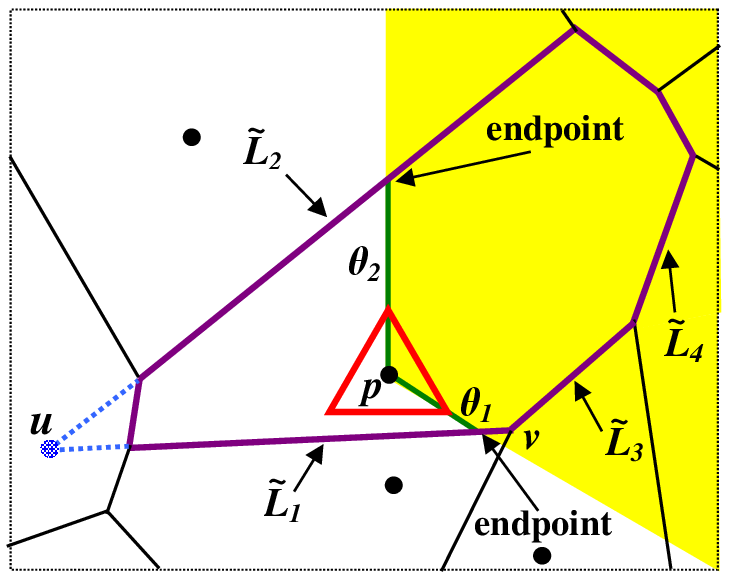}}
\end{center}
 \caption{Illustration of the algorithm. The wedge generated by  the subedge $\{\theta_1,\theta_2\}$ is shown.  The intersection between $L_1$ and $L_2$ (the lines on which the edges $\tilde{L}_1$ and $\tilde{L}_2$ are located) is the point $u$ located outside the wedge; hence the wedge is divided. The next two subedges are $\{\theta_1,\theta_3\},\,\{\theta_2,\theta_3\}$. }
\label{fig:projectorEuclidWedgePhase1}
\end{minipage}
\hfill
\begin{minipage}[t]{0.48\textwidth}
\begin{center}
{\includegraphics[scale=0.54]{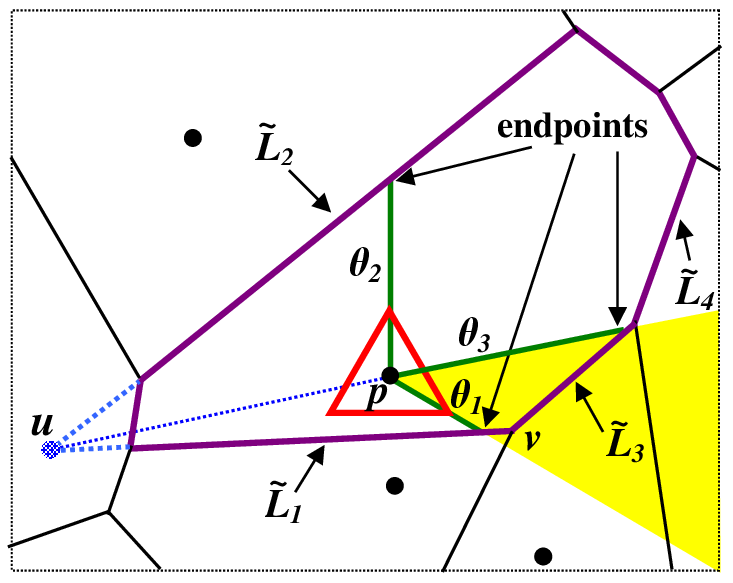}}
\end{center}
 \caption{Now the wedge generated by the subedge $\{\theta_1,\theta_3\}$ is shown. 
 Since $v=\tilde{L}_1\cap \tilde{L}_3$ and $v$ is located in the wedge and in the cell, $v$ is a vertex 
 and we do not need to further divide this wedge.  }
\label{fig:projectorEuclidWedgePhase2}
\end{minipage}

\end{figure}

This section presents a more detailed description of the projector algorithm, namely of Method \bref{method:projectorEuclidean}. A pseudocode is included too. 
See Figures \bref{fig:projectorEuclidWedgePhase1}--\bref{fig:projectorEuclidWedgePhase2} for an illustration. Additional relevant illustration is given in Figure \bref{fig:AlgorithmTree}. 
The description below and the one given in Algorithm 1 use the same notation. 

\footnotesize
\begin{algorithm}\label{alg:projectorEuclidean}

\caption{The Projector Algorithm: a detailed pseudocode}

\SetKwData{Left}{left}\SetKwData{This}{this}\SetKwData{Up}{up}
\SetKwFunction{Union}{Union}\SetKwFunction{FindCompress}{FindCompress}
\SetKwInOut{Input}{input}\SetKwInOut{Output}{output}
\Input{A site $p$ whose Voronoi cell is to be computed, the other sites}
\Output{The vertices and edges of the cell, other information}
\BlankLine 

Create the projector unit vectors\;
Create the projector edges and enter them into $EdgeQueue$\;

\While {EdgeQueue is nonempty}{
Consider the highest (first) subedge in $EdgeQueue$\; 
Denote it by $\{\theta_1,\theta_2\}$\; 
Compute the endpoints $p+T_i$\label{line:Endpoint} (see Methods \bref{method:Endpoint} and \bref{method:EndpointImproved}), where $T_i:=T(p,\theta_i)\theta_i$, $i=1,2$\;   
Find their neighbor sites $a_i,\, i=1,2$ (see Methods \bref{method:Endpoint} and  \bref{method:EndpointImproved})\; 
Compute the bisector line $L_i$ between $p$ and $a_i,\, i=1,2$ (see Methods \bref{method:Endpoint} and \bref{method:EndpointImproved})\;
If no such a site $a_i$ exists, then $p+T_i$ is on the boundary of the world. Call the corresponding boundary line  $L_i$\;

Consider the system of equations \beqref{eq:B_lambda} \\

\If(\tcp*[f]{no solution or infinitely many}){$\det({B})=0$} 
{
\If (\tcp*[h]{no vertices here}\label{line:SameLines}) {$L_1=L_2$} {continue$;$}
\Else (\tcp*[h]{parallel lines}\label{line:ParallelLines})
{Let $\theta_3:=\phi$, where $\phi$ is the direction vector of the lines\;
 If the ray from $p$ in the direction of $\phi$ is not in the wedge, then let $\theta_3:=-\phi$\;
 Insert the subedges $\{\theta_1, \theta_3\},\, 
\{\theta_2, \theta_3\}$ into $EdgeQueue$\;\label{line:DivideParallelLines}
}
}

\Else (\tcp*[h]{$\det(B)\neq 0$, unique solution $\lambda=(\lambda_1,\lambda_2)$})
{Let $u:=p+\lambda_1 T_1+\lambda_2 T_2$\;
 \If (\tcp*[h]{we're in the wedge}) {$\lambda$ is nonnegative} 
 {\If {$u$ is inside the cell\label{line:Vertex}}
  { Store $u$, $L_1,L_2$ (and/or neighbor sites); \label{line:Store}\tcp*[h]{$u$ is a vertex}
  }
  \Else (\tcp*[h]{$u$ is outside the cell}\label{line:OutsideCell})
  { Let $\theta_3:=(u-p)/|u-p|$\;\label{line:Theta3OutsideCell}
   Insert $\{\theta_1, \theta_3\},\,\{\theta_2, \theta_3\}$ 
   into $EdgeQueue$\; \label{line:DivideOutsideCell}       
  }
 }
 \Else (\tcp*[h]{$u$ isn't in the wedge}\label{line:OutsideWedge})
 {Let $\theta_3:=(p-u)/|p-u|$\; \label{line:Theta3OutsideWedge} 
  Insert $\{\theta_1, \theta_3\},\,\{\theta_2, \theta_3\}$ 
  into $EdgeQueue$\;\label{line:DivideOutsideWedge}
 }
}
Remove $\{\theta_1,\theta_2\}$ from $EdgeQueue$\;
}

\end{algorithm}
\normalsize


First, we create the three unit vectors $\theta_i$, $i\in \{1,2,3\}$ corresponding to a projector  around the point $p$. 
After choosing a projector subedge, shooting the two rays in the direction of 
$\theta_i,\,i=1,2$, finding the endpoints $p+T_i$  (using, e.g., Method \bref{method:Endpoint} 
in Section \bref{sec:Endpoint}; here $T_i:=T(p,\theta_i)\theta_i$) and finding the corresponding 
bisecting lines $L_i$, we want to use this information for finding all of the possible vertices 
in the wedge generated by the rays. By using equation \beqref{eq:B_lambda} we find the type of 
intersection between the lines $L_1,L_2$. This intersection is either the empty set, a point, or a line.

If \beqref{eq:B_lambda}  has no solution of any kind (including solutions which are not non-negative), then 
the lines $L_1$ and $L_2$ are parallel. This is a rare event but it must be taken into account.  
In this case $L_1$ and $L_2$ have the same direction vector $\phi$, 
i.e., $L_i=\{q_i+\phi t: t\in \R\}$ for some $q_i\in \R^2$, $i=1,2$ and some unit vector $\phi$. 
We check if the ray emanating from $p$  in the direction of $\phi$ is inside the wedge 
(this happens if and only if the solution $(\alpha_1,\alpha_2)$ to the linear equation 
$\phi=\alpha_1\theta_1+\alpha_2\theta_2$ is nonnegative). If yes, then we denote $\theta_3:=\phi$ and shoot a ray in the direction of $\theta_3$. Otherwise, we denote $\theta_3:=-\phi$ and shoot the ray in the direction of $\theta_3$. In both cases this ray will be contained in the wedge and will hit an edge of the cell 
not located on the lines $L_1$ and $L_2$ (the edge may be located on the boundary of the bounded world $X$). 
We divide the current projector subedge using $\theta_3$ and 
continue the process.  

If \beqref{eq:B_lambda}  has infinitely many solutions, then $L_1=L_2$ (and vice versa). Hence both endpoints are located on the same line. In this case there is no vertex  in the corresponding wedge (perhaps one of the endpoints $p+T_i$ is a vertex,  but this vertex will be found later using the neighbor subedge: see Lemma  \bref{lem:FoundVerticesEndpoint}). Hence we can finish with the current subedge and go to the other ones. Such a case is implicit in Figure  \bref{fig:projectorEuclidWedgePhase1} when the rays are shot in the directions of the first and third corners of the projector and hit $L_1$. 

If \beqref{eq:B_lambda}  has a unique solution  $\lambda=(\lambda_1,\lambda_2)$, then either it is not nonnegative, that is, the point $u:=p+\sum_{i=1}^2\lambda_i T_i$ is not in the wedge, or $\lambda$ is nonnegative, i.e., $u$ is in the wedge. In the first case the ray emanating from $p$ in the direction of $\theta_3:=(p-u)/|p-u|$ will hit an edge of the cell contained in the wedge but not located on $L_1$ or $L_2$.  Such a case is described  in Figure \bref{fig:projectorEuclidWedgePhase1} when considering the subedge $\{\theta_1,\theta_2\}$.  We divide the current projector subedge using $\theta_3$ and continue the process.  In the second case $u$ is in the wedge, but we should check whether $u$ is in the cell (can be checked, for instance, by distance comparisons). If $u$ is in the cell (a case corresponding to the case of the subedge $\{\theta_1,\theta_3\}$ in Figure \bref{fig:projectorEuclidWedgePhase2}, where $u=v$ there), then it is a vertex and we store it (together with other data: see Section \bref{sec:CombinatorialInformation}). We have finished with the subedge and can go to the other ones. Otherwise $u$ is not a vertex, and we actually found a new edge of the cell corresponding to the ray in the direction of $\theta_3:=(u-p)/|u-p|$ (this case is implicit in Figure  \bref{fig:projectorEuclidWedgePhase1} when one  considers the subedge $\{\hat{\theta}_2,\hat{\theta}_3\}$, i.e., the subedge induced by the rays which are shot in the directions of the second and third corners of the projector, respectively; in this case $u$ is the intersection of $L_1$ and $L_2$). We divide the subedge using $\theta_3$ and continue the process. We also note that $u\neq p$ (and hence $\theta_3$ is well defined)  since if, to the contrary, $u=p$, then, from the fact that $u$ is located on the lines $L_i$, $i\in\{1,2\}$ (see the discussion after \beqref{eq:B_lambda}),   it follows that $p$ is located on these lines too. This  is impossible since $L_1$ (and also $L_2$) is either the bisector between $p$ and some other site $a_1\neq p$,  and in this case obviously $p$ cannot be located on it, or $L_1$ is a line on which part of the boundary of $X$ is located and thus again, $p$ cannot be located on $L_1$ (since $p$ is in the interior of $X$).

It is possible to avoid a recursive implementation of the algorithm and base it on loops using a simple data structure called  $EdgeQueue$. This is a list which stores temporarily the projector subedges that are handled during the process (each subedge is represented by a set of  two unit vectors, which correspond to its corners). The algorithm runs until $EdgeQueue$ is empty. The initial list contains the sides (edges)  of the projector $\{\psi_1,\psi_2\},\,\{\psi_2,\psi_3\},\,\{\psi_1,\psi_3\}$, where we can take $\psi_1:=(\sqrt{3}/2,-1/2)$, $\psi_2:=(0,1)$ and $\psi_3:=(-\sqrt{3}/2,-1/2)$. \label{page:phi_i}
In this connection, it should be emphasized that the projector is used for handling the progress of the algorithm (using $EdgeQueue$), but the corresponding unit vectors in the direction of the corners of the subedges are not necessarily on the same line as the one on which the projector subedge is located. Despite this, it is convenient to represent a projector subedge by its associated unit vectors. 

\section{Finding the endpoints exactly}\label{sec:Endpoint}
In order to apply Method \bref{method:projectorEuclidean}, we should be able to find the endpoint $p+T(p,\theta_i)\theta_i$ emanating from the site $p$ in the direction of $\theta_i$ (see \beqref{eq:Tdef} and line \bref{line:Endpoint} in Algorithm 1). One possible method is to use the method described in \cite{ReemISVD2009proc}, but the problem is that the endpoint found by this method is  given up to some  error parameter, 
and unless this parameter is very small (which, in this specific case, implies slower computations), this may cause an accumulating error later when finding the vertices, due to numerical errors in the expressions in \beqref{eq:B_lambda}. 

In what follows we will describe a new method for finding the endpoint in a given  direction $\theta$ exactly. Of course, when using floating point arithmetic errors appear, but they are much smaller than the ones described above. See Figure \bref{fig:projectorEuclideanRay} for an illustration. After discussing this new method in Subsection \bref{subsec:Endpoint}, we discuss briefly an improvement of it in Subsection \bref{subsec:UniformDistribution}. Full details related to both methods, as well as many technical aspects, can be found in the first appendix (Section \bref{app:ImprovedEndpoint}). 

\subsection{The first method}\label{subsec:Endpoint}
\begin{method}$\,$\label{method:Endpoint}
\begin{itemize}
\item {\bf  Input: } A site $p$ and a unit vector $\theta$;
\item {\bf  Output: } the endpoint $p+T(p,\theta)\theta$. 
\end{itemize}
\begin{enumerate}
\item	Shoot a ray from $p$ in the direction of $\theta$ and stop it at a point $y$ which is either in the region $X$ but outside the cell of $p$ (see Method \bref{method:EndpointImproved} below), or it is the intersection of the ray with the boundary of the region (see Remark  \bref{rem:tCompute} below). If $y$ is chosen to be outside the cell, then go to Step \beqref{item:CloseNeighbor}. Otherwise, let  $L$ be the boundary line of $X$ on which $y$ is located; 
\item\label{item:InCell}	check whether $y$ is in the cell, e.g., by comparing $d(y,p)$ to $d(y,a)$ for any 
other site $a$, possibly with enhancements which allow to reduce the number of distance comparisons;
\item\label{item:Output}	if $y$ is in the cell, then $y$ is the endpoint (and $L$ is a bisector line, unless $L$ is a boundary line); the calculation along the ray is complete; 
\item\label{item:CloseNeighbor}	otherwise, $d(y,a)<d(y,p)$ for some site $a$. Let $CloseNeighbor:=a$;
\item\label{item:u}	find the point of intersection (call it $w$) between the given  ray and the bisector  line $L$ between $p$ and $CloseNeighbor$.  This intersection is always nonempty.  The line $L$ is easily found because it is vertical to the vector      $p-CloseNeighbor$ and passes via the point $(p+CloseNeighbor)/2$;  
\item	let $y:=w$; go to Step \beqref{item:InCell}.
\end{enumerate}
\end{method}

\begin{figure}[t]
\begin{minipage}[t]{1\textwidth}
\begin{center}
{\includegraphics[scale=0.55]{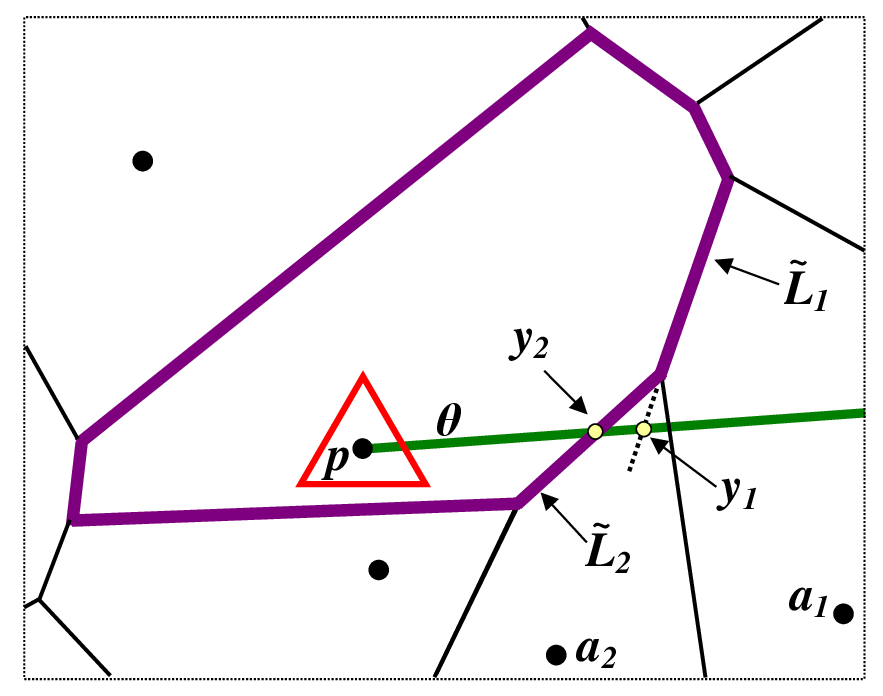}}
\end{center}
 \caption{Illustration of Method \bref{method:Endpoint} for some ray.  
 At the first displayed iteration  $CloseNeighbor$ is $a_1$. The intersection between $L_1$ (the line on which the edge $\tilde{L}_1$ is located) and the ray is the point $y_1$. At the next stage $CloseNeighbor$ is $a_2$ and hence $y=y_2$. The process terminates since $y_2$ is in the cell  of $p$, i.e., it is the endpoint.}
\label{fig:projectorEuclideanRay}
\end{minipage}
\end{figure}

 \subsection{An improvement of  Method \bref{method:Endpoint} for the case of certain distributions of sites}\label{subsec:UniformDistribution}

Now we discuss briefly an improvement of Method \bref{method:Endpoint} for computing the endpoint of a given ray. Full details of this method can be found in the first appendix (Method \bref{method:EndpointImproved} in Section \bref{app:ImprovedEndpoint}). This method is general  in the sense that it can be used for any type of distribution of the sites. However, it is most efficient for the case where the sites are distributed independently according to the  uniform distribution, i.e., the probability of each site to be in some region is proportional to the area of the region, and there is no dependence between the sites (when viewed as random vectors). For the sake of simplicity, we also assume that $X$ is a rectangle whose side lengths are integer multiplications of some real number $s>0$. 

A rough description of the method is as follows. We first perform a preprocessing stage,  whose time complexity is $O(n)$, in which the sites are inserted into  a corresponding data structure (buckets) mentioned in \cite{BentleyWeideYao1980}. This data structure is nothing but a collection of buckets which decompose $X$: each bucket is just a square having sides of length $s$. Then, when we compute the Voronoi cell of some site $p$, the buckets allow us to restrict the distance comparisons to sites located in buckets which are close to the temporary endpoint $y$. A more precise meaning of the word ``close'' appears in Lemma \bref{lem:SitesInLargeSquare} below, and this lemma ensures that each site $a$ which is  not close to $y$ will automatically satisfy $d(y,a)>d(y,p)$; as a result, there is no need to perform the distance comparison stage (in which we compare $d(y,a)$ to $d(y,p)$) when we try to see whether $y$ is in the cell of $p$. The buckets also help us to find a candidate for the first temporary endpoint which is usually close to $p$, since unless some square around $y$ is empty of sites (low probability: see Remark \bref{rem:MethodEndpointImproved}\beqref{item:2beta_s} below), Lemma \bref{lem:SitesInBetaSquare} below ensures that farther sites $a$ will automatically satisfy $d(y,a)>d(y,p)$. The rest of the method is roughly as Method \bref{method:Endpoint}.

\section{Storing the data}\label{sec:CombinatorialInformation} 

Given a point site $p=p_k$, when a vertex $u$ belonging to the cell of $p$ is found, one stores  the following parameters: its coordinates, the lines from which it was obtained (namely, $u$ belongs to the corresponding cell's edges located on  these lines), and the index $k$. For storing a line $L$ it is convenient to store the index of  its  associated neighbor site, 
namely the index (simply a number or a label) of the site which induces it (denoted by $CloseNeighbor$ in Method \bref{method:Endpoint}). If it is a boundary line, then it has a unique index number which is stored 
and from this index one can retrieve the parameters (the normal and the constant) defining the line.  Alternatively, these parameters can be stored directly. For some purposes it may be useful to store also some endpoints. 

A convenient data structure for the storage is a one dimensional array, indexed by $k$, in which the vertices (represented, as explained above,  by  coordinates and associated neighbor sites) and any additional information, such as endpoints,  are stored. Although the vertices are not stored 
according to a certain order, the method of search (Algorithm 1) implies that it is  quite easy to sort them later in clockwise or counterclockwise order, e.g., by labeling the generated rays with suitable values.

\section {Computing the Delaunay graph (Delaunay triangulation)}\label{sec:Delaunay}

As is widely known, the Delaunay graph is an important geometric structure  which is closely related  to the Voronoi diagram and by itself has many applications \cite{Aurenhammer, ChengDeyShewchuk2012, Edelsbrunner-book-1987, OBSC}. By definition, it consists of vertices and edges. The vertices  are the  sites. There is an edge between two sites if their Voronoi cells are neighbors (via an edge). See Figure  \bref{fig:DelaunayVoronoi} for an illustration. 

The computation of the restriction of the Delaunay graph to the given bounded world $X$, from the stored data structure of the Voronoi diagram, is simple and done almost automatically: one chooses a given site, goes over the  data structure and finds all the different neighbor sites of a given site. The procedure is repeated for each site and can easily be  implemented in parallel.

Note that here everything is restricted to the given bounded world $X$. Rarely it may happen that two sites whose cells are neighbors 
in the whole plane are not neighbors in $X$. This can happen only with cells which intersect the boundary of $X$. For overcoming this problem (if one considers this as a problem), one can simply take $X$ to be large enough or one can perform a separate check for the above-mentioned boundary cells.

\begin{figure}
\begin{minipage}[t]{1\textwidth}
\begin{center}
{\includegraphics[scale=0.65]{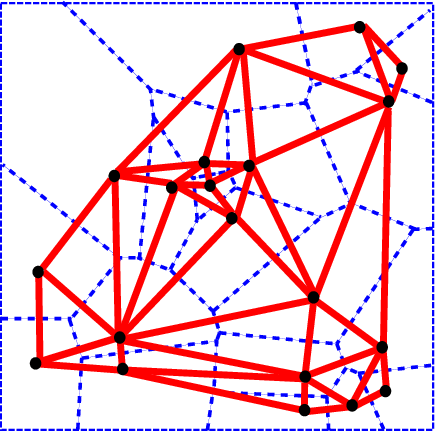}}
\end{center}
 \caption{The Delaunay graph of 20 sites, restricted to a square (thick lines),  
 together with the Voronoi cells of the sites (thin lines).}
\label{fig:DelaunayVoronoi}
\end{minipage}
\end{figure}

\section{A by-product: a new combinatorial representation for the cells in higher dimensions}\label{sec:HigherDimensions}
As explained in Section \bref{sec:CombinatorialInformation}, 
each vertex is stored by saving its coordinates and the edges (actually the corresponding neighbor sites which induce these edges) which intersect at the vertex. This section extends this idea to higher dimensions and presents a new  combinatorial representation for the Voronoi cells (and 
actually any convex polytope) in this setting. 

In dimension $m=2$, a vertex $u$ is always obtained from the intersection of two lines. In dimension $m\geq 3$ a vertex is usually obtained from exactly  $m$ hyperplanes, but in principle it can be obtained from $S$ hyperplanes, $S>m$. We call the set  $\{L_{i_1},\ldots, L_{i_{S}}\}$ of all the hyperplanes from which $u$ was obtained ``the combinatorial representation of $u$''.

As the examples below show (see also Section \bref{sec:Delaunay}), once the above-mentioned combinatorial representation is known and stored,  we can obtain other combinatorial information related to the cell, say  the neighbors of a given vertex, combinatorial information related to the $\ell$-dimensional faces of the cell, $\ell=0,1,\ldots,m-1$, and so on, and hence we do not need to store these types of information separately. 
This is in contrast to familiar methods for representing the combinatorial information in which one has to find and store all the $\ell$-dimensional faces, $\ell=0,1,\ldots,m-1$ of the cell. Since the combinatorial complexity of the cell (the number of multi-dimensional faces) can grow exponentially  with the dimension \cite{Aurenhammer, Klee}, our method may reduce the required space needed for the storage (at least by improving some constants). The price  is, however, that for retrieving some data, certain search operations will have to be done on the stored information. There is another difference between our method and other ones. In other methods one starts with the vertices as the initial (atomic) ingredients, and from them one constructs higher dimensional faces. For instance, an edge is represented by the two vertices which form its corners. However, in our method we start with the highest dimensional faces (located on hyperplanes) and from them we construct the vertices and the other multi-dimensional faces. As a matter of fact, our method of storage can be used to represent any multi-dimensional convex polytope, and even a class of nonconvex or abstract polytopes.

\begin{expl}
Information regarding all the  various $\ell$-dimensional faces of the cell, $\ell=0,1,\ldots,m-1$, can be found by observing that any such a face is located on the intersection of  hyperplanes. Indeed,  the $(m-1)$-dimensional faces are located on the hyperplanes $L_i$ which appear in the representation of the vertices. Thus, in order to find all the vertices which span a given $(m-1)$-dimensional face, say a one located on the hyperplane $L_{i_1}$, we simply need to search for all the vertices having $L_{i_1}$ in ther combinatorial representation. Consider now the $(m-2)$-dimensional faces; we fix an hyperplane $L_i$ and look at all the vertices $u$ having $L_i$ in their combinatorial representation; in the representation of 
any such a vertex $u$, there are other hyperplanes $L_j$, and an $(m-2)$-dimensional 
face is located on $L_i\cap L_j$. By going over all the possible hyperplanes of $u$, all the possible $u$, and all the possible $L_i$, we can find and represent all the possible linear subspaces on which the $(m-2)$-dimensional faces are located (note: $L_i\cap L_j$ can be represented by an array containing the parameters defining both $L_i$ and $L_j$). We can also find all the vertices which span a given $(m-2)$-dimensional face: such a face is located on the intersection of, say, the hyperplane $L_{i_1}$ and the hyperplane $L_{i_2}$; hence we simply need to find all the vertices having $L_{i_1}$ and $L_{i_2}$ in their combinatorial representation. Similar things can be said regarding the other $\ell$-dimensional faces, $\ell=0,1,\ldots,m-1$.
\end{expl}

\begin{expl} Given a vertex  $u$ with a given combinatorial representation, its neighbor 
vertices can be found by observing that if $u$ and $v$ are neighbors, then they are located 
on the same 1-dimensional face. This face is located on the linear subspace which is the intersection of $m-1$ hyperplanes. Hence, for finding the neighbor vertices of $u$ we simply need to go over the list of vertices $v$ and choose the ones whose combinatorial representation contains $m-1$ hyperplanes which also appear in the representation of $u$.
\end{expl} 

\section{A time complexity theorem}\label{sec:TimeComplextiyTheorem} 
This section presents the main theoretical result of the paper. Here we only present the theorem and discuss briefly the ideas behind its proof. Many aspects and clarifications related to the theorem can be found in the second appendix (Section \bref{app:TheoryPractice}). Its (long and technical) proof can be found in the third appendix (Section \bref{app:Proofs}).

\begin{thm}\label{thm:CorrectnessOfTheAlg}
Suppose that the world $X\subset \R^2$ is a compact and convex subset whose boundary is polygonal. Assume also that the $n\in\N$ distinct point sites $p_1,\ldots,p_n$ are contained in its interior. Then:
\begin{enumerate}[(a)]
\item Algorithm 1 is correct. More specifically, given any site $p=p_k$ for some $k\in K=\{1,\ldots,n\}$, the computation of the Voronoi cell of $p$  by Algorithm 1 terminates  after a finite number of steps and  the corresponding entries in  the output of the algorithm include all the vertices and edges of the cell;
\item the time complexity, for computing the cell of $p_k$, is $O(r_k e_k)$, where $r_k$ is the maximum number of distance comparisons done along each shot ray (compared between all the rays shot during the computation of the cell of $p_k$), and $e_k$ is the number of edges of the cell;
\item the time complexity, for the whole diagram, assuming one processing unit is involved, 
is $O(n^2)$ (this upper bound is not necessarily tight);
\item the time complexity, for the whole diagram, assuming $Q\in\N$ processing units are involved (independently) and processor $Q_i$ computes a set $A_i$ of cells, is 
\begin{equation*}
\max\left\{\sum_{k\in A_i}O(r_k e_k): i\in \{1,\ldots, Q\}\right\}.
\end{equation*}
\item\label{item:TimeComplexityUniformDistribution} 
Let $\epsilon\in (0,1)$ be arbitrary. Suppose that $n\in\N$ is large enough so that both $n>100$ and \beqref{eq:ln(n/epsilon)}--\beqref{eq:alpha2m} below hold (in what follows $\lfloor t\rfloor$ is the floor function applied to the real number $t$, that is, the greatest integer not exceeding $t$; in addition, $e$ is the base of the natural logarithm, namely $e\approx 2.71828$):  
\begin{equation}\label{eq:ln(n/epsilon)}
\alpha_1<\frac{1}{3}\left(\lfloor 0.2m\rfloor+2\right), 
\end{equation}
\begin{equation}\label{eq:alpha2m}
\alpha_2<e^{-2}(\lfloor 0.2m\rfloor+1)
\end{equation} 
where  
\begin{equation}\label{eq:m}
m:=(\lfloor\sqrt{n}\rfloor)^2,
\end{equation}
\begin{equation}\label{eq:alpha1}
\alpha_1:=\left(\left\lfloor\sqrt{\ln\left(\frac{n(1+2e^{-400})}{\epsilon}\right)}\right\rfloor+\textnormal{OneTwo}\right)^2, 
\end{equation}
\begin{equation}\label{eq:OneTwo}
\textnormal{OneTwo}=\left\{\begin{array}{lll}
1 & \textnormal{if}\,\,\left\lfloor\sqrt{\ln\left(\displaystyle{\frac{n(1+2e^{-400})}{\epsilon}}\right)}\right\rfloor\,\,\textnormal{is even},\\
2 & \textnormal{otherwise}
\end{array}
\right.
\end{equation}
\begin{equation}\label{eq:alpha2}
\alpha_2:=\left(1+\left\lfloor 8\sqrt{2}(\beta+1.01)\right\rfloor\right)^2. 
\end{equation}
\begin{equation}\label{eq:beta}
\beta:=\frac{\sqrt{\alpha_1}-1}{2}, 
\end{equation}
Let $m_1$ and $m_2$ be any natural numbers satisfying 
\begin{equation}\label{eq:m1m2}
m_1m_2=m,\quad \min\{m_1,m_2\}>2\sqrt{\alpha_2},
\end{equation}
e.g., $m_1:=m_2:=\sqrt{m}$.  Suppose further that the world $X$ is a rectangle of the form $X=[0,m_1s]\times [0,m_2s]$ for some positive number $s$, and assume that there are $n$ sites in $X$ which are distributed independently of each other and according  to the uniform distribution. Assume that a preprocessing  stage is done in which $X$ is decomposed into $m_1\times m_2$ buckets (that is, Method  \bref{method:Preprocessing} which is described in Section \bref{app:ImprovedEndpoint}). Then with probability which is at least $1-\epsilon$, for all $k\in K$ and all rays shot during the computation of the cell of $p_k$ using  Method \bref{method:EndpointImproved} (Section \bref{app:ImprovedEndpoint}), one has $r_k\leq 3\alpha_2$. Moreover, with probability which is at least $1-\epsilon$ , the total number of distance comparisons done using Algorithm 1 (with one processing unit) for the computation of all of the Voronoi cells, is at most 
\begin{equation*}
(3\alpha_2-1)\cdot 22n=O\left(n\log\left(\frac{n}{\epsilon}\right)\right).
\end{equation*}
\end{enumerate}
\end{thm}

The proof of Theorem \bref{thm:CorrectnessOfTheAlg} is quite technical, but the main idea regarding the bound on the time complexity is simple. Indeed, suppose that we consider the cell of the site $p_k$. Each time a ray is shot during the running of Algorithm 1, there can be three cases: either a new edge of the cell is detected, or a vertex is found, or the subwedge induced by the shot ray and another ray does not contain vertices. This shows (after a careful counting) that the number of rays used for each cell is bounded by a universal constant times the number of edges in the cell. The operations done along a given ray for detecting its endpoint are mainly distance comparisons and some $O(1)$ operations such as simple arithmetic and operations related to the memory. The maximum number of the distance comparisons, compared between all the shot rays, is $r_k$, and the upper bound follows.  

As for the upper bound regarding the time complexity for the whole diagram, one observes that 
the total number of edges is of the order of the size of the diagram and recalls the 
well known fact that this size is $O(n)$ (see also the proof of Lemma \bref{lem:NodesO(n)}). 
Since $r_k$ is obviously bounded by the number of sites, namely by $O(n)$ (for each $k$), the bound $O(n^2)$ follows. When the sites are distributed independently of each other and according to the uniform distribution, then with high probability all the endpoints 
will be not far from their site (indeed, when a point on a ray which is shot from the site $p_k$ is located too far from the site, then, because the sites are distributed uniformly, there is a high probability that this point will be closer to other sites, and hence be located outside the cell of $p_k$). More careful estimates show that  $r_k=O(\ln(n/\epsilon))$ and the bound $O(n\ln(n/\epsilon))$ follows. We believe that a  better (perhaps even linear) upper bound can be given: see Section \bref{sec:ConcludeRemarks}.

\section{Concluding remarks}\label{sec:ConcludeRemarks}
This paper was devoted to the (parallel) computation of the Voronoi diagram, a geometric data structure which has numerous applications in science and technology. More precisely, we presented and thoroughly analyzed the projector algorithm: a new and fairly simple algorithm for computing  Voronoi diagrams of point sites in the Euclidean plane. The algorithm is significantly different from other algorithms in the literature, and some of the concepts involved in its description are in the spirit of linear programming and optics. The main advantage of our algorithm is that it naturally  supports parallel computing, since it allows the computation of each Voronoi cell independently of the other cells, and even allows the computation of portions of a given cell independently of other portions of the same cell. This is in contrast to most other Voronoi diagram algorithms which are inherently serial in their nature and, in particular, cannot compute each cell independently of the other cells, a fact which has cast serious difficulties on attempts to implement these algorithms in a parallel computing environment. Another advantage of the new algorithm is that it can handle in a rather native way degenerate configurations of sites, again, in contrast to most of the other known algorithms. 

Here are additional relevant remarks. First, it should be emphasized again (see also Section  \bref{sec:Introduction}) that this paper is theoretical. Hence no experimental data is given. Our paper is not unique regarding this. Indeed, the literature in general and the literature related to Voronoi diagrams in particular, contains many papers discussing algorithms theoretically and without any experimental data: see e.g.,  \cite{ACGOY1988, ABY1987,  AmatoGoodrichRamos1994, BentleyWeideYao1980, ChanChen2010, ChazelleMatousek1995, Clarkson1999, ClarksonShor1989, ColeGoodrichODunlaing1996, Fortune1987, MacKenzieStout1998, Meyerhenke2005, RajasekaranRamaswami2002, ReifSen1992, ShamosHoey1975, SpielmanTengUngor2007} for a rather partial list of such papers. Despite this, we want to say something about some practical issues, or, more precisely, about the current implementation that we have. 

Its actual behavior is somewhat strange: for a reason which is not currently well understood, the running time sometimes grows in a way which is a little bit greater than linear with respect to the input (number of sites) when one processing unit is involved. However, when several processing units are involved, then the running time $t(n)$ is better than linear with respect to the input (i.e., $t(\alpha n)<\alpha t(n)$ for all tested $\alpha\in \{1,2,3,\ldots\}$; this does not contradict the obvious lower bound of $\Theta(n)$). Perhaps (in both cases) this may be related to something in the memory management or some influence of the operating system. 
  
Comparing to well-known implementations, our preliminary implementation performs quite good.  For example, it runs faster than Qhull 2011.1 \cite{QhullWeb} when a Voronoi diagram of $10^6$ sites is computed (a few seconds on an old computer). However, it does not run faster than Triangle \cite{Shewchuk1996,Shewchuk2002}. Both Qhull and Triangle are veteran  serial implementations which have adopted many enhancements over the years. Despite this, their behavior is worse than linear and their output is not always correct. In contrast, in our implementation only two people  have been involved, for a not so long period, and only one of them has done the programming work. Many enhancements are  waiting to be implemented. For instance, each vertex usually  belongs to three cells and hence it is computed three times, while in Qhull/Triangle it is computed only once. We also note that in our implementation, so far no incorrect output has been observed when double precision arithmetic was used. 

A detailed description of experimental results, including comparisons to other known implementations  such as CGAL \cite{CGALweb} and Boost \cite{Boost-web}, and more details about implementation issues in various environments, are planned to be discussed elsewhere.

The second remark is regarding Theorem  \bref{thm:CorrectnessOfTheAlg}\beqref{item:TimeComplexityUniformDistribution}. Although an $O(n\log(n/\epsilon))$ upper bound on the number of distance comparisons done in Method  \bref{method:EndpointImproved} for the computation of the whole diagram was given with probability $1-\epsilon$, we believe that by a better analysis, an  $O(n)$ upper bound can be given (with the constant $c$ inside the big $O$ symbol which may depend on $\epsilon$, say $c=1/\epsilon$). This conjecture is supported by experimental tests. It seems that in order  to obtain a linear upper bound one would need to improve the proof of Lemmas  \bref{lem:TooMuchInBox}--\bref{lem:EmptySquare}, since currently the upper bound on the  union events mentioned there (the ones dealing with the set of all the squares composed of $\alpha$ buckets) is very coarse.

Third, we believe that it is possible to extend the methods described here to other settings,  e.g., to compute Voronoi diagrams in higher dimensions, to compute  Voronoi diagrams in the flat torus (namely, in a setting with free boundary conditions) and also to compute other structures such as arrangements, certain manifolds, certain simplicial complexes, etc. In this connection, we note that if the world $X$ in which the sites are located is a bounded domain in $\R^2$ whose boundary is not necessarily polygonal (such as a circle), then we can embed $X$ and the sites in it inside a larger polygonal domain $\wt{X}$, such as a larger rectangle, and then we can apply in $\wt{X}$  the methods described here; the result will be the Voronoi cells of the given sites in $\wt{X}$, and in order to obtain the Voronoi cells of the given sites in $X$, we  simply need to intersect $X$ with the cells in $\wt{X}$ that we have already found (because from $X\subseteq \wt{X}$ we have $\{x\in X: d(x,p_k)\leq d(x,p_j), j\neq k, j\in K\}=X\cap \{x\in \wt{X}: d(x,p_k)\leq d(x,p_j), j\neq k\}$ for every index $k\in K$). 

Fourth, we believe that it is possible to improve and extend Methods \bref{method:Preprocessing} and \bref{method:EndpointImproved}, e.g., by modifying the bucketing structure so that it will handle, in a better manner, various types of site  distributions (a possible candidate in this direction: the data structure suggested in \cite{AMNSW1998}, or modifications of it).

Fifth, it will be interesting to provide a satisfactory analysis of the time complexity in the parallel case (Subsection \bref{subsec:Parallel}) for some of the models of  computation mentioned there without an accompanied analysis, to present a rigorous formulation and proof (in terms of probabilistic estimates depending on $n$ and $\epsilon$) regarding the claim that ``only in very rare  events one has $e_k>20$'', and to further develop the issues raised in Subsection \bref{subsec:ParallelAdditional} regarding the parallel case. It will also be interesting to extend Theorem  \bref{thm:CorrectnessOfTheAlg}\beqref{item:TimeComplexityUniformDistribution} to a world $X$ which is non-rectangular, possibly following the analysis presented in Remark  \bref{rem:TimeComplexityUniformDistribution}\beqref{item:NonRectangularBucket}.

Finally, it is worth saying something about non-point sites. Algorithm 1 assumes that the sites are points. However, it is possible to consider sites of any form and to use the same algorithm, up to slight modifications. Indeed, suppose that the sites $P_k$ are compact sets located inside the world $X$. Each such a set can be approximated to any required precision by a finite subset of points contained in it. We then consider these points as $P_k$ and compute the  Voronoi cell of each of the points $p\in P_k$, namely $\{x\in X: d(x,p)\leq d(x,\cup_{j\neq k} P_j)\}$. The Voronoi cell of $P_k$ is  nothing but $\cup_{p\in P_k}\{x\in X: d(x,p)\leq d(x,\cup_{j\neq k} P_j)\}$. Now the array used to store the diagram becomes a 2-dimensional array, where the index $p$ is added to it. 


 \section{Appendix 1: issues related to Method \bref{method:Endpoint} and its improvement}\label{app:ImprovedEndpoint}
 This section discusses various issues related to Method \bref{method:Endpoint} and its improvement. See Section \bref{sec:Endpoint} above for the formulation of Method \bref{method:Endpoint}.
 
\subsection{Remarks related to Method \bref{method:Endpoint}}\label{subsec_app:Endpoint} 
We start with several remarks related to Method \bref{method:Endpoint}.
 
\begin{remark}\label{rem:tCompute}
Below we discuss two simple methods for computing the intersection between the ray  $\Gamma_{\theta}$   emanating from $p$ in the direction of $\theta$ and the boundary of $X$. This ray can be represented as  $\Gamma_{\theta}:=\{p+t\theta: t\geq 0\}$. 

One method is an approximate one, in the spirit of \cite{ReemISVD2009proc}. We fix a small positive parameter $\epsilon$, select a point $y\in\Gamma_{\theta}$ which is known to be outside $X$ (just a point very far from $p$), denote $x:=p$, select the midpoint $z:=0.5(x+y)$ between $x$ and $y$, and check if $z$ is in $X$. If yes, then we let $x:=z$, otherwise we let $y:=z$. This bisecting process continues until the length of the segment $[x,y]$ is smaller than $\epsilon$. The advantage of this method is that it is simple and general (works for any type of boundary of $X$, not necessarily polygonal), but it can be somewhat slow and it does not lead to an exact output. 

The second method is exact and is restricted to the  case where we assume  that $X$ is a compact  polygon obtained from the intersection of finitely many (say $m$) half-planes $H_j$ (the basic  assumption in this paper). In this case the boundary of $X$ is a finite union $\cup_{j=1}^m I_j$ of line segments (edges) $I_j$. Given an index $j\in \{1,\ldots,m\}$, we can write $H_j=\{x=(x_1,x_2)\in\R^2: \langle N_j,x\rangle\leq c_j\}$ where $N_j\in\R^2$ is the normal to the boundary line $L_j=\{x=(x_1,x_2)\in\R^2: \langle N_j,x\rangle=c_j\}$ of $H_j$, and $c_j\in\R$ is a  constant. Since $p\in X$, the inequalities  $\langle N_j,p\rangle\leq c_j$ should hold for all $j\in \{1,\ldots,m\}$. 

The considered ray intersects $L_j$ if and only if there exists $t_j>0$ such that $p+t_j\theta\in L_j$, i.e., $\langle N_j,p\rangle +t_j\langle N_j,\theta\rangle=c_j$ (the case  $t_j=0$ is  impossible since $p$ is assumed to be in the interior of $X$). Therefore, either  $\langle N_j,\theta\rangle=0$ and then the ray is parallel to $L_j$ and the intersection is empty, or  $\langle N_j,\theta\rangle\neq 0$ and then 
\begin{equation}\label{eq:tBoundary}
t_j=\frac{c_j-\langle N_j,p\rangle}{\langle N_j,\theta\rangle}, 
\end{equation}
where it should be verified that $t_j>0$. 

Summarizing the above-mentioned discussion, we go over all the boundary edges $I_j$ and check whether $t_j$ defined by \beqref{eq:tBoundary} is well defined (i.e.,  $\langle N_j,\theta\rangle\neq 0$) and is positive. Among those $t_j$ which satisfy these conditions we still need to verify that $p+t_j\theta\in X$, i.e., that  $p+t_j\theta\in H_k$ for all $k\in \{1,\ldots,m\}$. Analytically, we should verify that for each such $j$ we have
\begin{equation}\label{eq:tInequality}
t_j\langle N_k,\theta\rangle\leq c_k-\langle N_k,p\rangle,\quad\forall k\in\{1,\ldots,m\}.
\end{equation}
Any $t_j$ which satisfies \beqref{eq:tInequality} leads to the point of intersection between the ray $\Gamma_{\theta}$ and the boundary of $X$ (usually there will be a unique $t_j$, unless the ray passes through a vertex of $X$; the existence of at least one $t_j$ is a consequence of the intermediate value theorem: see the proof of Lemma \bref{lem:DifferentEdge}). As a final remark, we note that in \beqref{eq:tInequality} we do not really need to go over all $k\in \{1,\ldots,m\}$ but rather on those $k\neq j$ for which $t_k$ defined by \beqref{eq:tBoundary} is well defined and positive. 
\end{remark}

\begin{remark}\label{rem:Array}
Whenever a distance comparison is made with some site $a=p_j$, $j\neq k$ (or even with $p$ itself) and it is found that $d(y,p)\leq d(y,a)$, then there is no need to consider this site in later distance comparisons. Indeed, the previous inequality means that $y$ is in the halfspace of $p$ (with respect to $a$). Because $y$ always remains on the same ray and gets closer to $p$, it remains in this halfspace, i.e., also later temporary endpoints $y$ will satisfy $d(y,p)\leq d(y,a)$. However, even if $d(a,p)<d(a,y)$, then the site $a$ should not be considered anymore, since this case implies that $y$ will be moved to the intersection between its ray and the bisector line between $p$ and $a$, and so in the next iteration it will satisfy $d(y,p)=d(y,a)$, and in later iterations it will satisfy $d(y,p)\leq d(y,a)$ (in particular, the total number of distance comparisons is not greater than $n$). 

A simple way to perform the above-mentioned operation of not considering a given site anymore is to go over the array of sites by incrementing an index starting from the first site. By doing this each site will be accessed exactly one time  and when the index arrives at the last site, the process will end.  
\end{remark}

\begin{figure}[t]
\begin{minipage}[t]{1\textwidth}
\begin{center}
{\includegraphics[scale=0.62]{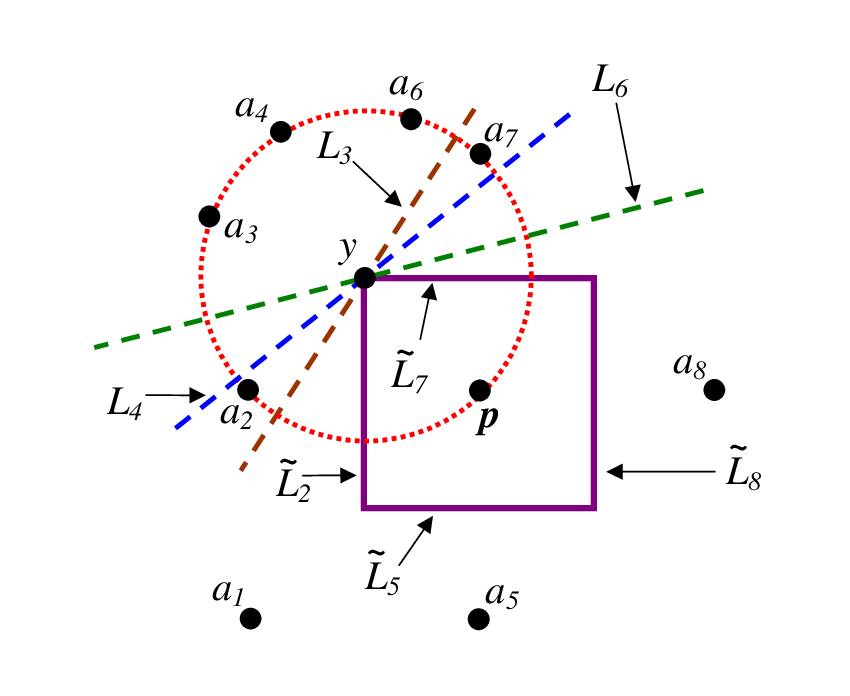}}
\end{center}
 \caption{Illustration of Remark  \bref{rem:bisector}.  Here $y$ is a vertex of the cell of $p$  (the square having $\tilde{L}_2$, $\tilde{L}_5$, $\tilde{L}_7$, and $\tilde{L}_8$ as its edges),  $EquiDistList=\{a_2,a_3,a_4,a_6,a_7\}$, and $L_i$ is the bisector between $p$ and $a_i$ for each $i\in \{3,4,6\}$. The smallest and largest angles (counterclockwise) $\angle pya_i\in [0,2\pi)$, $a_i\in EquiDistList$ are $\angle pya_7$ and $pya_2$, respectively, and indeed $L_7$ and $L_2$ contain edges of the cell of $p$. }
\label{fig:DegenerateCase}
\end{minipage}
\end{figure}

\begin{remark}\label{rem:bisector}
When finding the endpoint $y$, one can also find all of its neighbor sites since in the last time Step \beqref{item:InCell} is performed, one can easily find all the sites $a$ satisfying $d(a,y)=d(p,y)$ simply by storing any site $a$ satisfying this equality. Call the corresponding list $EquiDistList$.  Each $a\in EquiDistList$ induces a corresponding  bisector line $L$ between $p$ and $a$. In the rare  event where $y$ coincides with a vertex of the cell, there may be $a\in EquiDistList$ 
whose bisector $L$ may not contain an edge of the cell but rather it intersects the cell only at the vertex $y$ (this can happen only when $EquiDistList$ has at least three different elements, i.e., at least three distinct sites located on a circle around $y$, and hence this cannot happen when the sites are in general position). See Figure  \bref{fig:DegenerateCase} for an illustration. This is a problem, since we want to make sure that when we make operations with the endpoint, we use a line $L$ on which this endpoint is located and on which an edge of the cell is located. 

In order to overcome the problem, we  need to determine the sites $a\in EquiDistList$  that induce bisectors which contain an edge of the cell, and then to consider these sites and their associated bisector lines for later operations. This can be achieved by sorting, in an increasing order, all the angles $\angle pya$, $a\in EquiDistList$ (we choose, say, the counterclockwise direction and assume that the angles are between 0 and $2\pi$). The sites corresponding to the smallest and largest values are the ones which induce  a desired bisector, and we can associate with $y$ any one of these sites and their corresponding bisectors. 
\end{remark}

\begin{remark}\label{rem:correction}
Method  \bref{method:Endpoint} is correct, as shown below. First, the method terminates after finitely many steps because 
there are finitely many sites. The point $w$ in Step  \beqref{item:u} is well defined because $y$ is in the halfplane of $CloseNeighbor$ and hence the considered ray  intersects the boundary $L$ of this halfplane. The point $y$ is outputted in Step  \beqref{item:Output}, and by the description of this step $y$ is in the Voronoi cell of $p$. If $y$ is on the boundary of $X$, then, since it is in the cell, it is an endpoint (according to \beqref{eq:Tdef}). Otherwise $y$ is on a bisector between $p$ and another site, and therefore any point on the ray beyond $y$ is outside the Voronoi cell of $p$. Since $y$ is in the cell, this means that $y$ must be an endpoint (again, see \beqref{eq:Tdef}; here $A$ is the set of all sites $a\neq p$ and $T(p,\theta)=|y-p|$).  
\end{remark}

\subsection{The improvement of Method \bref{method:Endpoint}}\label{subsec_app:ImprpovedMethod} 
 This subsection presents the full details of an improvement of Method \bref{method:Endpoint}, an improvement which was mentioned briefly in Subsection \bref{subsec:UniformDistribution} above. Related issues are also discussed.
  
We use the following terminology. We denote by $B[y,r]$ the disk of radius 
$r>0$ and center $y=(y_1,y_2)\in\R^2$, and by $S[y,r]$ the square of radius 
$r>0$ and center $y\in \R^2$, i.e., $S[y,r]:=\{x\in \R^2: |y-x|_{\infty}\leq r\}$ where $|(w_1,w_2)|_{\infty}:=\max\{|w_1|,|w_2|\}$  is the $\ell_{\infty}$ norm of a point $w\in \R^2$. In other words,   $S[y,r]=[y_1-r,y_1+r]\times [y_2-r,y_2+r]$. The integer rectangle $S_I[y,r]$ is the rectangle defined as follows: 
\begin{equation}\label{eq:S_I}
S_I[y,r]:=\Big[s\lfloor (y_1-r)/s\rfloor,s(\lfloor (y_1+r)/s\rfloor+1)\Big]
\times \Big[s\lfloor (y_2-r)/s\rfloor,s(\lfloor (y_2+r)/s\rfloor+1)\Big].
\end{equation}
Here $\lfloor\cdot\rfloor$ is the floor function, i.e., $\lfloor t\rfloor$ is the largest integer  not exceeding $t\in \R$. See Figures \bref{fig:VoronoiParallelBoxImprove1}--\bref{fig:VoronoiParallelBoxImprove2} for an illustration. Usually $S_I[y,r]$ is the smallest rectangle which contains $S[y,r]$ and composed of buckets. The exceptions are when $S[y,r]$ itself is  composed of buckets and then $S_I[y,r]$ contains one more column of buckets to the right, and one more row of buckets to the up. It is convenient to work with the integer rectangle for the theoretical analysis, and this is what will be done from now on. 

\begin{method}$\,${(\bf Preprocessing)}\label{method:Preprocessing}
\begin{itemize}
\item {\bf  Input: } A rectangle $X$ (the world), the sites, a positive number $s$ (the side length of the buckets), two positive integers $m_1$ and $m_2$, a positive number $\omega$ (a width) which is an integer multiplied by $s$; we also assume that $X=[0,m_1s]\times [0,m_2s]$, namely the side lengths of $X$ are $m_1s$ and $m_2 s$. 
\item {\bf  Output: } a bucketing structure. 
\end{itemize}
\begin{enumerate}
\item Decompose $X$ into $m_1\cdot m_2$ buckets, where each bucket is a square of side length $s$. 
\item Associate each site to the bucket containing it and having the lowest 
indices, i.e., if the site is, say, $a:=(a_1,a_2)\in \R^2$, then the indices of the associated bucket are $i_j=\lfloor a_j/s\rfloor$, $j\in \{1,2\}$. 
\item Enlarge the grid of buckets around $X$ by the width $\omega$, namely by adding a shell of buckets (whose thickness is $\omega$) around $X$ 
where each added bucket is empty of sites. 
\end{enumerate}
\end{method}

The last step in Method \bref{method:Preprocessing} is performed in order 
to avoid complications when some rectangles go out of $X$. This last stage can be avoided by working with the intersection of integer rectangles and $X$, but this complicates a bit the programming and the analysis. 

We now describe Method \bref{method:EndpointImproved} which improves upon Method  \bref{method:Endpoint}. Its theoretical justification is described in Proposition \bref{prop:MethodEndpointImprovedIsCorrect} below.

\begin{method}{\bf (Improved endpoint computation)}$\,$\label{method:EndpointImproved}
\begin{itemize}
\item {\bf  Input: } The world $X$ (a rectangle); the sites in the bucketing data structure 
mentioned in Method \bref{method:Preprocessing}, with the same parameter $s>0$; a site $p\in X$; a unit vector $\theta$; an integer $\beta>0$.
\item {\bf  Output: } The endpoint $p+T(p,\theta)\theta$. 
\end{itemize}
\begin{enumerate}
\item\label{method:EndpointImproved:1.01} Let $\eta:=4\sqrt{2}(\beta+1.01)s$ and $y:=p+\eta\theta$. In addition, let $y_X$ be the intersection between the boundary of $X$ and the ray emanating from $p$ in the direction of $\theta$ (see Remark \bref{rem:tCompute} regarding the computation of $y_X$). 
\item\label{step:twice_y} Check whether $y$ is in $X$. 
\item\label{step:2eta>d(y_x,p)} If $y$ is not in $X$, equivalently, if $\eta>d(y_X,p)$, then let $y:=y_X$. In this case  we consider $y$ as the first temporary endpoint, and we go to Step \beqref{item:RectangularListOfsites}. 
\item\label{item:S_I[y,2 beta s]} If $y\in X$, then we construct  $S_I[y,2\beta s]$ and check whether 
$d(y,p)\leq d(a,y)$ for all sites $a\in S_I[y,2\beta s]$. If yes, then let $y:=y_X$. Otherwise, $y$ is outside the cell of $p$. In both cases we consider $y$ as the first temporary endpoint. 
\item\label{item:RectangularListOfsites} Create the ``rectangular  list of sites'': this is simply the list of buckets and sites contained in $S_I[\tilde{y},d(p,\tilde{y})]$, where $\tilde{y}$ is the first temporary endpoint (see Remark \bref{rem:MethodEndpointImproved}\beqref{item:SubsetRectangularList} below; note that at this stage $y=\tilde{y}$). 
\item\label{item:y_outside_cell} If $y$ is known to be outside the cell of $p$, then go to Step \beqref{item:CloseNeighborImproved}. Otherwise, we know from previous steps that $y=y_X$. Let  $L$ be the boundary line (an edge of $X$) on which $y$ is located.
\item\label{item:InCellImproved}	Check whether $y$ is in the cell of $p$ by comparing $d(y,p)$ to $d(y,a)$ for all sites $a$ in $S_I[y,d(p,y)]$ (as a subset of the rectangular list of sites: see Remark \bref{rem:MethodEndpointImproved}\beqref{item:SubsetRectangularList} below). Whenever a site $a$ has been considered, it is removed from the rectangular list of sites and is not considered anymore. If either $d(y,p)\leq d(y,a)$ for all sites $a$ in $S_I[y,d(p,y)]$ or if $S_I[y,d(p,y)]$ is empty of sites, then $y$ belongs to the cell of $p$ (see Remark \bref{rem:MethodEndpointImproved}\beqref{item:EmptyOfSites} below) and we go to Step \beqref{item:OutputImproved}. If $d(y,a)<d(y,p)$ for some site $a\in S_I[y,d(p,y)]$, then we go to Step \beqref{item:CloseNeighborImproved}. 
\item\label{item:OutputImproved}	If $y$ is in the cell of $p$, then $y$ is the endpoint and $L$ is its associated  line, namely $L$ is a bisector or an edge of the boundary of $X$ on which $y$ is located. The calculation along the ray is complete.
\item\label{item:CloseNeighborImproved}	Now we know that $y$ is not in the cell of $p$, namely $d(y,a)<d(y,p)$ for some site $a$. Let $CloseNeighbor:=a$ (if there are several such sites $a$, then pick one of them, no matter which).
\item\label{item:uImproved}	Find the point of intersection (call it $w$) between the given  ray and the bisector line $L$ between $p$ and $CloseNeighbor$.  This intersection is always nonempty.  The line $L$ can  easily be found because it is vertical to the vector $p-CloseNeighbor$ and passes via the point $(p+CloseNeighbor)/2$. 
\item\label{item:y=u}	let $y:=w$; go to Step \beqref{item:InCellImproved}.
\end{enumerate}
\end{method}
Illustrations related to Method \bref{method:EndpointImproved} can be found in Figures 
\bref{fig:VoronoiParallelBoxImprove1}--\bref{fig:VoronoiParallelBoxImprove2}.

\begin{figure}[t]
\begin{minipage}[t]{1\textwidth}
\begin{center}
{\includegraphics[clip, scale=0.8]{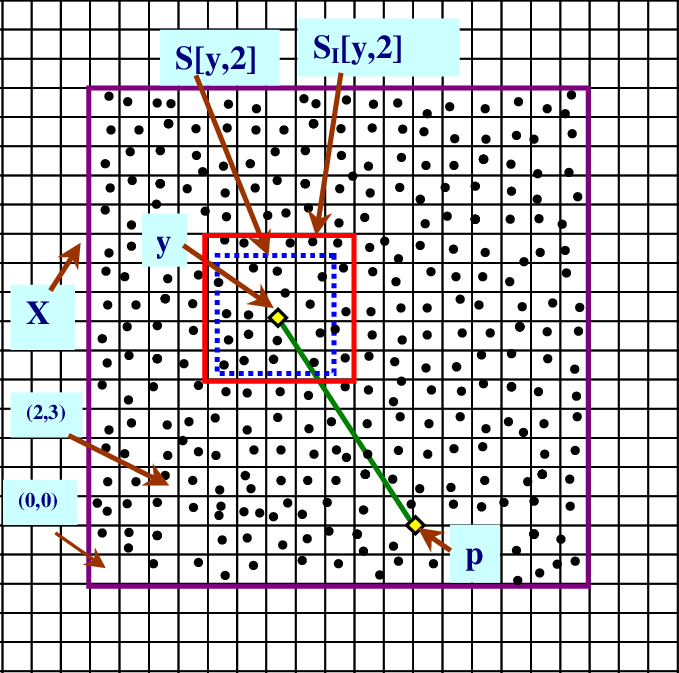}}
\end{center}
 \caption{An illustration of issues related to Method \bref{method:EndpointImproved}. Here $s=1$.}
\label{fig:VoronoiParallelBoxImprove1}
\end{minipage}
\end{figure}

\begin{figure}[t]
\begin{minipage}[t]{1\textwidth}
\begin{center}
{\includegraphics[clip, scale=0.8]{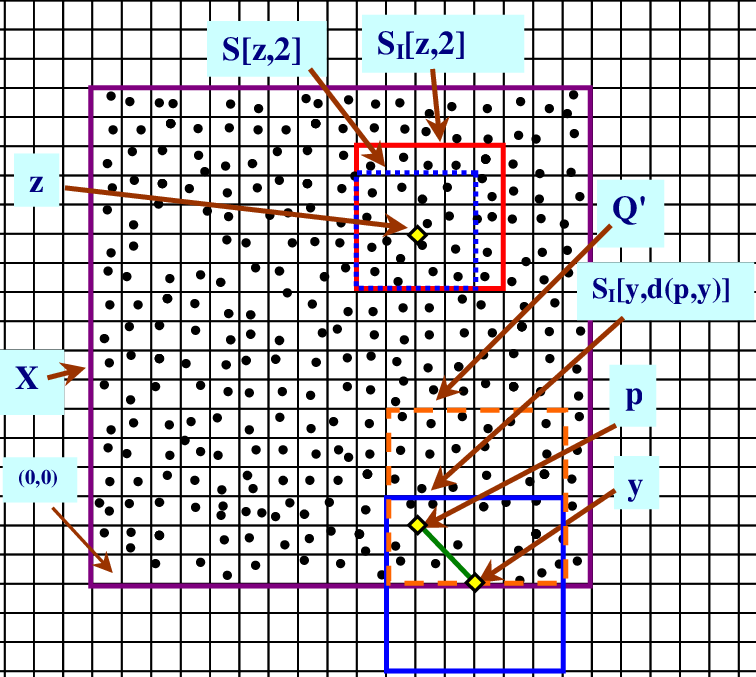}}
\end{center}
 \caption{Another illustration of issues related to Method \bref{method:EndpointImproved}. Here $s=1$ and $Q'$ is an integer rectangle contained in $X$ and contains $X\cap S_I[y,d(p,y)]$. It also  contains the same number of buckets as $S_I[y,d(p,y)]$.}
\label{fig:VoronoiParallelBoxImprove2}
\end{minipage}
\end{figure}

\begin{remark}\label{rem:MethodEndpointImproved}
Here are a few comments regarding Method \bref{method:EndpointImproved}. 
\begin{enumerate}[(i)]
\item The number $1.01$ in $\eta$ from Step \beqref{method:EndpointImproved:1.01} is somewhat arbitrary. It can be replaced by any number greater than 1, e.g., 1.0000001. See the proof of Lemma \bref{lem:SitesInBetaSquare}. 
\item The distance comparisons between $d(p,y)$ and $d(y,a)$ in Steps \beqref{item:S_I[y,2 beta s]} and \beqref{item:InCellImproved} can be  done, e.g., by  starting with shell 0, namely with the bucket in which $y$ is located, then comparing $d(p,y)$ to $d(y,a)$ for all sites $a$ in shell 0; then (if $d(p,y)\leq d(y,a)$ for all sites $a$ in shell 0), comparing $d(p,y)$ to $d(y,a)$ for all sites $a$ in shell 1 (namely the shell of eight buckets around shell 0), then (if $d(p,y)\leq d(y,a)$ for all sites $a$ in shell 1) doing the same with shell 2, i.e., with the shell of 16 buckets around shell 1, and so forth, until some site $a$ in one of the shells is found to satisfy $d(y,a)<d(p,y)$ or until we go over all sites $a$ in the corresponding integer rectangle (namely, the integer rectangle is $S_I[y,2\beta s]$ in Step \beqref{item:S_I[y,2 beta s]}  and $S_I[y,d(p,y)]$ in Step \beqref{item:InCellImproved}). 
\item\label{item:2beta_s} When the sites are independently and uniformly distributed, then there is a rather low probability that in Step \beqref{item:S_I[y,2 beta s]} the inequality $d(p,y)\leq d(y,a)$ holds for all the sites in $a\in S_I[y,2\beta s]$. The reason is as follows: elementary considerations show that at least quarter of $S_I[y,2\beta s]$ is contained in $X$, and this quarter is a square of side  length $(2\beta+1)s$, namely this square is composed of $2\beta+1$ buckets; an analysis similar to the one given in the proof of Lemma \bref{lem:EmptySquare} shows that this quarter contains, with probability which not smaller than $1-ne^{-(2\beta+1)^2}$, at least one site (so, if $n=10^{15}$ and $\beta=3$, then this probability is not smaller than $0.999999$); however, any site $a$ located in $S_I[y,2\beta s]$ satisfies $d(y,a)<d(p,y)$: this is just a consequence of Lemma \bref{lem:SitesInBetaSquare} below.

\item\label{item:DeleteSite} The reason for creating the rectangular list of sites (in Step \beqref{item:RectangularListOfsites}) is to avoid problems in later computations. Indeed, consider for example the case where we first compute the endpoint of some ray in the direction of $\theta$. Assume that instead of creating the rectangular list of sites and removing the sites from it, we remove, in Step \beqref{item:InCellImproved}, the sites from the global list of sites. If now we want to compute the endpoint of another ray, or the Voronoi cell of another site, we are left with much less sites in the global list of sites; in other words, the list of sites that we have in our hands for these calculations is erroneous (this list might even be empty, and so we will not be able to start the computations related to the Voronoi cell of the sites which have been removed from the list).

\item\label{item:SubsetRectangularList} The rectangular list of sites is just a bucketing structure of the type mentioned in Method \bref{method:Preprocessing}, but instead of containing all the buckets and sites in the entire world $X$, it contains copies of the buckets and sites contained in the integer rectangle $S_{I}[\tilde{y},d(p,\tilde{y})]$, where $\tilde{y}$ is the first temporary endpoint. After creating this rectangular list, each time when we check whether some sites are in $S_I[y,d(p,y)]$, where $y$ is a temporary endpoint, we consider $S_I[y,d(p,y)]$ as a subset of  the rectangular list, namely we look for  sites which are in the rectangular list and  their locations are in $S_I[y,d(p,y)]$. 

\item\label{item:EmptyOfSites} Lemma \bref{lem:SitesInLargeSquare} (parts \beqref{item:true_endpoint1}--\beqref{item:true_endpoint2}) below ensures that for the sake of distance comparisons, it is sufficient to consider only the sites located inside the integer rectangles $S_{I}[y,d(p,y)]$, where $y$ is a temporary endpoint. Once it is verified that $d(p,y)\leq d(a,y)$ for all $a\in S_{I}[y,d(p,y)]$, then $y$ is the true endpoint, or, if $S_{I}[y,d(p,y)]$ is empty of sites, then $y$ is also the true endpoint.

\item It is possible to formulate other versions of Method  \bref{method:EndpointImproved}. One possible version is a one in which in we do not create the rectangular list of sites in Step \beqref{item:RectangularListOfsites}, and in Step \beqref{item:InCellImproved} we do not delete the site $a$ from the list of sites after comparing  $d(p,y)$ to $d(a,y)$, but rather we do nothing with $a$. This version of the method might be a bit slower (since the site $a$ may be considered again in future distance comparisons), but it is still correct, for reasons which are quite similar to the ones given in the proof of Proposition \bref{prop:MethodEndpointImprovedIsCorrect}.
\end{enumerate}
\end{remark}

\section{Appendix 2: several theoretical and practical aspects related to the algorithm and to Theorem  \bref{thm:CorrectnessOfTheAlg} }\label{app:TheoryPractice}
This section discusses issues related to the time complexity of Algorithm 1, as well as other theoretical and practical aspects related to this algorithm. For the sake of convenience of the reading, the section is divided into a few subsections. 

\subsection{The model of computation: the serial case}\label{subsec:SerialModelComputation}  
The model of computation in the case of one processing unit is an ordinary computer with memory, namely a sequential model such as the von Neumann model. We assume that each of the $n$ sites is stored in a memory cell of its own, an assumption which usually holds in practice, e.g., when these sites are originated from real-world measurements. Note, however, that one can think of cases in which the above-mentioned assumption does not hold, as in  the  case when the sites can be obtained by some predetermined formula so that  their locations are known in advance, at least in principle; as an illustration to this situation, suppose that $X:=[0,1]^2$ and that the number of sites $n$ satisfies $n:=\nu^2$ for some natural number $\nu$ which is given as an input; suppose further that $p_k$, namely site number $k\in\{1,\ldots,n\}$, satisfies  $p_k:=(1/(\delta_k+2),1/(\sigma_k+1))$ where $\delta_k\in \{0,1\ldots,\nu-1\}$ and $\sigma_k\in \{1,2,\ldots,\nu\}$ are the unique nonnegative integers which satisfy $k=\delta_k \nu+\sigma_k$ (Euclidean division); in this case we do not need to store any of the sites, but can calculate them in real time, although this procedure is time consuming. 

A further assumption in the model is that arithmetic operations, such as  
comparisons between two numbers, array manipulations, reading or writing to the memory and so on, are $O(1)$ independently of $n$. This assumption is implicit in the analyses of the time complexity  of many geometric algorithms, including all of the Voronoi algorithms that we are aware of. 
In particular, this assumption implies that each number is represented by at most $MaxBit$ bits for some known and fixed number $MaxBit$ (say, 64 or 8196), an assumption which holds true in standard computers and in widely used data structures (different data types, e.g., ``int'' and ``double'', may be represented using a different number of bits, but $MaxBit$ is 
an upper bound on the sizes of all of these types). However, this assumption casts a (very large) strict upper bound on the size of the input $n$, namely $n\leq 2^{MaxBit}$.  If one wants to take into account the case where arithmetic operations depend on $n$ (including the number of digits after the floating point) and if we further assume that the model of computation allows an unlimited number of memory cells, then a multiplying factor expressing this dependence (e.g., $O(\log(n))$) should be added to all of the complexity results in Theorem  \bref{thm:CorrectnessOfTheAlg} below, as well as to all of the corresponding complexity results in the literature. 

\subsection{A few comments related to Theorem  \bref{thm:CorrectnessOfTheAlg}\beqref{item:TimeComplexityUniformDistribution} }
Here are several comments regarding Theorem  \bref{thm:CorrectnessOfTheAlg}\beqref{item:TimeComplexityUniformDistribution}:

\begin{remark}\label{rem:TimeComplexityUniformDistribution}
\begin{enumerate}[(i)]
\item First, the parameter OneTwo is needed only to ensure that $\alpha_1$ is an odd number so that $\beta$ will be an integer, and also that  $\sqrt{\alpha_1}>\lfloor\sqrt{\ln\left(n(1+2e^{-400})/\epsilon\right)}\rfloor$, because of a technical reason needed in the proof. 
\item Second, the number 1.01 which appears in \beqref{eq:alpha2} is somewhat arbitrary: any number greater than 1 is OK, say 1.00000001. 
\item\label{item:alpha_1< alpha_2<} Third, the conditions \beqref{eq:ln(n/epsilon)}--\beqref{eq:alpha2m} are always satisfied for large enough $n$ because the floor function $\lfloor \cdot \rfloor$ grows at the same rate as the identity function, and  for arbitrary $\tau_1,\tau_2,\tau_3>0$ one has $\lim_{n\to\infty}(\tau_1  n/\ln^{\tau_2}(\tau_3 n))=\infty$. 
\item Fourth, once $n\in\N$ is large enough, then the numbers $m_1:=m_2:=\sqrt{m}$ mentioned after \beqref{eq:m1m2} are fine in the sense that they do satisfy \beqref{eq:m1m2}, because from \beqref{eq:alpha2m} it follows that $\sqrt{m}>2\sqrt{\alpha_2}$. Of course, in this case $X$ should be a square. Conversely, once we start with a world $X$ which is a square with a side length $\ell_X>0$, and define, for each $n\in\N$ (large or not) $m_1:=m_2:=\lfloor\sqrt{n}\rfloor$, $s:=\ell_X/m_1$ and $m:=m_1^2$, then obviously $m_1=m_2=\sqrt{m}$ and $m=\lfloor\sqrt{n}\rfloor^2$, and now, if  we let $100<n\in\N$ to be sufficiently large, then, as explained in Item \beqref{item:alpha_1< alpha_2<} above, Conditions \beqref{eq:ln(n/epsilon)}--\beqref{eq:alpha2m} will be satisfied (and hence also the condition $\sqrt{m}>2\sqrt{\alpha_2}$ mentioned in \beqref{eq:m1m2}, since it follows immediately from \beqref{eq:alpha2m}). In other words, if we start with a world $X$ which is a square and fix some $\epsilon\in (0,1)$,  then, for all $n\in\N$ large enough, if $n$ distinct point sites are generated in $X$ independently according to the uniform distribution, and if we apply our computational methods (Algorithm 1 and Method \bref{method:EndpointImproved}) for computing the Voronoi cells of these sites, then we are guaranteed that with a probability which is at least $1-\epsilon$, the total number of distance comparisons for computing all the $n$ Voronoi cells is bounded by some universal constant multiplied by $n\log(n/\epsilon)$.  

\item Fifth, it is worth substituting some  values in the various parameters involved in the theorem; for instance, if $200000\leq n\leq 10^{100}$ and $\epsilon=0.000001$, then \beqref{eq:ln(n/epsilon)}--\beqref{eq:alpha2m} are satisfied, and, moreover, it is guaranteed that with probability which is at least 0.999999, not more  than $686664n$ distance comparisons are made in the computation of all of the Voronoi cells. 
[Note: the above-mentioned estimate on the number of distance comparisons is rather practical for the foreseeable future as explained in the following lines; according to contemporary estimates, the number of subatomic particles in the observable universe (namely, the part of the space located in a ball centered at the Earth and having a radius of about 46 billion light years) is estimated to be between $10^{80}$ \cite{Padilla2017misc} to $10^{82}$ \cite{Villanueva2009-2018misc}; now, in our model of computation we assume that each of the $n$ sites is stored in a separate memory unit (namely, in a separate memory cell); since each  such a memory unit is composed of at least one (actually much more than one) subatomic particle and all the memory units are different from each other, we conclude that $n$, namely, the number of sites, is bounded above by the number of these subatomic particles; therefore $n$ is bounded above by $10^{82}$, and so it definitely cannot exceed $10^{100}$.] 

\item Sixth, it is worth elaborating more on the meaning of the probabilistic estimates given in Part \beqref{item:TimeComplexityUniformDistribution} of Theorem \bref{thm:CorrectnessOfTheAlg}. In our case, since the $\epsilon$ parameter can be arbitrary small, the estimate $1-\epsilon$ can be arbitrary close to 1, and hence ``with high  probability'' Algorithm 1 will do at most $O(n\log(n/\epsilon))$ distance comparisons when it computes the entire Voronoi diagram using one processing unit. On the other hand, in the literature one can see frequently other ``probabilistic notions'' regarding the time complexity of some algorithms, for example  that the time complexity is, ``on average'', such and such (other common phrases:  ``the expected time  complexity'' or ``the time complexity in expectation'' or ``the average case time complexity''). 
These probabilistic notions are considerably different from the probabilistic notion that we use, namely different from ``with high  probability''. 

To see the difference, consider for instance the case where $n$ balls are put in $n$ boxes so that each ball has the same probability $1/n$ to be in each of the $n$ boxes, and the balls are distributed in the boxes in an independent manner (i.e., we can think of the locations of the balls as being i.i.d random variables distributed according to the uniform distribution).  Elementary considerations show that the expected number of balls in each box is 1. However,  the probability that in each box there is exactly one ball is $n!/n^n$. As follows from  Stirling's formula (see \beqref{eq:StirlingRobbins} below), this number tends exponentially fast to 0 as $n$ tends to  infinity. In other words, although on average  there is one ball in each box, the probability that a random configuration (according to the uniform distribution) of balls in the boxes leads to exactly one ball in each box is very low.  Hence, although there are ``linear expected time'' algorithms  for computing the Voronoi diagram of uniformly distributed sites    \cite{BentleyWeideYao1980, Dwyer1991}, it definitely does not mean, and - as far as we know - it has never been proved  mathematically, that with high probability these algorithms do $O(n)$ number of calculations when computing the Voronoi diagram of a random configuration of sites. 

\item\label{item:NonRectangularBucket} Seventh, one may wonder whether it is possible to generalize Theorem  \bref{thm:CorrectnessOfTheAlg}\beqref{item:TimeComplexityUniformDistribution} for the case where $X$ is not a square or any other rectangle for which the conditions of Theorem  \bref{thm:CorrectnessOfTheAlg}\beqref{item:TimeComplexityUniformDistribution} hold. My intuition is that it is possible, and in what follows I want to present two possible paths for showing this. 

The first path is still for special cases of $X$, such as when $X$ is a perfect (equilateral) triangle or a perfect hexagon. In this case, instead of buckets which are squares as the building blocks in the  partitioning of $X$, one will use buckets of other forms such as perfect triangles and perfect hexagons, respectively. It seems that the analysis given in Subsection \bref{subsec:Improvements} for proving Theorem  \bref{thm:CorrectnessOfTheAlg}\beqref{item:TimeComplexityUniformDistribution} can be extended  without much efforts also for these new cases.    

The second path for extending Theorem  \bref{thm:CorrectnessOfTheAlg}\beqref{item:TimeComplexityUniformDistribution} is for the case where $X$ is a general bounded and convex polygon with a nonempty interior, and even for the case where  $X$ is a general bounded and convex set with a nonempty interior. The idea is to partition $X$ into (usually) non-rectangular small  buckets which are induced by square buckets, with the help of a certain (usually nonlinear) one-to-one mapping $\mathscr{F}$ from the square $[-1,1]^2$ onto $X$. 

Such a mapping always exists: for instance, one can define $\mathscr{F}(x):=c_X+(\|x\|_{\infty}/\M_{X-c_X}(x))x$ for all $0\neq x\in [-1,1]^2$ and $\mathscr{F}(0):=c_X$, where $c_X$ is the center of mass of $X$, the $\|\cdot\|_{\infty}$ norm satisfies $\|(x_1,x_2)\|:=\max\{|x_1|,|x_2|\}$, and $\M_{X-c_x}$ is the Minkowski functional (which induces a so-called ``convex distance function'') of the translation of $X$ by $c_X$ so that the center of mass of the translated $X$ will be the origin. The mapping $\mathscr{F}$ is continuous since $\M_{X-c_X}$ is continuous (the only continuity issue of $\mathscr{F}$ is at the origin, but since the Minkowski functional is homogeneous, one has $\M_{X-c_X}(\mathscr{F}(x)-c_X)=(\|x\|_{\infty})/\M_{X-c_X}(x))\M_{X-c_x}(x)=\|x\|_{\infty}\to 0$ as $x\to 0$, and so $\mathscr{F}$ is continuous at the origin too). Moreover, it is invertible and its inverse $\mathscr{F}^{-1}$ satisfies $\mathscr{F}^{-1}(y)=(\M_{X-c_X}(y-c_X)/\|y-c_X\|_{\infty})(y-c_X)$ for all $c_X\neq y\in X$ and $\mathscr{F}^{-1}(c_X)=0$, and also $\mathscr{F}^{-1}$ is continuous (this is a consequence of, say, \cite[Inequality (2.2)]{ReemReich2018jour(Polarity)} with $C:=X-c_X$ and the fact that all norms on $\R^2$ - in particular, the Euclidean norm and the $\|\cdot\|_{\infty}$ norm - are equivalent). In other words, the above-mentioned $\mathscr{F}$ is a homeomorphism. 

Now, given any  homeomorphism $\mathscr{F}$ from $[-1,1]^2$ onto $X$, whatever its form is (not necessarily the one mentioned in the previous paragraph), we look at the set $\{\mathscr{F}^{-1}(p_k): k\in K\}$  located at $[-1,1]^2$, a set which is induced by the set of sites $\{p_k: k\in K\}$ located in $X$. We call any point in  $\{\mathscr{F}^{-1}(p_k): k\in K\}$  a pre-site. We now partition $[-1,1]^2$ into square buckets and put any pre-site in a corresponding bucket according to Method \bref{method:Preprocessing}. We then partition $X$ into $\lfloor \sqrt{n}\rfloor^2$ (usually) non-rectangular small buckets by letting $B_i:=\mathscr{F}(\wt{B_i})$ whenever $\wt{B}_i$ is a square bucket in $[-1,1]^2$, $i\in\{1,2,\ldots,\lfloor \sqrt{n}\rfloor^2\}$, and we put any site $p_k\in X$ in its corresponding non-rectangular small bucket, namely $p_k$ is in some non-rectangular bucket $B\subseteq X$ if and only if $\mathscr{F}^{-1}(p_k)$ is in a square bucket $\wt{B}\subseteq [-1,1]^2$. Our intuition is that the analysis done in Subsection \bref{subsec:Improvements}, possibly with some modifications, for the case of non-rectangular buckets, will be fine too, at least for certain mappings $\mathscr{F}$, possibly homeomorphisms which satisfy additional properties, but of course, a full proof is required.  (We also note that in the special case where $X$ is a parallelogram, then one can obviously let $\mathscr{F}$ to be any affine mapping which maps $[-1,1]^2$ onto $X$, such as the affine mapping which is induced by the linear mapping which maps the basis $\{(0,1),(1,0)\}$ onto the basis which induces $X$.)

\end{enumerate}
\end{remark}

\subsection{Possible improvements }  
It is not clear whether the bound on the time complexity presented in Theorem  \bref{thm:CorrectnessOfTheAlg} is tight, since maybe it can be improved, possibly by a better analysis and by other enhancements to the algorithm in addition to the one presented in Subsection  \bref{subsec_app:ImprpovedMethod}. It might be that such enhancements can  give better bounds or at least better constants (at least in the one processor case): for example, this may be achieved by computing the cells in a certain order (e.g., using plane sweep), or by improving the endpoint computation (Method \bref{method:Endpoint}), or by taking into account in a better way the distribution of the sites (for example, by a more general bucketing technique than the one presented in Subsection \bref{subsec_app:ImprpovedMethod}), and so on. 

\subsection{Serial time complexity: a comparison with other algorithms}\label{subsec:TimeComplexityComparison} 
The upper bound on the worst case serial time complexity (when one processing unit is involved) is at least as good as the 
bound $O(n^2)$ of the incremental method \cite{GreenSibson1977}, \cite{OhyaIriMurota1984}, \cite[pp. 242-251]{OBSC}. The bound is also better than the corresponding bound of the naive method \cite[pp. 230-233]{OBSC} which is $O(n^2\log(n))$. On the other hand, it is  worse than the $O(n\log(n))$ of some algorithms (plane-sweep, divide-and-conquer, the method based on convex hulls) \cite{Aurenhammer}. It should be emphasized however that it is still not known that the established upper bound is tight, and in addition, even if it is tight, our algorithm has various advantages, as described in Section \bref{sec:Introduction}. Furthermore, in common scenarios,  such as when the sites are distributed uniformly, the performance is better with  high probability (Theorem  \bref{thm:CorrectnessOfTheAlg}\beqref{item:TimeComplexityUniformDistribution}). 

In this connection, it is worth saying again that we assume in the formulation and the proof of Theorem  \bref{thm:CorrectnessOfTheAlg} that all the sites are distinct, namely $p_j\neq p_k$ whenever $j,k\in K$ satisfy $j\neq k$. In order to check whether this assumption does hold without increasing the time complexity estimates, one can first order the list of sites according to the lexicographic order (that is, $p_j\leq p_k$ if and only if either $p_{j,1}<p_{k,1}$ or ($p_{j,1}=p_{k,1}$ and  $p_{j,2}\leq p_{k,2}$), where $p_j=(p_{j,1},p_{j,2})$ and $p_k=(p_{k,1},p_{k,2})$), a sorting task which takes $O(n\log(n))$  machine operations; then one goes over the ordered list of sites and checks, starting from the second site in the list, whether a site is equal to the previous site, and  removes this site if indeed it is equal to the previous site. The second task takes $O(n)$ machine operations, and so in total one performs $O(n\log(n))$ machine operations for verifying  whether all the sites are distinct (instead of performing $O(n^2)$ machine operations by  going over all the $0.5n(n-1)$ pairs of sites and checking whether the two sites in each pair are equal). If some sites in the initial list are equal, then the resulting list of sites will be smaller than the initial list of sites, and this new list of sites (with, say, $n'<n$ sites) will be the input for the algorithm.

\subsection{Serial time complexity: worst case vs. common scenarios}\label{subsec:CommonScenarious} 
The difference between the potential worst case scenarios 
and the common one is similar in some sense to the Quicksort algorithm for sorting \cite{CLRS-book-2001,Hoare1962} whose average case time complexity is $O(n\log(n))$, but its worst case time complexity is $O(n^2)$. 
The simplex algorithm for linear programming \cite{Dantzig1951,Dantzig1963} provides another  example (efficient in practice, with polynomial average case complexity \cite{SpielmanTeng2004}, but exponential time complexity in the worst case \cite{KleeMinty1972}). Related phenomena   
occur in the Voronoi case, e.g., the average case complexity $O(n\log(n))$ of the incremental method \cite{GuibasKnuthSharir} versus the $O(n^2)$ worst case. 
	
\subsection{The parallel case: model of computation and time complexity}\label{subsec:Parallel}\label{subsec:Parallel} So far we have discussed mainly the serial case in which only one processing unit is involved. In this subsection we discuss the case where $Q\in\N$ processing units  are involved in the computation of the whole diagram. Our model of computation is a one in which every processing unit can work in principle independently of the other processing units. We wrote ``in principle'', since an implementation of this model in  practice may suffer from issues: for instance, if we assume a shared memory model in which the list of sites is stored once in the memory and then becomes available to all the processing units, then whenever two or more  different processing units need to consider at the same time the same site (for instance, for a distance comparison calculation), then they need to access at the same time the same memory places, and this leads to a conflict between them, unless the hardware supports concurrent read. 

One way to overcome this problem is to assume in advance the so-called PRAM  (or P-RAM) CREW model, that is, Parallel Random Access Memory Concurrent Read Exclusive Write model \cite{AumannRabin1994jour, BrodnikGrgurovic2018jour, FortuneWyllie1978incol, GibbonsSpirakis1993book, Wyllie1979PhD}, at least for the part of the memory in which the list of sites is stored. In other words, each of the memory cells in which the sites are stored can be read, at the same time, by every processor. The rest of the memory cells can be divided between the processors so that each processor will have its own memory cells and will not need to read or write into memory cells which other processors work with.

Another way to overcome this problem is to make sure that the processors are truly independent, in the sense that each one of them has its own memory and a copy of the list of sites is available to each of the processors, and they will not interact in any way with each other; of course, the task of coping the list of sites to the memory of each processing unit will add a $\Theta(Qn)$ term to the total time complexity. In this  connection we note that a linear (as a function of $n$) lower bound on the time complexity always exists no matter how large the number of processors $Q$, simply because storing the $n$ sites in the memory (or even creating the $n$ labels $1,\ldots,n$ for the sites) requires $\Omega(n)$ machine operations, although the coefficient before the $n$ is rather small (actually, even creating in advance $n$ memory cells in the machine adds an $\Omega(n)$ term, but one can argue that this term is not part of the calculation).

Similarly to Subsection  \bref{subsec:SerialModelComputation}, another assumption in the model is that arithmetic operations  such as comparisons between two numbers, array manipulations, reading or writing to the memory and so on, are $O(1)$ independently of $n$. We also assume a synchronous relationship between the processors in the sense that they start their activity at the same time and that each one computes its sets of Voronoi cells without any halt.   

In what follows we assume any of the above-mentioned models. We note that the methods described in this paper (Algorithm 1, Method  \bref{method:EndpointImproved} and so on) are suitable for other models as well, such as a shared memory model in which no concurrent read of several processors is possible (that is, EREW: Exclusive Read Exclusive Write), and a model in which the processors can interact by sharing between themselves  some data such as the edges and vertices that they find during the computation of each cell. However, the time complexity analysis in these  cases becomes more complicated and is not considered here.  

From now on we assume that processor $Q_i$ computes a set $A_i$ of cells determined by their sites, $i\in\{1,\ldots, Q\}$ (an illustration: assume that there are $n=10$ sites, $Q=3$ processors, and define $A_1:=\{1,3,5,7,9\}$, $A_2:=\{2,4,6\}$ and $A_3:=\{8,10\}$). Then the time complexity is $\max\{\sum_{k\in A_i}O(r_k e_k): i\in \{1,\ldots, Q\}\}$, as shown in Lemma \bref{lem:ParallelComplexity}. In addition, and similarly to the explanation given in Subsection  \bref{subsec:Improvements}, it can be shown that given $\epsilon\in (0,1)$, if the sites are (independently and) uniformly distributed, then with probability which is at least $1-\epsilon$, no more than $O(r_k e_k)=O(\log(n/\epsilon)e_k)$ distance comparisons are done, and hence the time complexity becomes $\max\{\sum_{k\in A_i}O(\log(n/\epsilon)e_k): i\in \{1,\ldots,Q\}\}$. This estimate is rather good, since  only in very rare  events one has $e_k>20$ (often $e_k=6$), and hence, in the common case where $e_k\leq 20$ for each $k\in K$ and where each processor computes $\approx n/Q$ Voronoi cells, the above-mentioned estimate becomes $20(n/Q)O(\log(n/\epsilon))$. 

\subsection{The parallel case: additional remarks}\label{subsec:ParallelAdditional}
In the analysis presented in Subsection \bref{subsec:Parallel}, and also in Section \bref{app:Proofs}, it is assumed that each cell is computed by a one processing unit without the help of other processing units. However, it should be noted that one of the advantages of our algorithm is that it also allows parallelizing of the computation of each cell, because each subwedge can be handled independently of other subwedges and so the work can be divided between several processing units. Hence one may  also find it helpful to analyze the case of several processing units which compute the same Voronoi cell. 

Anyway, returning back to the worst case scenario, one may want to 
avoid cases where a certain large cell (with many edges) slows down the whole computation. To avoid such a case, one can program in advance each processing unit to check during its computation of a given cell whether this  cell has too many edges (e.g., more than 20 edges); in such a case the processing unit can halt its work and go to other cells in its list of cells to be computed. The edges and vertices of the problematic cell will be found later, either from the computation of edges and vertices of neighbor cells with less edges, or by the help of other processing units which can come and help to compute the problematic cell once they finish their job of computing less problematic cells.

\section{Appendix 3: proofs}\label{app:Proofs}
In this section we prove the correctness of  Method \bref{method:EndpointImproved} (Proposition \bref{prop:MethodEndpointImprovedIsCorrect} below) and present the full proof of Theorem \bref{thm:CorrectnessOfTheAlg}. The proofs of these claims can be found in several subsections. The first subsection (Subsection \bref{subsec:MethodEndpointImprovedIsCorrect}) is devoted to the proof of Proposition \bref{prop:MethodEndpointImprovedIsCorrect}, the second subsection (Subsection \bref{subsec:SerialParallelTimeComplexity}) is devoted to the proof of most of Theorem \bref{thm:CorrectnessOfTheAlg} (with the exception of Part \beqref{item:TimeComplexityUniformDistribution}), and the last subsection (Subsection \bref{subsec:Improvements}) is devoted to the proof of Theorem \bref{thm:CorrectnessOfTheAlg}\beqref{item:TimeComplexityUniformDistribution}.

Before presenting the proofs, we want to give a few details regarding our 
notation and terminology for this section. In general, we use the notation mentioned in Section \bref{sec:Preliminaries}. In addition, in the sequel $p$ represents a given site, say $p=p_k$ for some $k\in K:=\{1,\ldots,n\}$. As already mentioned, we assume that all the sites are 
distinct, known in  advance, and no site is located on the boundary of the world $X$.
In particular $\min\{d(p_k,p_{j}): k\neq j\}>0$. 
Since the distance between $p$ and the boundary of the cell is positive, 
there is a small circle with center $p$ which is contained in the interior of the cell.  
On this circle we construct the projector (Figure \bref{fig:SubedgeSubwedge}), which, as a consequence, has a positive distance from the boundary  
of the cell and from $p$. Any other projector which may be used in the algorithm for producing the rays and  
detecting the vertices is simply a scalar multiplication of this small projector,  and it  
produces exactly the same rays (hence enables to detect the same vertices).

We say that a vertex $v$ of the cell corresponds to a subedge $F$ of the projector if 
$v$ is located inside the wedge corresponding to $F$. For instance, in Figure \bref{fig:projectorEuclidWedgePhase2} 
the vertex  $v$ corresponds to the subedge $F:=\{\theta_1,\theta_3\}$ and no other vertex  
corresponds to $F$. All the rays that we consider emanate from a given site $p=p_k$ and are 
sometimes identified with their direction vectors. Given a wedge generated by two rays, the 
rays on the boundary of the wedge are called boundary rays and any other ray in the wedge is called  
an intermediate ray. For instance, in Figure \bref{fig:projectorEuclidWedgePhase2},  if we look 
at the wedge generated by $\theta_1$ and $\theta_2$, then these rays are boundary rays and $\theta_3$ 
is an intermediate ray. Continuing with Figure \bref{fig:projectorEuclidWedgePhase2}, 
after additional steps, more rays will be generated between $\theta_1$ and $\theta_3$. These rays will be intermediate rays in the wedge generated by $\theta_1$ and $\theta_2$ 
(and also of the wedge generated by $\theta_2$ and $\theta_3$).  Between $\theta_1$ and $\theta_3$ no 
additional rays will be generated. We denote by $|A|$ the number of elements in the finite set $A$.

\begin{figure}[t]
\begin{minipage}[t]{1\textwidth}
\begin{center}
{\includegraphics[scale=0.76]{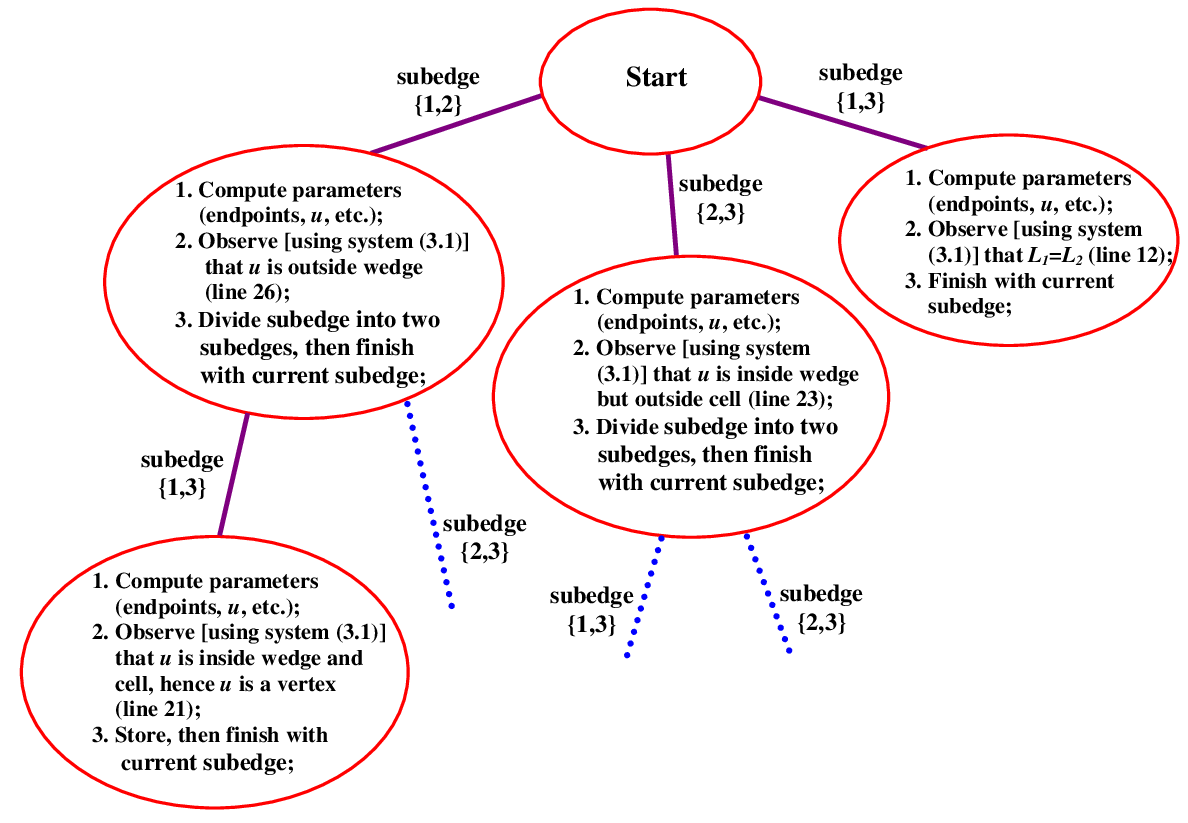}}
\end{center}
 \caption{Illustration of the algorithm tree for the cell of 
 Figure \bref{fig:projectorEuclidWedgePhase1}. }
\label{fig:AlgorithmTree}
\end{minipage}
\end{figure}

Part of the proof is to analyze the time complexity of the algorithm and, in particular,  
to prove that the algorithm terminates after finitely many steps. For doing this 
 it is convenient to consider the algorithm tree, 
namely the graph whose nodes are the main stages in the algorithm, and a node has subnodes (children nodes) 
whenever a subsedge is divided. More precisely, a node contains the various calculations done with respect to a given subedge until it is removed from $EdgeQueue$, among them computing the corresponding endpoints and doing arithmetic operations. Upon being removed, either the considered subedge is divided into two subedges (in lines \beqref{line:DivideParallelLines}, \beqref{line:DivideOutsideCell} or \beqref{line:DivideOutsideWedge}) with which the algorithm continues (these nodes are the children nodes of the current node), or nothing happens with it since either the algorithm detects no vertices in the corresponding wedge (line \beqref{line:SameLines}) or it detects a vertex (line \beqref{line:Vertex}) and then it stores it. Of course, once we are able to derive an upper bound on
 the number of nodes in the tree and on the number of operations done in each node, we have 
 an upper bound on the time complexity of the algorithm. Figure \bref{fig:AlgorithmTree} shows  an illustration of the algorithm tree related to the cell mentioned in Figure \bref{fig:projectorEuclidWedgePhase1}.

\subsection{Proof of the correction of the improvement of Method \bref{method:Endpoint}}\label{subsec:MethodEndpointImprovedIsCorrect} In this subsection we prove that Method \bref{method:EndpointImproved} is correct. More precisely, we prove the following proposition:

\begin{prop}\label{prop:MethodEndpointImprovedIsCorrect}
Method \bref{method:EndpointImproved} is correct, namely it outputs, after finitely many machine operations, the true endpoint. 
\end{prop}

The proof is based on the following lemmas. 
\begin{lem}\label{lem:Byy'Syy'}
Let $p,y\in X$, $p\neq y$. Let $B[y,d(p,y)]$ be the disk of radius $d(p,y)$ and 
center $y$. If $y'$ is in the segment $[p,y]$, then the disk $B[y',d(p,y')]$ is contained in 
$B[y,d(p,y)]$ and the integer rectangle $S_I[y',d(p,y')]$ is contained in 
the integer rectangle $S_I[y,d(p,y)]$.
\end{lem}
\begin{proof}
Let $x\in B[y',d(p,y')]$. Then $d(x,y')\leq d(p,y')$. Since $y'\in [p,y]$, we obtain from the triangle  inequality that 
\begin{equation*}
d(x,y)\leq d(x,y')+d(y',y)\leq d(p,y')+d(y',y)=d(p,y).  
\end{equation*}
Thus $x\in B[y,d(p,y)]$. Since $x$ was  arbitrary, the inclusion 
$B[y',d(p,y')]\subseteq B[y,d(p,y)]$ holds. 

In order to see the inclusion $S_I[y',d(p,y')]\subseteq S_I[y,d(p,y)]$, we observe that from \beqref{eq:S_I} one has $S_I[y,d(p,y)]=J_1\times J_2$ and $S_I[y',d(p,y')]=J'_1\times J'_2$, where
\begin{equation}\label{eq:J}
 J_j:=[s\lfloor (y_j-d(p,y))/s\rfloor,s(\lfloor (y_j+d(p,y))/s\rfloor+1)],  
\end{equation}
\begin{equation}\label{eq:J'}
 J'_j:=[s\lfloor (y'_j-d(p,y'))/s\rfloor,s(\lfloor (y'_j+d(p,y'))/s\rfloor+1)], 
\end{equation} 
 and where $j\in \{1,2\}$, $y=(y_1,y_2)$, and $y'=(y'_1,y'_2)$. 
The point $y'-d(p,y')(1,0)$ is in $B[y',d(p,y')]$, thus it is in $B[y,d(p,y)]$ from the previous paragraph. 
Hence $d(y'-d(p,y')(1,0),y)\leq d(p,y)$ and thus $|y_1-(y'_1-d(p,y'))|\leq d(p,y)$. Therefore  $y_1-d(p,y) \leq y'_1-d(p,y')$.  
Similarly  the points $y'+d(p,y')(1,0)$, $y'-d(p,y')(0,1)$, $y'+d(p,y')(0,1)$ are in  $B[y,d(p,y)]$  and the following inequalities hold 
\begin{equation*}
y'_1+d(p,y')\leq y_1+d(p,y),\quad y_2-d(p,y) \leq y'_2-d(p,y'), \quad y'_2+d(p,y')\leq  y_2+d(p,y). 
\end{equation*}
Since the floor function $\lfloor\cdot\rfloor$ is increasing on $\R$, the previous inequalities and \beqref{eq:J}-\beqref{eq:J'} imply the inclusion 
$S_I[y',d(p,y')]\subseteq S_I[y,d(p,y)]$. 
\end{proof}

\begin{lem}\label{lem:zCloserThan_y}
Consider Method  \bref{method:EndpointImproved}, and suppose that  $y$ is the $\ell$-th temporary endpoint for some $\ell\in\N$. More precisely, if $\ell=1$, then $y$ is the point obtained in either Step \beqref{step:2eta>d(y_x,p)} or Step \beqref{item:S_I[y,2 beta s]}. If $\ell>1$, then $y$ is the point obtained  after performing $\ell$ times  Step \beqref{item:InCellImproved}, assuming in Step \beqref{item:y_outside_cell} one goes to Step \beqref{item:InCellImproved}. If in Step \beqref{item:y_outside_cell} one goes to Step \beqref{item:CloseNeighborImproved}, then $y$ is the point obtained after performing $\ell-1$ times Step \beqref{item:y=u}.  Let $z$ be the $(\ell+1)$-th temporary endpoint (assuming the process has not been ended in the $\ell$-th temporary endpoint). Then $d(p,z)<d(p,y)$. In particular, if $\tilde{y}$ is the first temporary endpoint and $y$ is any temporary endpoint (first or not), then $y\in [p,\tilde{y}]$. 
\end{lem}
\begin{proof}
Since $\ell+1\geq 2$, it means that during the application of Method  \bref{method:EndpointImproved}, Step \beqref{item:y=u} was performed at least one time. But once we perform this step at least one time, the $\ell$-th  temporary endpoint will be the point called $w$ in Step \beqref{item:uImproved}, namely the point of intersection between our ray  and the bisector line $L$ between $p$ and some site $a$ (which we called ``$CloseNeighbor$'' in Step \beqref{item:CloseNeighborImproved}). However, $w$ is strictly inside $[p,y]$.  Indeed, since we arrived at Step \beqref{item:CloseNeighborImproved} before defining $z$, we know that $y$ is closer to $CloseNeighbor$ than to $p$. Thus $y$ is strictly inside  the half-plane of $CloseNeighbor$ (the half-plane whose boundary is $L$). Since $p$ is in the other half-plane and since our ray emanates from $p$ and arrives at $y$, it intersects the boundary of that half-plane, namely it intersects the bisector $L$ (at $w$)  before it arrives at $y$. Hence $w$ is in the half-open line segment $[p,y)$, and obviously $w\neq p$ because $w$ is on the bisector while $p$ is not. In other words, $w$ is indeed strictly inside $[p,y]$, as claimed. Because $z=w$, the assertion follows. Finally, if $\tilde{y}$ is the first temporary endpoint and $y$ is some temporary endpoint (first or not), then the previous lines and induction imply that $y$ belongs to  $[p,\tilde{y}]$.
\end{proof}

\begin{lem}\label{lem:SitesInLargeSquare}
Consider the ray emanating from the site $p$ in the direction of the unit vector $\theta$. 
\begin{enumerate}[(i)]
\item\label{item:d(a,y)>d(p,y)} Given an arbitrary point $y$ located on this ray, for each site $a\in X$ which is not in $S_I[y,d(p,y)]$ we have $d(y,p)<d(y,a)$.
\item\label{item:true_endpoint1} When applying Method  \bref{method:EndpointImproved} and performing distance comparisons between the current temporary endpoint $y$ and some sites, it is sufficient to consider only the sites in $S_I[y,d(p,y)]$ and not any other site in $X$. If the result of these distance comparisons is that $d(y,p)\leq d(y,a)$ for all $a\in S_I[y,d(p,y)]$, then $y$ is the true endpoint. 
\item\label{item:true_endpoint2} If $y$ is the current temporary endpoint and $S_I[y,d(p,y)]$ is empty of sites, then $y$ is the true endpoint. 
\end{enumerate}
\end{lem}

\begin{proof}
\begin{enumerate}[(i)]
\item Let $x\in B[y,d(p,y)]$. Then $|x-y|\leq d(p,y)$. Since $|z|_{\infty}\leq |z|$ for all $z\in \R^2$, where $|\cdot|$ is the Euclidean norm and $|\cdot|_{\infty}$ is $\ell_{\infty}$ norm, 
 the inequality $|x-y|_{\infty}\leq d(p,y)$ holds. Therefore $B[y,d(p,y)]\subseteq S[y,d(p,y)]$. The definition of the integer rectangle implies that   
 $S[y,d(p,y)]\subseteq S_I[y,d(p,y)]$. We conclude that 
 $B[y,d(p,y)]\subseteq S_I[y,d(p,y)]$, and hence, from the choice of $a$, we have  
  $a\notin B[y,d(p,y)]$. Hence $d(a,y)>d(p,y)$, as claimed. 
\item Suppose that $y$ is a temporary endpoint and we check, during Method  \bref{method:EndpointImproved},  whether $y$ is the true endpoint.  According to Part \beqref{item:d(a,y)>d(p,y)}, there is no need to perform any distance comparison between $d(p,y)$ and $d(y,a)$ whenever $a$ is a site located outside $S_I[y,d(p,y)]$, since the result of this comparison is known in advance, namely we know in advance that $d(p,y)<d(y,a)$ without doing the actual comparison. 

Suppose now that after performing the distance comparisons between $d(p,a)$ and $d(y,a)$ for every $a\in S_I[y,d(p,y)]$, we discover that $d(y,p)\leq d(y,a)$ for all $a\in S_I[y,d(p,y)]$. Since we already know from Part \beqref{item:d(a,y)>d(p,y)} that  $d(y,p)<d(y,a)$ whenever $a\notin S_I[y,d(p,y)]$, it follows that $y$ is in the Voronoi cell of $p$. In order to see that $y$ is the true endpoint, we need to show that the part of the ray beyond $y$ (in the direction of $\theta$) is outside the cell of $p$. According to Method  \bref{method:EndpointImproved}, either $y=y_X$, or $y$ is on the bisector line $L$ located between $p$ and some site $a$. In the first case  obviously the part of the ray beyond $y$ is outside the cell, since it is outside $X$ (and the Voronoi cell that we consider is restricted to $X$). In the second case the part of the ray beyond $y$ is strictly inside the half-plane of $a$ (namely, the half-plane determined by the bisector $L$), and hence any point $z$ on it satisfies $d(z,a)<d(z,p)$; thus $z$ cannot be in the Voronoi cell of $p$, as required. 
\item The explanation is similar to the one given in Part \beqref{item:true_endpoint1}. 
\end{enumerate}
\end{proof}

\begin{lem}\label{lem:SitesInBetaSquare}
Let $p\in \R^2$ and let $\theta$ be a unit vector. Let $\beta\in\N$, $s>0$, and  $y:=p+\eta\wt{\theta}$,  where $\wt{\eta}>\sqrt{2}(2\beta+1)s$. Then any $a\in S_I[y,2\beta s]$ satisfies $d(a,y)<d(p,y)$.
\end{lem}
\begin{proof}
We first recall that $|z|\leq \sqrt{2}|z|_{\infty}$ for all $z\in \R^2$; in addition, since $a\in S_I[y,2\beta s]$, we have $d(a,y)=|a-y|\leq \sqrt{2}|a-y|_{\infty}\leq \sqrt{2}(2\beta+1)s$; therefore, the  assumption of $\wt{\eta}$ and the fact that $y=p+\wt{\eta}\theta$ imply that 
$d(y,a)\leq \sqrt{2}|y-a|_{\infty}\leq \sqrt{2}\cdot(2\beta+1)s<\wt{\eta}=d(p,y)$, as required. 
\end{proof}
Now we can prove Proposition \bref{prop:MethodEndpointImprovedIsCorrect}. 
\begin{proof}[Proof of Proposition \bref{prop:MethodEndpointImprovedIsCorrect}]
First we need to show that Method \bref{method:EndpointImproved} terminates after finitely many machine operations. It can be observed that any of the steps in this method terminates after finitely many machine operations (e.g., in   Steps \beqref{item:S_I[y,2 beta s]} and \beqref{item:InCellImproved} this is immediate from the fact that there are finitely many sites). In particular, after finitely many machine operations we will arrive at Step \beqref{item:y_outside_cell}. Now we need to show that no infinite loop can occur, and once we show this, it will follow that Method \bref{method:EndpointImproved} does indeed terminate after finitely many machine operations. 

When we arrive at Step \beqref{item:y_outside_cell} at the first time, either the first temporary endpoint is some point $y$ located outside the cell of $p$ (we know that it is located outside the cell of $p$ since in previous steps we found some site $a\neq p$ which satisfies $d(y,a)<d(y,p)$), or we do not know whether $y$ is located outside the cell, but we know that $y=y_X$. In both cases we enter the loop which is composed of Steps \beqref{item:InCellImproved}--\beqref{item:y=u} (in the first case we enter this loop via Step \beqref{item:CloseNeighborImproved}) and, in particular, arrive at Step \beqref{item:InCellImproved}. 

In order to see that this loop terminates after finitely many machine operations, we observe that if, after performing Step \beqref{item:InCellImproved} at the first time, we find that the rectangular list of sites  is empty, then the loop terminates in Step \beqref{item:OutputImproved}. Otherwise, we make at least one distance comparison between $d(p,y)$ and $d(a,y)$ for some site $a\in S_I[y,d(p,y)]$ (probably many sites), and each time that we do such a distance comparison, we remove the considered site from the rectangular list of sites. Hence the number of sites in the rectangular list of sites strictly decreases. The same reasoning holds whenever we perform Step \beqref{item:InCellImproved}. Since the number of sites in the rectangular list of sites is finite (it is not larger than $n$, namely the total number of sites), we can perform Step \beqref{item:InCellImproved}  no more than this finite number before either the list becomes empty or we find that $d(p,y)\leq d(a,y)$ for every site $a\in S_I[y,d(p,y)]$. In both cases the loop terminates at Step \beqref{item:OutputImproved}, as required. 

Finally, it remains to see that the output of Method \bref{method:EndpointImproved} is correct, namely that it is the true endpoint (as defined after \beqref{eq:Tdef}). We already know from previous paragraphs that the method terminates after finitely many times in Step \beqref{item:OutputImproved}. Hence, we can use Lemma \bref{lem:SitesInLargeSquare} (either Part \beqref{item:true_endpoint1} or Part \beqref{item:true_endpoint2}, depending on whether, at the last time when Step \beqref{item:InCellImproved}  is performed, either $S_I[y,d(p,y)]$ is not empty of sites and $d(y,p)\leq d(y,a)$ for all $a\in S_I[y,d(p,y)]$, or whether  $S_I[y,d(p,y)]$ is empty of sites) to conclude that the last temporary endpoint $y$ is the true endpoint, as required. 

Apparently we finished the proof, but, as a matter of fact, there is one additional issue which is a little bit subtle (but simple) that should be addressed. The issue is why, whenever we perform Step \beqref{item:InCellImproved} and look for sites in $S_I[y,d(p,y)]$, as a subset of the rectangular list of sites, it is indeed true to regard $S_I[y,d(p,y)]$ as a subset of the rectangular list of sites. In other words, we need to show why the sites in the integer rectangle $S_I[y,d(p,y)]$, where $y$ is the temporary endpoint, do belong to the integer rectangle which induces the rectangular list of sites, namely to $S_I[\tilde{y},d(p,\tilde{y})]$, where $\tilde{y}$ is the first temporary endpoint. The reason is simple: Lemma \bref{lem:zCloserThan_y} ensures that $y$ belong to $[p,\tilde{y}]$, and hence, according to Lemma \bref{lem:Byy'Syy'}, we have $S_I[y,d(p,y)]\subseteq S_I[\tilde{y},d(p,\tilde{y})]$, as required. 
\end{proof}

\subsection{Proof of Theorem \bref{thm:CorrectnessOfTheAlg}: serial and parallel time complexity (without improvement)}\label{subsec:SerialParallelTimeComplexity}
In this section we discuss the proof of most of the parts of Theorem \bref{thm:CorrectnessOfTheAlg}, with the exception of Part \beqref{item:TimeComplexityUniformDistribution}. The proof is based on several lemmas. 
\begin{lem}\label{lem:EndpointEdge}
Any ray generated by the algorithm  hits an edge of the considered Voronoi cell. Moreover, the point hit by the ray is the true endpoint.
\end{lem}
\begin{proof}
Suppose that our ray, which starts at $p$, goes in the direction of some unit vector $\theta$. Since the  ray  starts at a point located inside the cell (namely, $p$), and goes outside the cell, possibly outside the region $X$  (because $X$ is bounded), the intermediate value theorem implies that this ray intersects the boundary of the cell. Since the boundary of the cell is composed of edges, the ray hits some edge, say $\wt{L}$. As a matter of fact, the point hit by the ray is $p+t\theta$, where $t=\sup\{s\in [0,\infty): p+s\theta\,\,\textnormal{in the cell}\}$. 

Denote by  $g$ the point hit by the ray. We claim that this point is the true endpoint. Indeed, 
by the definition of the endpoint (see \beqref{eq:Tdef} and the lines after it) the ray has an endpoint $p+T(p,\theta)\theta$. Since any point on the ray beyond $g$ is strictly outside the cell (either because it is outside the world $X$ or because it is on a halfspace of another site), it follows that $p+T(p,\theta)\theta\in [p,g]$, and so $[p,p+T(p,\theta)\theta]\subseteq [p,g]$. However, any point beyond $[p,p+T(p,\theta)\theta]$ is outside the cell, by the definition  of the endpoint. Since we know that $g$ is in the cell (it belongs to the edge $\wt{L}$),  we obtain that $g\in [p,p+T(p,\theta)\theta]$. Thus $[p,g]\subseteq [p,p+T(p,\theta)\theta]$ and hence $[p,g]=[p,p+T(p,\theta)\theta]$. Therefore $g=p+T(p,\theta)\theta$, i.e., $g$ is a true endpoint, as required. 
\end{proof}

\begin{lem}\label{lem:DifferentEdge}
Consider a projector subedge $F=\{\theta_1,\theta_2\}$ and the corresponding wedge 
generated from it. Let $L_1$ and $L_2$ be the lines on which the endpoints 
$p+T(p,\theta_i)\theta_i$, $i=1,2$ are located. 
Suppose that an intermediate ray in the direction of $\theta_3$ is 
generated when $F$ is considered. Then the ray of $\theta_3$ hits an edge 
different from  $\wt{L_i}$, $i=1,2$.
\end{lem}
\begin{proof} 
Lemma \bref{lem:EndpointEdge} ensures that  the ray of $\theta_3$ hits some edge, and the point hits by the ray is the true endpoint. 

Now we observe that an intermediate ray is created only when either $L_1$ and $L_2$ are parallel 
(line \bref{line:ParallelLines}),  or when they intersect in the wedge but outside the cell (line \bref{line:OutsideCell}), 
or when they  intersect outside the wedge (line \bref{line:OutsideWedge}). 

In the first case the ray of $\theta_3$ is located on a line which is parallel to $L_1$ 
 and $L_2$, and so the ray of $\theta_3$ cannot hit them either $L_1$ or $L_2$. Therefore $\wt{L_3}\neq \wt{L_i},i=1,2$. 

Now consider the second case and suppose to the contrary that, say,  $\wt{L_1}=\wt{L_3}$.  
Let $u:=L_1\cap L_2$.  By our assumption $u$ is outside the cell and  (see line \bref{line:OutsideCell})  it is hit by the ray of $\theta_3$. Thus $u\neq p+T(p,\theta_3)\theta_3$. 
But both $u$ and $p+T(p,\theta_3)\theta_3$ are assumed to belong to $\wt{L_1}$ and hence to $L_1$. 
Since a line is fully determined by two distinct points on it, and since $u$ and $p+T(p,\theta_3)\theta_3$ are two different points located on both $L_1$ and the ray of $\theta_3$, it follows that $L_1$ contains the whole ray of $\theta_3$. In particular, $L_1$ passes via $p$. This is impossible since $p$ 
is in the interior of the cell and hence its distance to any one of the boundary lines (including $L_1$) is positive. Thus $\wt{L_3}\neq \wt{L_1}$, and similarly, $\wt{L_3}\neq \wt{L_2}$. 

Now consider the third case and suppose to the contrary that, say,  $\wt{L_1}=\wt{L_3}$.  
Let $u:=L_1\cap L_2$.  By our assumption $u$ is outside the wedge generated by the rays of $\theta_1$ 
and $\theta_2$. Since $\theta_3$ is in the direction of $p-u$ (see line \bref{line:Theta3OutsideWedge}), its ray does not hit $u$ and hence $u\neq p+T(p,\theta_3)\theta_3$.  But both $u$ and $p+T(p,\theta_3)\theta_3$ are assumed to belong to $\wt{L_1}$ and hence to $L_1$. Therefore $L_1$ contains the segment $[u,p+T(p,\theta_3)\theta_3]$ and, in particular, it passes via $p$. This is impossible since $p$ 
is in the interior of the cell and hence its distance to any of the boundary lines (including $L_1$) is positive. Thus $\wt{L_3}\neq \wt{L_1}$, and similarly, $\wt{L_3}\neq \wt{L_2}$. 
\end{proof}

\begin{lem}\label{lem:IntermediateRay}
Consider two different rays generated by the algorithm, say in the direction of 
$\phi_1$ and $\phi_2$. Assume further that they are between, and possibly coincide with, two initial rays, i.e., between rays induced by two corners of the projector (thus, in particular, if the initial rays are in the directions of $\hat{\theta}_i$, $i\in\{1,2,3\}$, then the case where $\phi_1$ is between $\hat{\theta}_1$ and $\hat{\theta}_2$, and $\phi_2$ is between $\hat{\theta}_2$ and $\hat{\theta}_3$ is forbidden). Consider an intermediate ray between 
the rays of $\phi_i,\,i=1,2$ whose direction vector is $\theta_3$. 
Then the endpoint of the ray of $\theta_3$ and the endpoints of 
the rays of $\phi_1$ and $\phi_2$  must be located on different edges of the cell. 
\end{lem} 
\begin{proof}
Suppose to the contrary that the assertion is not true, say the endpoint of  
the ray of $\phi_1$ and the endpoint of 
the  ray of $\theta_3$ are located on the same edge $\wt{L}$. 
Because the ray of $\theta_3$ is an intermediate ray located between two initial rays, it is 
generated from some subedge $\{\theta_1,\theta_2\}$ of the projector. Here  
 $\theta_1$ is between $\phi_1$ and $\theta_3$ 
and possibly equals $\phi_1$, and $\theta_2$ is between $\theta_3$ and $\phi_2$ and possibly equals $\phi_2$. 
This generation corresponds to the routines of  lines \bref{line:ParallelLines}, \bref{line:OutsideWedge}, or \bref{line:OutsideCell} of the algorithm, and so in particular $\theta_1\neq \theta_3$ and $\theta_2\neq \theta_3$. 
Denote by $\wt{L_i}$, $i=1,2$ the edges on which the endpoints corresponding to $\theta_i$, $i=1,2$ are located.

Since the endpoint of the ray of $\phi_1$ and the endpoint of the ray of $\theta_3$ are located on the same edge $\wt{L}$, and since the ray  of $\theta_1$ is between these rays, the ray  of $\theta_1$  intersects $\wt{L}$. Denote by $g$ this  point of intersection. Then $g$ is the true endpoint of the ray  of $\theta_1$, as follows from Lemma \bref{lem:EndpointEdge}. See also Figure \bref{fig:IntermediateRay}. 
We conclude that the endpoints of the rays of $\theta_3$ and of $\theta_1$ are 
 located on the same edge. But this is impossible, since by our assumption $\theta_3$ 
 was created from the projector subedge $\{\theta_1,\theta_2\}$ and hence, according to Lemma \bref{lem:DifferentEdge}, the ray of $\theta_3$ hits an edge different from the ones 
  associated with $\theta_1$ and $\theta_2$.   We arrived at a contradiction which proves the assertion. 
\end{proof}

\begin{figure}
\begin{minipage}[t]{1\textwidth}
\begin{center}
{\includegraphics[scale=0.8]{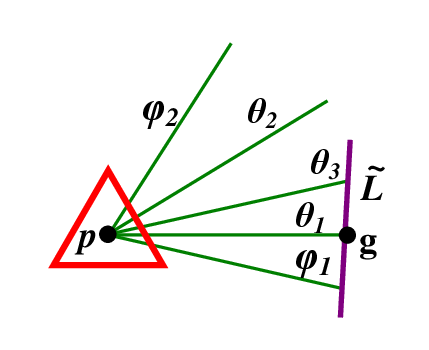}}
\end{center}
 \caption{An illustration of Lemma  \bref{lem:IntermediateRay}.}
\label{fig:IntermediateRay}
\end{minipage}
\end{figure}

\begin{lem}\label{lem:TerminateAlg}
Algorithm 1 terminates after a finite number of steps. In particular, the number 
of intermediate rays is finite. 
\end{lem}
\begin{proof}
This follows from Lemma \bref{lem:IntermediateRay}. Indeed, suppose to the contrary that 
the algorithm does not terminate after finitely many steps. This means that the list 
of projector subedges $EdgeQueue$ is never empty. But in each stage of the algorithm (each node), 
either a subedge is divided (lines \bref{line:ParallelLines}, \bref{line:OutsideCell}, and \bref{line:OutsideWedge}) 
and then  replaced by its children subedges and deleted from $EdgeQueue$, 
or it has no children subedges  (lines \bref{line:SameLines} and \bref{line:Vertex}) 
 and hence it is deleted from $EdgeQueue$ shortly 
after its creation. After a subedge is deleted from $EdgeQueue$ it is never created again. 
As a result, the fact that the algorithm does not terminate after finitely many steps 
implies that we can find an infinite sequence $(F_i)_{i=1}^{\infty}$ of nested distinct subedges. Since these subedges are nested, their rays are located between two initial rays.  Now, when some subedge $F_{i+1}$ is created during the algorithm, it is created from $F_i$ by dividing $F_i$ (in lines \bref{line:ParallelLines}, \bref{line:OutsideCell}, and \bref{line:OutsideWedge}) into two subedges, where one of them is $F_{i+1}$ and the other is some subedge which we are not interested in. In any case, one of the rays of $F_{i+1}$ is equal to one of the rays of $F_i$, and the other ray is an intermediate ray between the rays of $F_i$, and it is different from both rays. 

We claim that this intermediate ray is located on an edge of the cell which is different from any other edge of the cell detected so far, namely any other edge on which the endpoints of all the previous subedges $F_j$, $j\in \{1,\ldots, i\}$ are located. Indeed, assume that the claim is not true and denote the intermediate ray by $\phi'$. Then the endpoint of the ray of $\phi'$ is located on some edge $\tilde{L}$ on which two additional endpoints are located, say the ones corresponding to the rays of $\phi''$ and $\phi'''$. But $\phi''$ and $\phi'''$ belong to previous subedges $F_{j_1}$ and $F_{j_2}$ for some $j_1,j_2\in \{1,\ldots, i\}$,  while $\phi'$ belong to $F_{i+1}$, and so the fact that the subedges are nested implies that the ray of $\phi'$ is located between the rays of $\phi''$ and $\phi'''$. This contradicts Lemma \bref{lem:IntermediateRay} and proves what we claimed. In other words, we conclude that when one of the nested subedges is created, a new edge of the cell is detected. Thus infinitely many distinct edges are detected, contradicting the fact that 
each cell has only finitely many edges (the argument mentioned above actually implies that no more than $n$ new edges can be detected). Hence the algorithm does indeed terminate after finitely many steps  
and, in particular, finitely many intermediate rays are created. 
\end{proof}

\begin{lem}\label{lem:Ewedge}
Given a wedge generated by two distinct unit vectors $\theta_1$ and $\theta_2$ 
(located between initial rays, possibly equal to the initial rays), 
let $e_{\textnormal{IntRays}}$ be the number of edges of the cell hit by intermediate rays, i.e., 
rays produced by the algorithm inside the wedge excluding the boundary rays of the wedge. 
Then the tree of the restriction of the algorithm to the given wedge contains exactly $2e_{\textnormal{IntRays}}+1$ nodes. 
\end{lem}
\begin{proof}
The proof is by induction on $e_{\textnormal{IntRays}}$. By Lemma \bref{lem:TerminateAlg} we know that $e_{\textnormal{IntRays}}$ is finite. 
Let $\wt{L_i}$ be the edge on which the endpoint  $p+T(p,\theta_i)\theta_i$ is located, 
$i=1,2$. If $e_{\textnormal{IntRays}}=0$, then there are two cases. In the first case the rays generated by  
$\theta_i$, $i=1,2$ hit the same edge and in this case (line \bref{line:SameLines}) the current projector subedge 
is deleted from $EdgeQueue$. Thus no subnode of the current node is created, i.e., the restriction 
of the algorithm tree to the wedge contains $1=2e_{\textnormal{IntRays}}+1$ nodes. In the second case the rays 
must hit different edges which intersect at a vertex (lines \beqref{line:Vertex}--\beqref{line:Store}), since otherwise either no intersection 
occurs or the intersection is a point outside the cell or outside the wedge, and the corresponding 
ray in the direction of $\theta_3$ (lines \bref{line:ParallelLines}, \bref{line:OutsideCell}, 
 \bref{line:OutsideWedge}) hits an edge contained in the wedge, contradicting the assumption that $e_{\textnormal{IntRays}}=0$. 
Therefore also in this case the current subedge 
is deleted from $EdgeQueue$ and no subnode of the current node is created and the restriction 
of the algorithm tree to the wedge contains $2e_{\textnormal{IntRays}}+1=1$ nodes. 

Now assume that $e_{\textnormal{IntRays}}\geq 1$ and the claim holds for any nonnegative integer not exceeding $e_{\textnormal{IntRays}}-1$. We show below that the claim holds for $e_{\textnormal{IntRays}}$ too. Indeed, if the current projector subedge $\{\theta_1,\theta_2\}$, which induces our wedge,  is not divided into two subedges $\{\theta_1,\theta_3\}$ and $\{\theta_2,\theta_3\}$, then 
the restriction of the algorithm tree to the wedge generated by $\{\theta_1,\theta_2\}$ contains only one node. Hence there cannot be any edge different from $\wt{L_1}$ and $\wt{L_2}$ which is hit by rays produced by the algorithm,  contradicting the assumption that $e_{\textnormal{IntRays}}\geq 1$. Therefore $\{\theta_1,\theta_2\}$ is divided, i.e., the root node has two children. Consider the number of edges $e_{1,3}$ and $e_{2,3}$ of the cell contained in  
the wedges generated by $\{\theta_1,\theta_3\}$ and $\{\theta_2,\theta_3\}$ respectively, excluding, in each wedge, 
 the edges (one or two) hit by the boundary rays.  
Since $e_{\textnormal{IntRays}}=e_{1,3}+e_{2,3}+1$ (where the 1 comes from the edge hit by the ray in the direction of $\theta_3$), 
it follows that $e_{1,3}<e_{\textnormal{IntRays}}$ and $e_{2,3}<e_{\textnormal{IntRays}}$. Thus the induction hypothesis implies that 
the trees of the restriction of the algorithm to the wedges generated by 
the subedges $\{\theta_1,\theta_3\}$ and $\{\theta_2,\theta_3\}$ contain exactly 
$2e_{1,3}+1$ and $2e_{2,3}+1$ nodes, respectively. Hence the tree generated by the restriction 
of the algorithm to the original wedge (the one corresponding to $\{\theta_1,\theta_2\}$) 
contains $(1+2e_{1,3})+(1+2e_{2,3})+1=2e_{\textnormal{IntRays}}+1$ nodes, as claimed. 
\end{proof}

\begin{lem}\label{lem:e_k}
Given a cell of some site $p=p_k$, the restriction of the algorithm tree to this cell 
contains at most $2e_k$ nodes, where $e_k$ is the number of edges of the cell.
\end{lem}
\begin{proof}
The set of edges of the cell can be written as the union $I_{1,2}\cup I_{2,3}\cup I_{1,3}\cup I_0$, where $I_{1,2}$, $I_{2,3}$ and $I_{1,3}$ are the sets of edges hit by the intermediate rays belonging to the three  initial  wedges, and $I_0$ is the set of edges hit by the three initial rays, i.e., the rays shot in the direction of the corners of the projector. Lemma \bref{lem:IntermediateRay} ensures that $(I_{1,2}\cup I_{2,3}\cup I_{1,3})\cap I_0=\emptyset$. Moreover, below we prove that $I_{1,2}$, $I_{2,3}$ and $I_{1,3}$ are disjoint too. Once this is known, it follows that $e_k=|I_{1,2}|+|I_{2,3}|+|I_{1,3}|+|I_0|$. Since $|I_0|\geq 2$, we have 
\begin{equation}\label{eq:BoundRays}
|I_{1,2}|+|I_{2,3}|+|I_{1,3}|\leq e_k-2.
\end{equation}
However, the number of nodes in the algorithm tree restricted to the cell 
is the sum of nodes in each of the trees corresponding to the above-mentioned wedges, plus one additional root (Start) node, and by Lemma \bref{lem:Ewedge} we know that this number is $1+(2|I_{1,2}|+1)+(2|I_{2,3}|+1)+(2|I_{1,3}|+1)=2(|I_{1,2}|+|I_{2,3}|+|I_{1,3}|)+4$. Hence from \beqref{eq:BoundRays} it follows that the number of nodes in the algorithm tree restricted to the cell is at most $2e_k$, as required.

It remains to show that $I_{1,2}$, $I_{2,3}$ and $I_{1,3}$ are disjoint. Consider for instance $I_{1,2}$ and  $I_{2,3}$. We need to show that no edge $\wt{L}$ can be hit by an intermediate ray which belongs to $I_{1,2}$ and an intermediate ray which belongs to $I_{2,3}$. Suppose to the contrary that this happens. Let 
$\theta_{1,2}$ be the direction of the first intermediate ray which hits $\wt{L}$ and belonging to $I_{1,2}$, and let $\theta_{2,3}$ be the  direction of the second intermediate ray which hits $\wt{L}$ and belonging to $I_{2,3}$. Suppose that the first initial wedge is generated by $\widehat{\theta}_1$ and $\widehat{\theta}_2$ (these are the directions of the corners of the projector), and the second initial wedge is generated by $\widehat{\theta}_2$ and $\widehat{\theta}_3$. The ray of $\theta_{1,2}$ and $\theta_{2,3}$ cannot be opposite, otherwise $\wt{L}$ will be parallel  to itself (without coinciding with itself). Hence the linear wedge spanned by $\theta_{1,2}$ and $\theta_{2,3}$ (which is not a wedge generated by the algorithm: it is just a mathematical wedge needed for the proof) is not a half-plane. Since $\widehat{\theta}_2$ is strictly inside this linear wedge, we can write $\widehat{\theta}_2=\lambda_{1,2}\theta_{1,2}+\lambda_{2,3}\theta_{2,3}$ for 
some $\lambda_{1,2}>0$ and $\lambda_{2,3}>0$. The intersection of the edge $\wt{L}$ and this linear wedge is the segment $[p+T(p,\theta_{1,2})\theta_{1,2},p+T(p,\theta_{2,3})\theta_{2,3}]$. 

A direct calculation shows that  the ray of $\widehat{\theta}_2$ intersects $\wt{L}$ at a unique point $g$:  
this is done, for instance, by proving that there exist unique $t>0$ and $\alpha\in (0,1)$ satisfying 
\begin{equation}\label{eq:RayEdge}
p+t(\lambda_{1,2}\theta_{1,2}+\lambda_{2,3}\theta_{2,3})=
\alpha(p+T(p,\theta_{1,2})\theta_{1,2})+(1-\alpha)(p+T(p,\theta_{2,3})\theta_{2,3}),  
\end{equation}
e.g., by using the facts that $\theta_{1,2}$ and $\theta_{2,3}$ are linearly independent and $T(p,\theta_{1,2})>0$ and $T(p,\theta_{2,3})>0$, and then  equating the corresponding coefficients. Lemma \bref{lem:EndpointEdge} ensures that $g=p+T(p,\widehat{\theta}_2)\widehat{\theta}_2$. This contradicts Lemma  \bref{lem:IntermediateRay} which implies that the endpoint corresponding to $\widehat{\theta}_2$ is located on an edge different from the ones corresponding to $\theta_{1,2}$ and $\theta_{2,3}$, i.e., 
different from $\wt{L}$. This contradiction proves the assertion, namely that $\wt{L}$ cannot be hit by two intermediate rays belonging to two different initial wedges, as required. 
\end{proof}

\begin{lem}\label{lem:NodesO(n)}
Let $|X_e|$ be the number of edges of the polygon $X$. Consider the tree generated by the algorithm for the whole diagram. Then the number of nodes in this tree is bounded above by  $1+(12+2|X_e|)n=O(n)$. 
\end{lem}
\begin{proof}
The algorithm tree for the whole diagram is the union of the trees corresponding to 
each cell and an additional root node. By Lemma \bref{lem:e_k} the tree of the cell of $p_k$ has at most 
$2e_k$ nodes, where $e_k$ is the number of edges of the cell. The set of edges of each cell can be written as the disjoint union of two sets: the set of edges located on the boundary of the polygonal world $X$, and the set of edges located on a bisector between $p_k$ and another site $p_j$, $j\neq k$. 
Let $w_k$ be the size of the first set. Then $w_k$ is not greater than $|X_e|$, namely the number of edges of $X$. Thus  $\sum_{k=1}^n w_k \leq n|X_e|=O(n)$. Denote by $b_k$ the size of the second set. It is well known that $\sum_{k=1}^n b_k=O(n)$, in fact, $\sum_{k=1}^n b_k\leq 6n$. Indeed, from \cite[p. 347]{Aurenhammer}, \cite[pp. 173--175]{ORourke1994} we know that the number of edges in the Voronoi diagram is at most $3n$, and hence, since each edge is counted twice (it is shared by two cells), we have $\sum_{k=1}^n b_k\leq 6n$. (In \cite{Aurenhammer} it is assumed that the sites are in general position, namely no three sites are located on the same line, and no four sites are located on the same circle. However, if this is not true, then we can perturb the sites slightly so that their new configuration is a general position configuration; such a configuration actually enlarges the number of edges in the diagram, as explained in \cite{ORourke1994}. Hence $\sum_{k=1}^n b_k$ is bounded above by twice the number of edges in a general position configuration, which is bounded above by $6n$, as said before.)  Therefore 
\begin{equation}\label{eq:sum_e_k}
\sum_{k=1}^n e_k=\sum_{k=1}^n w_k+\sum_{k=1}^n b_k\leq (|X_e|+6)n=O(n), 
\end{equation}and hence the total number of nodes in the algorithm tree of the diagram is bounded above by  $1+\sum_{k=1}^n2e_k=1+(12+2|X_e|)n=O(n)$. 
\end{proof}

\begin{lem}\label{lem:BoundShotRays}
Let $\rho_k$ be the number of rays shot during the computation of the cell of the site $p_k$. 
 Then 
\begin{equation}\label{eq:NumberShotRays}
\sum_{k=1}^n\rho_k\leq (|X_e|+7)n.
\end{equation}
\end{lem}
\begin{proof}
We can write $\rho_k$ as the sum of the number of intermediate rays and the three initial rays. From the  proof of Lemma \bref{lem:e_k} we know that the first term is at most $e_k-2$ (see \beqref{eq:BoundRays}). 
Thus $\rho_k\leq e_k-2+3=e_k+1$. Hence 
$\sum_{k=1}^n\rho_k\leq n+\sum_{k=1}^n e_k$. The proof of Lemma \bref{lem:NodesO(n)}  (see \beqref{eq:sum_e_k}) implies that $\sum_{k=1}^n e_k\leq (|X_e|+6)n$. We conclude that 
$\sum_{k=1}^n\rho_k\leq (|X_e|+7)n$. 
\end{proof}

\begin{lem}\label{lem:EndpointO(n)}
The number of calculations needed to find the endpoint in some given direction, using Method \bref{method:Endpoint}, is bounded by a linear expression of $n$.  
\end{lem}
\begin{proof}
As in previous lemmas, one can build the algorithm tree for Method \bref{method:Endpoint} 
(see also Figure \bref{fig:projectorEuclideanRay}). Each node (with the exception of the initial node) is a place where it is checked whether the temporary endpoint $y$ is in the cell, using distance comparisons. If yes, then the algorithm terminates, and if not, then $y$ is further gets closer to the given site $p=p_k$. The obtained tree is linear, namely each node (with the exception of the last one) has exactly one child node. As explained in Remark \bref{rem:Array} above, each site $a$ is considered not more than one time for the distance comparison. It follows that the total number of distance comparisons is not greater than $n$.  Hence the number of nodes in the tree is $O(n)$.  Each calculation done in a given node is either 
an arithmetic operation, array manipulations, etc., i.e., it is $O(1)$, or it involves a distance comparison (each comparison is $O(1)$). Hence the total number of calculation is $O(n)$ and the assertion follows. 
\end{proof}

\begin{lem}\label{lem:ComplexityCell}
The time complexity for computing the cell of $p_k$ 
is bounded above  by $O(r_k e_k)$, where $r_k$ is the maximum 
number of distance comparisons done along each shot ray (compared between all shot rays), 
and $e_k$ is the number of edges of the cell. 
\end{lem}
\begin{proof}
Lemma \bref{lem:e_k} implies that the algorithm tree restricted to the cell 
of $p_k$ contains at most $2e_k=O(e_k)$ nodes. 
 The calculations done in each node are either calculations done when a ray is shot 
 (for computing its endpoint)  or  some $O(1)$ calculations, 
 when there is no need to compute a new endpoint (since the 
considered endpoints are already known) and only arithmetic operations, array manipulations, etc., are done. When an endpoint is computed, some operations are done along the corresponding ray. These operations are either distance comparisons or some related $O(1)$ operations 
(arithmetic operations, etc.: see the proof of Lemma  \bref{lem:EndpointO(n)}). 
Hence the number of operations is linear in the number of distance comparisons and hence 
the maximum number of these operations, compared between all the shot rays, is $O(r_k)$. 
The upper bound $O(r_k e_k)$ follows.
\end{proof}

\begin{lem}\label{lem:ParallelComplexity}
The time complexity, for the whole diagram, assuming $Q$ processing units are involved (independently) 
and processor $Q_i$ computes a set $A_i$ of cells, 
is $\max\{\sum_{k\in A_i}O(r_k e_k): i\in \{1,\ldots, Q\}\}$.
\end{lem}

\begin{proof}
This is a simple consequence of Lemma \bref{lem:ComplexityCell} because the processing 
units do not interact with each other and the cumulative time for computing a set of 
cells is the sum of times for computing each cell separately. 
\end{proof}

\begin{lem}\label{lem:Complexity}
The time complexity of the algorithm for the whole diagram, when one processing unit is involved, is bounded above by a quadratic expression in $n$.
\end{lem}
\begin{proof}
By Lemma \bref{lem:ComplexityCell} the time complexity of the algorithm for cell $k$ 
is bounded by $O(r_k e_k)$. Since $r_k\leq n$ (see Remark \bref{rem:Array}) and since $\sum_{k=1}^n e_k=O(n)$ 
(see \beqref{eq:sum_e_k}), the time complexity of the algorithm for the whole diagram is $\sum_{k=1}^n r_k e_k=n\sum_{k=1}^n e_k$, and the $O(n^2)$ upper bound follows. Alternatively, one can obtain the 
same bound by observing that the  number of nodes in the algorithm tree of the 
whole diagram is $O(n)$ (according to Lemma \bref{lem:NodesO(n)}), and the number of calculations done in each node is at most $O(n)$ (according to Lemma \bref{lem:EndpointO(n)}). 
\end{proof}

\subsection{The correctness of the output of the algorithm}
This subsection is devoted to the proof that the output of Algorithm 1 is correct, namely it comprises all  the vertices and edges of the considered Voronoi cell. 
\begin{lem}\label{lem:CorrectAlg}
Let $F=\{\theta_1,\theta_2\}$ be a subedge of the projector and let $p+T(p,\theta_i)\theta_i$, $i\in\{1,2\}$  be the corresponding endpoints of the rays shot in the directions of $\theta_1$ and $\theta_2$, respectively. 
\begin{enumerate}[(a)]
\item\label{lem:CorrectAlg:itemL} If both endpoints are on the same edge, then the only possible vertices corresponding to $F$ are the endpoints. 
\item\label{lem:CorrectAlg:itemLv} If one endpoint is on one edge and the other is on another one, and both edges intersect at some vertex $v$ of the cell, then the only possible vertices corresponding to $F$ are the endpoints and $v$. 
\end{enumerate}
\end{lem}

\begin{figure}
\begin{minipage}[t]{0.47\textwidth}
\begin{center}
{\includegraphics[scale=0.8]{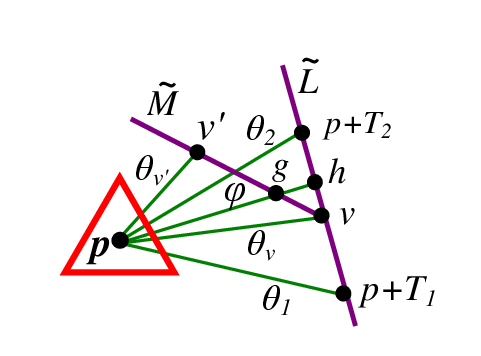}}
\end{center}
 \caption{An illustration of one of the cases described in Lemma \bref{lem:CorrectAlg}\beqref{lem:CorrectAlg:itemL}.}
\label{fig:Lem_v_prime_L_M}
\end{minipage}
\end{figure}

\begin{proof}
We first prove Part \beqref{lem:CorrectAlg:itemL}. This part seems quite obvious, but it turns out that a complete  proof  which takes into account all the details requires some work. For an illustration, see Figure \bref{fig:Lem_v_prime_L_M}. 
Let $T_i:=T(p,\theta_i)\theta_i,\,i=1,2$. By our assumption, both endpoints $p+T_i,\,\,i=1,2$ 
are on some edge $\wt{L}$ of the cell. 
The segment $[p+T_1,p+T_2]$ is contained in $\wt{L}$  since $\wt{L}$ is convex. Suppose to the contrary  that $v$ is a vertex of the cell that corresponds to $F$ and $v$ is not one of the endpoints $p+T_1$ or $p+T_2$. Then $v$ is located strictly inside the wedge generated by the endpoints, that is, $v=p+\lambda_1T_1+\lambda_2T_2$ for some $\lambda_1,\lambda_2\in (0,\infty)$. A simple calculation 
(similar to that of \beqref{eq:RayEdge}) shows that the ray  emanating from $p$  in direction $v-p$ intersects the segment $[p+T_1,p+T_2]$ at exactly one point $w$. Lemma \bref{lem:EndpointEdge} ensures that $w=v$. As a result, $v\in[p+T_1,p+T_2]\subseteq \wt{L}$. 

However, since $v$ is a vertex of the cell, it must belong to another edge $\wt{M}$. 
Let $v'\neq v$ be some point in  $\wt{M}$. Then $v'$ cannot be on the line $L$ on which $\wt{L}$ is located, since in this case $\wt{L}\cap\wt{M}$ will include a non-degenerate interval (the interval $[v,v']$), a contradiction to the fact that two different edges intersect at a point or do not intersect at all. We claim that $v'$ cannot be in the half-plane generated  by $L$ in which $p$ is located. Indeed, suppose to the contrary that this happens (see Figure \bref{fig:Lem_v_prime_L_M}). 

First note  that $v'$ is not on the ray which emanates from $v$ and passes via $p$; indeed, if, to the contrary, this happens, then either $p\in [v,v']$ or $v'\in [v,p]$, depending which point ($v'$ or $p$) is located on this ray at a greater distance  from $v$;  the first case is impossible because it means that $p\in [v,v']\subseteq \wt{M}$, a contradiction to the fact that $p$ is an interior point; the second case is impossible too since according  to a well known fact \cite[Theorem 6.1, p. 45]{Rockafellar1970} (this well-known fact says that if the underlying subset, namely the Voronoi cell in our case, is convex, and if we consider two points $g$ and $h$ in the subset, where $g$ is in the subset and $h$ is located in the interior  of the subset, then $(g,h]$ is contained in the interior of the subset), the half-open segment $(v,p]$ is contained in the interior of the cell because $p$  is in the interior of the cell and $v$ belongs to the cell;  thus $v'$, which is in $(v,p]$ (because $v\neq v'$ and we assume that $v'\in [v,p]$), is located in the interior of the cell, a contradiction to the fact that it is on $\wt{M}$ and hence  located on the boundary of the cell; therefore $v'$ is not located on the ray which emanates from $v$ and passes via $p$, as claimed. 

Consider now the wedge generated by the rays emanating from $p$ and passing via $v$ and $v'$. Denote $\theta_v:=(v-p)/|v-p|$ and $\theta_{v'}:=(v'-p)/|v'-p|$. We know from the previous paragraph that these two rays are different. A simple calculation (which is based on the fact that $v$ is in the  open segment $(p+T_1,p+T_2)$ and the assumption that $v'$ is located both in $\wt{M}$ and in the half-plane generated by $L$ in which $p$ is located) shows that any ray between $\theta_v$ and $\theta_{v'}$ which is close enough to $\theta_v$ (i.e., its generating unit vector $\phi$ is close enough to $\theta_v$) must intersect the  segment $[v',v)\subset\wt{M}$, and later it must intersect $\wt{L}$. See Figure \bref{fig:Lem_v_prime_L_M} for an illustration. However, once a ray which emanates from $p$ in the direction of $\phi$ intersects the edge $\wt{M}$, the point of intersection, which we denote by $g$, is the true endpoint of this ray, as follows from Lemma  \bref{lem:EndpointEdge}. Hence the point of intersection of the ray of $\phi$ with $\wt{L}$ (we  denote this point of intersection by $h$), which is located  beyond $g$, must be outside the cell. This is a contradiction to the fact that $h$ is on $\wt{L}$ and hence it is in the cell. 

The above-mentioned contradiction originated from our assumption that $v'$ is in the half-plane generated  by $L$ in which $p$ is located, and hence this case cannot happen. As a result, the only possibility for $v'$ is to be in the other half-plane generated by $L$, a contradiction to the fact that any point in this half-space is outside the cell. This contradiction shows that the initial assumption that the vertex $v$ is not  one of the endpoints $p+T_1$ or $p+T_2$ is false, as required. 

Now consider Part \beqref{lem:CorrectAlg:itemLv}. Let $\theta_3:=(v-p)/|v-p|$, and let $F_1:=\{\theta_1,\theta_3\}$ and $F_2:=\{\theta_2,\theta_3\}$. The possible vertices corresponding to $F$ are the ones corresponding to $F_1$ and $F_2$. According to out assumption, the endpoints $p+T_1$ and $p+T_3$ are located  on one edge, and the endpoints $p+T_2$ and $p+T_3$ are located on another edge.  By Part \beqref{lem:CorrectAlg:itemL} the only possible vertices  corresponding to $F_1$ are their endpoints $p+T_1$ and $p+T_3$, and the only possible vertices  corresponding to $F_2$ are their endpoints $p+T_2$ and $p+T_3$. Since actually $v=p+T_3$ (as follows from Lemma  \bref{lem:EndpointEdge}), we conclude that the only possible vertices  corresponding to $F$ are the vertex $v$ and the endpoints $p+T_1$ and $p+T_2$, as required. 
\end{proof}

\begin{lem}\label{lem:StoredVertices}
The point $u$ stored in line \bref{line:Store} of Algorithm 1 is a vertex of the cell.  
\end{lem}
\begin{proof}
This is evident, since $u$ is inside the cell and it is the intersection of two edges of the cell. 
\end{proof}

In the following lemmas we use the concept of a ``prime subedge'', namely  a subedge created by the algorithm 
at some stage but which has not been further divided after its creation. Because of Lemma \bref{lem:TerminateAlg} 
there are finitely many prime subedges. Any two such subedges either do not intersect or intersect at exactly 
one point (their corner), and their union  is the projector around $p$. See Figure \bref{fig:PrimeSubedges} 
for an illustration. 

\begin{figure}
\begin{minipage}[t]{1\textwidth}
\begin{center}
{\includegraphics[clip, scale=0.75]{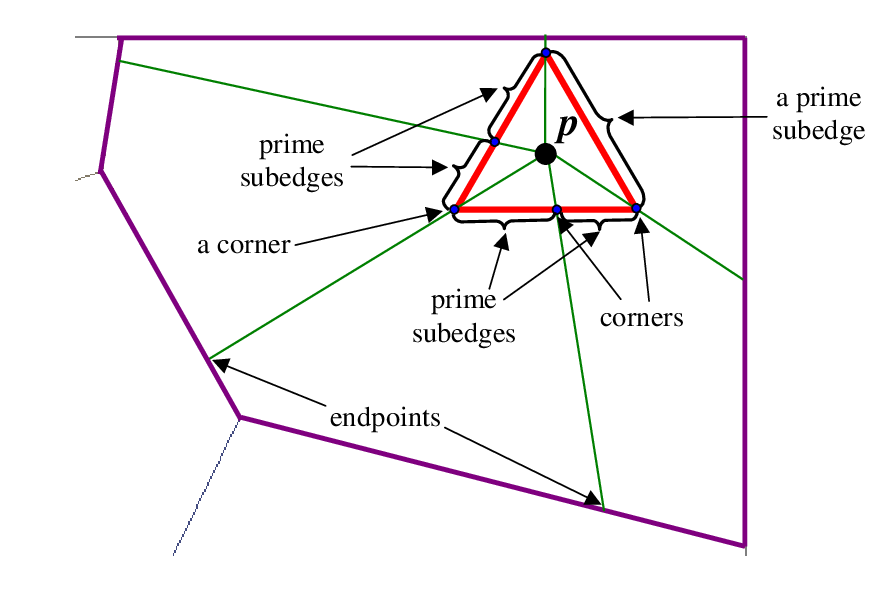}}
\end{center}
 \caption{Prime subedges and their associated rays.}
\label{fig:PrimeSubedges}
\end{minipage}
\end{figure}

\begin{lem}\label{lem:FoundVerticesEndpoint}
Let $v$ be a vertex of the cell and assume that it coincides with an endpoint corresponding to a corner of a prime subedge $F$. Then $v$, as a vertex, is found and stored by the algorithm. 
\end{lem}
\begin{proof}
The proof is not immediate as it may perhaps seem at first, since the algorithm is able to recognize 
a vertex  only after an intersection between two edges of the cell is detected, and so although the algorithm may find $v$ as an endpoint, it is still not known that $v$ is actually a vertex, and further calculations are needed in order to decide whether $v$ is a vertex or not. Such a problematic situation mainly occurs at the routine of line \bref{line:SameLines} (the endpoints 
of a given subedge are on the same edges), because even if one of the endpoints is a vertex, 
it is not stored. Hence this vertex must be detected (and stored) in another iteration of the algorithm, namely when it considers another subedge. 

Suppose to the contrary that $v$, as a vertex, is not found by the algorithm. 
Let $p+T_1$ and $p+T_2$ be the endpoints corresponding to the corners of $F$.  Let  $\wt{L_1}$, $\wt{L_2}$ be the corresponding edges of the cell on which these endpoints are located. The edges $\wt{L_i},\,i=1,2$ are located on corresponding lines $L_1,L_2$. By our assumption $v$ coincides with one of the endpoints, say with $p+T_1$. 
See Figure \bref{fig:VertexEndpoint}.

\begin{figure}
\begin{minipage}[t]{1\textwidth}
\begin{center}
{\includegraphics[scale=0.67]{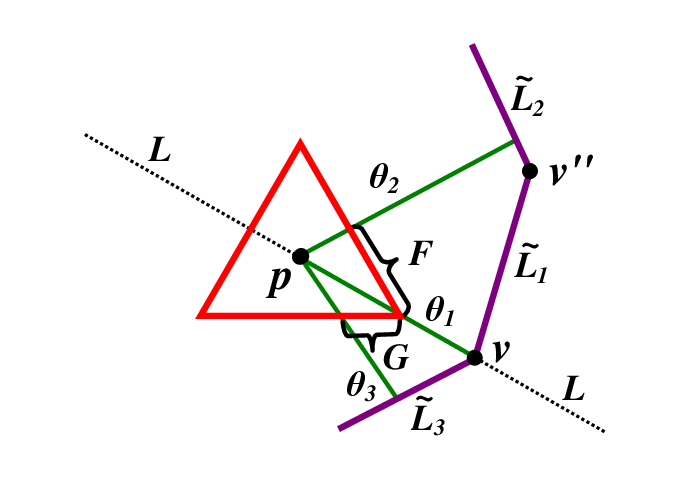}}
\end{center}
 \caption{Illustration of Lemma \bref{lem:FoundVerticesEndpoint}.}
\label{fig:VertexEndpoint}
\end{minipage}
\end{figure}

Since $p+T_1$ corresponds to a corner of $F$, this corner is located on another prime subedge $G$. Therefore $v$ is also in the wedge corresponding to $G$. Denote $T_3:=T(p,\theta_3)\theta_3$ and let $p+T_3$ be the other endpoint corresponding to $G$. Then  $p+T_3$ is located on some edge $\wt{L_3}$ of the  cell. 
Consider the lines $L_1$ and $L_3$ corresponding to $\wt{L_1}$ and $\wt{L_3}$, respectively. If they 
do not intersect,  or intersect inside the wedge corresponding to $G$ but outside the cell, or intersect at a point outside the wedge, then by the definition of Algorithm 1 (see lines \beqref{line:DivideParallelLines}, \beqref{line:DivideOutsideCell} and \beqref{line:DivideOutsideWedge}, respectively) it follows that $G$ must be divided, a contradiction to the assumption that $G$ is a prime subedge. 
Note also that the equality $\wt{L_3}=\wt{L_1}$ is impossible 
because it will imply that $v$ is in the interior of the edge $\wt{L_1}$ and this cannot happen to 
a vertex. Hence $L_1\neq L_3$ and, from previous lines, it  follows that $L_1$ and $L_3$ intersect at a point $v_{13}$ which is in the wedge and in the cell, i.e., $v_{13}$ is a vertex. Note that $v_{13}$ is found, as a vertex, during the running time of the algorithm at line \bref{line:Vertex}, when the algorithm  considers the subedge $G$ (it must consider it since it is a prime subedge). 

Our goal is to show that $v=v_{13}$. Once we show this, we arrive at a contradiction since we assumed that $v$, as a vertex, is not detected by the algorithm. A basic property of the polygonal boundary of our bounded cell is that any edge which belongs to it  intersects exactly two additional edges. In particular this is 
true for $\wt{L_1}$: one intersection occurs at the vertex $v=\wt{L}_1\cap \wt{L}_3$ and another one 
at another vertex $v''=\wt{L}_1\cap \wt{L}_2$. Consider the line $L$ passing via $p$ and $v$, namely the line which contains the ray in the direction of $\theta_1$: it separates the plane into two half-planes; one half-plane contains the ray $\theta_2$ and hence also $v''$ and $\wt{L}_1$; the other half-plane contains the ray $\theta_3$ and hence also the edge $\wt{L_3}$. Since we know that $\wt{L}_3$ intersects $\wt{L}_1$ (at $v_{13}$) and both edges are located at different half-planes, their intersection is located on the intersection of these half-planes, namely on $L$. But $L$ intersects $\wt{L}_1$ only at $v$, and so $v_{13}=\wt{L}_1\cap\wt{L}_3\subseteq L\cap \wt{L}_1=v$, that is, $v_{13}=v$, as required. 
\end{proof}

\begin{lem}\label{lem:FoundVertices}
The algorithm finds all the vertices and edges of the cell. 
\end{lem}
\begin{proof}
By Lemma \bref{lem:TerminateAlg} the algorithm terminates after a finite number of steps. Let $(F_j)_{j=1}^m$ be the finite list of all prime subedges, $m\in\N$. Suppose to the contrary that some vertex $u$ is not found.  Then $u$ corresponds to some point located on some prime subedge  $F$ of the projector, since the ray in the direction of $u-p$ intersects the projector at exactly one point, and each point on the projector belongs to some prime subedge.

Let $p+T_1$ and $p+T_2$ be the endpoints corresponding to the corners of $F$, and let $\wt{L_1}$, $\wt{L_2}$ be the corresponding edges of the cell on which these endpoints are located. The edges $\wt{L_i},\,i=1,2$ are located on corresponding lines $L_1,L_2$. Let $B$ the matrix from \beqref{eq:B_lambda}. 

Assume first that $\det(B)=0$. Then it must be that $L_1=L_2$, since otherwise $L_1$ and $L_1$ are parallel and hence $F$ is divided into two subedges (line \bref{line:ParallelLines}), a contradiction to the  assumption that $F$ is a prime subedge. However, if $L_1=L_2$, then  $\wt{L_1}=\wt{L_2}$, and hence, by Lemma 
\bref{lem:CorrectAlg}\beqref{lem:CorrectAlg:itemL}, the only possible vertices corresponding to $F$ are the endpoints $p+T_1,\,p+T_2$. In particular, the missing vertex $u$ coincides with one of these endpoints. But then, according to Lemma \bref{lem:FoundVerticesEndpoint}, the algorithm finds $u$ as a vertex, a contradiction to what we assumed regarding $u$ in the beginning of the proof.

As a result, the case $\det(B)=0$ cannot occur. Assume now the case $\det(B)\neq 0$. Then $L_1\neq L_2$. It must be that $\lambda$ from \beqref{eq:B_lambda} is nonnegative, since otherwise $F$ is  divided into two subedges (line \bref{line:OutsideWedge}), a contradiction to the  assumption that $F$ is a prime subedge. The point of intersection between $L_1$ and $L_2$ is $v=p+\lambda_1 T_1+\lambda_2 T_2$, and it must be in the wedge corresponding to $F$ and also in the cell of $p$, since otherwise $F$ is divided into two subedges (line \bref{line:OutsideCell}), a contradiction to the  assumption that $F$ is a prime subedge. Hence $v$ is a vertex corresponding to $F$ and it is found by the algorithm at the stage when  \beqref{eq:B_lambda} is considered. If $v=u$, then the algorithm finds $u$ when it considers $F$, a contradiction to what we assumed regarding $u$ in the beginning of the proof. Hence $v\neq u$, and by Lemma \bref{lem:CorrectAlg}\beqref{lem:CorrectAlg:itemLv} it must be that $u$ coincides with one of the endpoints $p+T_1$ or $p+T_2$. But then, according to Lemma \bref{lem:FoundVerticesEndpoint}, the algorithm finds $u$ as a vertex, a contradiction to what we assumed regarding $u$ in the beginning of the proof.

Therefore all the vertices of the cell are detected. As explained in Method   \bref{method:Endpoint} and 
Section  \bref{sec:CombinatorialInformation}, when a vertex is detected, also the edges which intersect 
at it are detected. Since all the possible vertices  are found by the algorithm, then so are all 
the possible edges. 
\end{proof}

\begin{proof}[Proof of Theorem \bref{thm:CorrectnessOfTheAlg} with the exception of Part \beqref{item:TimeComplexityUniformDistribution}]
This is a simple consequence of Lemmas  \bref{lem:TerminateAlg}, \bref{lem:ComplexityCell}, \bref{lem:ParallelComplexity}, \bref{lem:Complexity}, 
\bref{lem:StoredVertices}, and \bref{lem:FoundVertices}.
\end{proof}

\subsection{Proof of Theorem \bref{thm:CorrectnessOfTheAlg}\beqref{item:TimeComplexityUniformDistribution} }\label{subsec:Improvements} 
In this subsection a proof of Theorem \bref{thm:CorrectnessOfTheAlg} is presented. We use the terminology and notation mentioned in the formulation of the theorem. In particular, unless otherwise stated, $e$  means the base of the natural logarithm, namely $e\approx 2.71828\ldots$.  The proof is  based on several lemmas. Before formulating them, it is worth clarifying the phrase 
``the sites are distributed independently according to the  uniform distribution''. Here the probability space is $X$ (the sample space) together with  the standard normalized Lebesgue  measure which assigns to each region (namely, to each Lebesgue measurable subset of $X$) its area divided by the area of  $X$. The location of each site is a random vector (from $X$ to itself) and the probability for a site to  be in some region of area $\sigma$ is the ratio between $\sigma$ and the area of $X$. All of  these random vectors are assumed to be independent.

We start by formulating the following two known lemmas which will be used later. 
The proof of the second one is given for the sake of completeness. 

\begin{lem}\label{lem:StirlingRobbins} {\bf (An improvement of Stirling's formula by Robbins \cite{Robbins1955jour}, \cite[pp. 52--53]{Feller1957book})} 
Given a natural number $\ell$, the following equality holds:
\begin{equation}\label{eq:StirlingRobbins}
\ell!=\sqrt{2\pi \ell}\left(\frac{\ell}{e}\right)^{\ell}e^{r_{\ell}},
\end{equation}
where $r_{\ell}$ is a positive number satisfying $1/(12\ell+1)<r_{\ell}<1/(12\ell)$. 
In particular, $\ell!\geq \sqrt{2\pi \ell}(\ell/e)^{\ell}$. 
\end{lem}

\begin{lem}\label{lem:er}
Let $r\geq 1$. Then 
\begin{equation}\label{eq:ee-1}
\left(1-\frac{1}{r}\right)^r\leq e^{-1}\quad\textnormal{and}\quad  \left(1+\frac{1}{r}\right)^r\leq e.
\end{equation}
\end{lem}
\begin{proof}
The assertion is obvious for $r=1$. For all $r>1$ let 
\begin{equation*}
f(r):=\left(1-\frac{1}{r}\right)^r, \quad g(r):=\ln(f(r))=r\ln\left(1-\frac{1}{r}\right). 
\end{equation*}
Then 
\begin{equation*}
g'(r)=\ln\left(1-\frac{1}{r}\right)+\frac{r}{1-\frac{1}{r}}\cdot\frac{1}{r^2}=\ln\left(1-\frac{1}{r}\right)+\frac{1}{r-1},\quad, g''(r)=\frac{1}{r-1}\left( \frac{1}{r}-\frac{1}{r-1}\right).
\end{equation*}
Therefore $g''(r)<0$ for all $r>1$ and hence $g'$ is strictly decreasing on $(1,\infty)$. 
Since $\lim_{r\to\infty}g'(r)=0$, elementary calculus shows that $g'(r)>0$ for all $r>1$. It follows that $g$ is strictly increasing on $(1,\infty)$, and hence so is $f$. This implies the assertion because of the well known limit $\lim_{r\to\infty} f(r)=e^{-1}$. The second inequality in \beqref{eq:ee-1} can be proved in a similar way.
\end{proof}

The following four lemmas contain the core of the proof of Theorem   \bref{thm:CorrectnessOfTheAlg}\beqref{item:TimeComplexityUniformDistribution}.

\begin{lem}\label{lem:TooMuchInBox}
Let  $\alpha,\gamma\in \N$ be given. Assume that $\gamma>1$ and 
\begin{equation}\label{eq:alpha-kappa-m}
100<m,\quad 3\alpha<\lfloor 0.2m\rfloor+2, \quad \alpha\gamma<n. 
\end{equation}
Denote by $E_V$ the event in which a subset  
$V\subseteq X$, which is the union  of $\alpha$ buckets, contains at least $\gamma\alpha$ sites, and denote by $P(E_V)$  the probability of this event. Then 
\begin{equation}\label{eq:V}
P(E_V)\leq \frac{e^{4-0.8\alpha}}{\sqrt{\gamma\alpha}}
\left(\frac{e}{\gamma}\right)^{\gamma\alpha}
+\frac{e^{4+0.2\alpha}}{\sqrt{(\lfloor 0.2m\rfloor+1)}}\left(\frac{e\alpha}{\lfloor 0.2m\rfloor+1}\right)^{\lfloor 0.2m\rfloor+1}.
\end{equation}
If, in addition, $\sqrt{\alpha}\in \N$  and $\sqrt{\alpha}\leq\min\{m_1,m_2\}$, then consider the set of all squares included in $X$ and consisting of $\alpha$  buckets. Denote by $G_V$ the event in which  at least one of these  squares contains at least $\gamma\alpha$ sites. Then 
\begin{equation}\label{eq:Vn}
P(G_V)\leq n\left(\frac{e^{4-0.8\alpha}}{\sqrt{\gamma\alpha}}
\left(\frac{e}{\gamma}\right)^{\gamma\alpha}
+\frac{e^{4+0.2\alpha}}{\sqrt{(\lfloor 0.2m\rfloor+1)}}\left(\frac{e\alpha}{\lfloor 0.2m\rfloor+1}\right)^{\lfloor 0.2m\rfloor+1}\right).
\end{equation}
\end{lem}

\begin{proof}
Let $j\in [\gamma\alpha,n]$ be a given natural number. The probability that some site 
 $a$ will be in the region $V$ is $\alpha/m$ because the sites are distributed according to the 
 uniform distribution and the ratio between the area of $V$ to the area of $X$ is  $\alpha/m$ (since $V$ is composed of $\alpha$ buckets and $X$ is composed of $m$ buckets). Since the locations of the sites are independent random vectors, the probability of the event  that certain $j$ different sites  $a_{i_1},\ldots, a_{i_j}$  will be in $V$ 
 and the other $n-j$ sites will be in $X\backslash V$ is $(\alpha/m)^j(1-(\alpha/m))^{n-j}$.  There are  $\binom{n}{j}$ events of this kind and all of them are disjoint, so the 
 probability of the event that exactly $j$ sites are in $V$ is $\binom{n}{j}(\alpha/m)^j(1-(\alpha/m))^{n-j}$. 
 The event $E_V$ in which at least $\gamma\alpha$ sites are in $V$ is the disjoint union of the 
 events in which exactly $j$ sites are in $V$, where $j$ runs from $\gamma\alpha$ to $n$. We conclude that 
\begin{equation}\label{eq:E_V}
P(E_V)=\sum_{j=\gamma\alpha}^n\binom{n}{j}\left(\frac{\alpha}{m}\right)^j\left(1-\frac{\alpha}{m}\right)^{n-j}. 
\end{equation} 
The goal now is to estimate $P(E_V)$ from above.  First we observe that 
\begin{equation}\label{eq:nm}
n\leq m+2\sqrt{m}. 
\end{equation}
Indeed, this follows from the facts that $\sqrt{m}=\lfloor \sqrt{n}\rfloor$ is an integer and $n=m+2\sqrt{m}r+r^2<m+2\sqrt{m}+1$ for some $r\in [0,1)$. Now let $j\in [\gamma\alpha,n]\cap \N$ 
be given. Either $j\leq 2\sqrt{m}$ or $j>2\sqrt{m}$. In the first case we have 
\begin{multline}\label{eq:m/m}
\frac{(m+2\sqrt{m})}{m}\frac{(m+2\sqrt{m}-1)}{m}\ldots\frac{(m+2\sqrt{m}-(j-1))}{m}\\
\leq 
\frac{(m+2\sqrt{m})}{m}\frac{(m+2\sqrt{m}-1)}{m}\ldots\frac{(m+1)}{m}.
\end{multline}
In the second case \beqref{eq:m/m} holds too, because the following expression 
\begin{equation}\label{eq:factor}
\frac{m}{m}\cdot\frac{(m-1)}{m}\cdot\ldots \cdot\frac{(m+2\sqrt{m}-(j-1))}{m} 
\end{equation}
 appears as a factor in the expression 
\begin{equation}\label{eq:(m+2sqrt(m))/m}
\frac{(m+2\sqrt{m})}{m}\frac{(m+2\sqrt{m}-1)}{m}\ldots \frac{(m+2\sqrt{m}-(j-1))}{m},
\end{equation}
in addition to the expression  
\begin{equation*}
\frac{(m+2\sqrt{m})}{m}\frac{(m+2\sqrt{m}-1)}{m}\ldots \frac{(m+1)}{m}, 
\end{equation*}
which appears in \beqref{eq:(m+2sqrt(m))/m} as well, and the expression in  
\beqref{eq:factor} is a real number in $(0,1]$ (this follows from the inequality $j\leq n$ and from \beqref{eq:nm}). It follows from \beqref{eq:nm}, \beqref{eq:m/m}, the inequality $\alpha\leq m$, Lemma \bref{lem:er} and the inequality $m\leq n$,  that 
\begin{multline}\label{eq:BinomEstimate}
\binom{n}{j}\left(\frac{\alpha}{m}\right)^j\left(1-\frac{\alpha}{m}\right)^{n-j}
=\frac{n(n-1)\ldots(n-(j-1))}{m^j}\frac{\alpha^j}{j!}
\left(1-\frac{\alpha}{m}\right)^{(m/\alpha)\alpha(n-j)/m}\\
\leq \frac{(m+2\sqrt{m})}{m}\frac{(m+2\sqrt{m}-1)}{m}\ldots\frac{(m+1)}{m}
\frac{\alpha^j}{j!}
e^{\alpha(j-n)/m}\\
\leq \left(1+\frac{2}{\sqrt{m}}\right)^{2\sqrt{m}}\frac{\alpha^j}{j!}e^{-\alpha n}e^{\alpha j/m}\leq \left(\left(1+\frac{2}{\sqrt{m}}\right)^{\sqrt{m}/2}\right)^4\frac{\alpha^j}{j!}e^{-\alpha}e^{\alpha j/m}\\
\leq e^{4}\frac{\alpha^j}{j!}e^{-\alpha}e^{\alpha j/m}.
\end{multline} 
As long as $j\leq \lfloor 0.2m\rfloor$ we have $e^{-\alpha}e^{\alpha j/m}\leq e^{-0.8\alpha}$. 
For larger $j$ we use the trivial estimate $j\leq n$, \beqref{eq:nm},  and \beqref{eq:alpha-kappa-m} (the fact that $m>100$) to conclude that  
\begin{equation}\label{eq:0.2}
e^{-\alpha}e^{\alpha j/m}\leq e^{\alpha(-1+1+(2/\sqrt{m}))}\leq e^{0.2\alpha}. 
\end{equation}
From  \beqref{eq:alpha-kappa-m}, \beqref{eq:BinomEstimate}, \beqref{eq:0.2} and simple arithmetic,  it follows that  
\begin{multline*}
\sum_{j=\gamma\alpha}^n\binom{n}{j}\left(\frac{\alpha}{m}\right)^j\left(\frac{m-\alpha}{m}\right)^{n-j}
\leq \left(e^{4-0.8\alpha}\sum_{j=\gamma\alpha}^{\lfloor 0.2m\rfloor}\frac{\alpha^j}{j!}\right)+
\left(e^{4+0.2\alpha}\sum_{j=\lfloor 0.2m\rfloor+1}^n\frac{\alpha^j}{j!}\right)\\
\leq e^{4-0.8\alpha}\frac{\alpha^{\gamma\alpha}}{(\gamma\alpha)!}\left(1+
\sum_{i=1}^{\infty }\frac{\alpha^i}{(\gamma\alpha+1)\ldots (\gamma\alpha+i)}\right)+\\
e^{4+0.2\alpha}\frac{\alpha^{(\lfloor 0.2m\rfloor+1)}}{(\lfloor 0.2m\rfloor+1)!}\left(1+\sum_{i=1}^{\infty}\frac{\alpha^i}{(\lfloor 0.2m\rfloor+2)\ldots(\lfloor 0.2m\rfloor+1+i)}\right)\\
\leq e^{4-0.8\alpha}\frac{\alpha^{\gamma\alpha}}{(\gamma\alpha)!}
\left(1+\sum_{i=1}^{\infty }\left(\frac{\alpha}{\gamma\alpha+1}\right)^i\right)+
e^{4+0.2\alpha}\frac{\alpha^{(\lfloor 0.2m\rfloor+1)}}{(\lfloor 0.2m\rfloor+1)!}\left(1+\sum_{i=1}^{\infty}\left(\frac{\alpha}{\lfloor 0.2m\rfloor+2}\right)^i\right)\\
=\frac{e^{4-0.8\alpha}\alpha^{\gamma\alpha}}{(\gamma\alpha)!}\left(\frac{1}{1-\frac{\alpha}{\gamma\alpha+1}}\right)+\frac{e^{4+0.2\alpha}\alpha^{(\lfloor 0.2m\rfloor+1)}}{(\lfloor 0.2m\rfloor+1)!}\left(\frac{1}{1-\frac{\alpha}{\lfloor 0.2m\rfloor+2}}\right),
\end{multline*}
where the first term in the right-hand side of the first inequality is 0 if $\gamma\alpha>\lfloor 0.2m\rfloor$. This inequality, together with Lemma \bref{lem:StirlingRobbins}, \beqref{eq:alpha-kappa-m},  the assumption that $\gamma>1$ (an assumption which implies that $\alpha/(\gamma\alpha+1)<1/2$) and simple arithmetic, all imply that  
\begin{multline*}
\sum_{j=\gamma\alpha}^n\binom{n}{j}\left(\frac{\alpha}{m}\right)^j\left(1-\frac{\alpha}{m}\right)^{n-j}\\
\leq \frac{2e^{4-0.8\alpha}\alpha^{\gamma\alpha}}{\sqrt{2\pi\gamma\alpha}}
\left(\frac{e}{\gamma\alpha}\right)^{\gamma\alpha}
+
\frac{2e^{4+0.2\alpha}\alpha^{(\lfloor 0.2m\rfloor+1)}}{\sqrt{2\pi (\lfloor 0.2m\rfloor+1)}}\left(\frac{e}{\lfloor 0.2m\rfloor+1}\right)^{\lfloor 0.2m\rfloor+1}
\\
\leq \frac{e^{4-0.8\alpha}}{\sqrt{\gamma\alpha}}
\left(\frac{e}{\gamma}\right)^{\gamma\alpha}
+
\frac{e^{4+0.2\alpha}}{\sqrt{(\lfloor 0.2m\rfloor+1)}}\left(\frac{e\alpha}{\lfloor 0.2m\rfloor+1}\right)^{\lfloor 0.2m\rfloor+1}.
\end{multline*}
This inequality and \beqref{eq:E_V} prove \beqref{eq:V}. It remains to show \beqref{eq:Vn} under the assumptions that $\sqrt{\alpha}$ is a natural number which satisfies $\sqrt{\alpha}\leq\min\{m_1,m_2\}$. Since the side length of a square $V_i$ consisting of $\alpha$ buckets is $\sqrt{\alpha}$ and since $X$ is a rectangle consisting of $m_1\times m_2$ buckets, simple combinatorics implies that the number $I_{\textnormal{sqrs}}$ of these squares is $(1+m_1-\sqrt{\alpha})(1+m_2-\sqrt{\alpha})$ (this is a natural number from the assumption on $\alpha$), and hence it is bounded above by $m_1m_2=m\leq n$. The probability that at least 
one of these sqaures contains at least $\gamma\alpha$ sites is not greater than 
$\sum_{i=1}^{I_{\textnormal{sqrs}}} P(E_{V_i})$. Since \beqref{eq:V} implies that $\sum_{i=1}^{I_{\textnormal{sqrs}}}P(E_{V_i})$  is bounded above by $n$ times the expression on the right-hand side of \beqref{eq:V}, we conclude that \beqref{eq:Vn} holds. 
\end{proof}

\begin{lem}\label{lem:EmptySquare}
In the notation and assumptions of Lemma \bref{lem:TooMuchInBox}, let $E$ be the event in which  at least one square which is composed of $\alpha\in \N$ buckets is empty of sites, where $\sqrt{\alpha}$ is an integer satisfying $\sqrt{\alpha}\leq\min\{m_1,m_2\}$. Then  the probability of $E$ is not greater than $ne^{-\alpha}$.
\end{lem}
\begin{proof}
Let $V$ be a square which is composed of $\alpha$ buckets. The probability that some site $a$ is not in $V$ is $1-(\alpha/m)$,  because the sites are distributed according to the uniform distribution. Denote by $E_{V,\textnormal{empty}}$ the probability that $V$ is empty of sites. Since the sites are independent, we have $P(E_{V,\textnormal{empty}})=(1-(\alpha/m))^n$. From Lemma \bref{lem:er} 
and the inequality $0<\alpha<m\leq n$, we have 
\begin{equation}\label{eq:P(E_V,empty)}
P(E_{V,\textnormal{empty}})=\left(1-\frac{\alpha}{m}\right)^n=\left(1-\frac{\alpha}{m}\right)^{(m/\alpha)\alpha n/m}
\leq e^{-\alpha n/m}\leq e^{-\alpha}.
\end{equation} 
Consider the set of  squares $V_i\subseteq X$ which are composed of $\alpha$ buckets. As explained at the end of the proof of Lemma \bref{lem:TooMuchInBox}, the number of these squares is $I_{\textnormal{sqrs}}:=(1+m_1-\sqrt{\alpha})(1+m_2-\sqrt{\alpha})$ and it is bounded above by $m_1m_2=m\leq n$. Since $P(E)=P(\cup_{i=1}^{I_{\textnormal{sqrs}}}E_{V_i,\textnormal{empty}})\leq \sum_{i=1}^{I_{\textnormal{sqrs}}}P(E_{V_i,\textnormal{empty}})$, it follows from \beqref{eq:P(E_V,empty)} that $P(E)\leq ne^{-\alpha}$, as claimed. 
\end{proof}

\begin{lem}\label{lem:EventEndpoint}
In the setting of  Lemma \bref{lem:TooMuchInBox}, let $\beta\in \N$ be given and assume that $\beta\leq(\min\{m_1,m_2\}-1)/2$.  Let $E$ be the event in which each  square which is contained in $X$ and composed of $\alpha_2:=(1+\lfloor 8\sqrt{2}(\beta+1.01)\rfloor)^2$ buckets contains at most $\gamma \alpha_2-1$  
sites. Given $\gamma\in \N$, let $G$ be the event in which at most 
$2(\gamma\alpha_2-1)$ distance comparisons are made in the computation of each endpoint along each ray shot from each site $p_k$, $k\in K$, using Method \bref{method:EndpointImproved}. Then 
\begin{multline}\label{eq:P(F)P(E)}
P(G)\geq P(E)\geq \\
1-n\left(\frac{e^{4-0.8\alpha_2}}{\sqrt{\gamma\alpha_2}}
\left(\frac{e}{\gamma}\right)^{\gamma\alpha_2}
+\frac{e^{4+0.2\alpha_2}}{\sqrt{(\lfloor 0.2m\rfloor+1)}}\left(\frac{e\alpha_2}{\lfloor 0.2m\rfloor+1}\right)^{\lfloor 0.2m\rfloor+1}\right). 
\end{multline}
\end{lem}
\begin{proof}
 The right-most inequality in \beqref{eq:P(F)P(E)} follows from  Lemma \bref{lem:TooMuchInBox} and the equality $P(E)=1-P(E')$ (where $E'$ is the complement of $E$) which imply that 
\begin{multline}\label{eq:P(F)}
P(E)=1-P(E')\\
\geq 1-n\left(\frac{e^{4-0.8\alpha_2}}{\sqrt{\gamma\alpha_2}}
\left(\frac{e}{\gamma}\right)^{\gamma\alpha_2}
+\frac{e^{4+0.2\alpha_2}}{\sqrt{(\lfloor 0.2m\rfloor+1)}}\left(\frac{e\alpha_2}{\lfloor 0.2m\rfloor+1}\right)^{\lfloor 0.2m\rfloor+1}\right).
\end{multline}
For showing the inequality $P(G)\geq P(E)$, it is sufficient 
to show that when $E$ holds, then $G$ holds, i.e., that $E\subseteq G$. 
The rest of the proof shows this.

Indeed, consider an arbitrary site $p_k$, an arbitrary unit vector $\theta$. According to Step \beqref{method:EndpointImproved:1.01} of Method \bref{method:EndpointImproved}, if $y$ denotes the candidate to be the first temporary endpoint along the  ray which emanates from $p_k$ in the direction of  $\theta$, then we have $y:=p_k+4\sqrt{2}(\beta+1.01)s\theta$. Now we go to Step \beqref{step:twice_y} and there are a few possibilities. In the first one, $y$ is outside $X$, namely $d(y_X,p_k)<4\sqrt{2}(\beta+1.01)s$, and we move $y$ towards $p_k$ by re-defining $y:=y_X$, where $y_X$ is the intersection of the ray with the boundary of $X$. According to Method \bref{method:EndpointImproved}, this $y$ is the first temporary endpoint and we also denote it by $\tilde{y}$. Now  we go to Step \beqref{item:RectangularListOfsites} and create the rectangular list of sites. Method \bref{method:EndpointImproved} guarantees that from now on all the distance comparisons  that will be done will be restricted to the sites contained in the rectangular list of sites, namely to the sites located in $S_I[\tilde{y},d(\tilde{y},p_k)]$. 

The integer rectangle $S_I[\tilde{y},d(\tilde{y},p_k)]$ may not be contained entirely in $X$, but from \beqref{eq:alpha2}--\beqref{eq:m1m2} its intersection with $X$ is contained in some integer rectangle $Q'$ which is both contained in $X$ and contains $(1+\lfloor 8\sqrt{2}(\beta+1.01))\rfloor)^2$ buckets. Indeed, $S_I[\tilde{y},d(\tilde{y},p_k)]\cap X$ is a rectangle whose side lengths are composed of at most $2d(\tilde{y},p_k)+1$ buckets each, and hence at most $1+\lfloor 8\sqrt{2}(\beta+1.01))\rfloor$ buckets each;  however, each of the side lengths of $X$ is composed of, according to \beqref{eq:alpha2}--\beqref{eq:m1m2}, at least $2(1+\lfloor 8\sqrt{2}(\beta+1.01))\rfloor)+1$ buckets. (As an illustration, if $X=[0,300]\times [0,300]$ buckets, $s=1$, $\beta=3$, and $\tilde{y}=(0.4,50.2)$. Then $\tilde{y}$ is in a bucket whose indices  are $(0,50)$. In addition, $4\sqrt{2}(\beta+1.01)\approx 22.68$. Hence  $S_I[\tilde{y},4\sqrt{2}(\beta+1.01)]\cap X$ is contained in the integer rectangle $[0,24]\times [27,72]$, which is contained in the integer rectangle $Q':=[0,46]\times [26,72]$ which is indeed contained in $X$ and contains $(1+\lfloor 8\sqrt{2}(\beta+1.01))\rfloor)^2=46^2$ buckets. See also Figure \bref{fig:VoronoiParallelBoxImprove2}.)

In the second possibility in Step \beqref{step:twice_y} the result is that $y$ is inside $X$ and hence, according to Step \beqref{item:S_I[y,2 beta s]}, we consider the integer rectangle $S_I[y,2\beta s]$ and check whether  $d(y,p)\leq d(y,a)$ for every site $a$ in $S_I[y,2\beta s]$. The number of distance comparisons done in this process is bounded above by $\gamma\alpha-1$, as we prove in the last paragraph of the proof. Anyhow, if indeed $d(y,p)\leq d(y,a)$ for every site $a$ in $S_I[y,2\beta s]$, then we let $y:=y_X$ and this $y$ is the first temporary endpoint which we also denote by $\tilde{y}$. Otherwise, $y$ is outside the cell of $p_k$ and we consider $y$ to be the first temporary endpoint which we also denote by $\tilde{y}$. Now  we go to Step \beqref{item:RectangularListOfsites} and create the rectangular list of sites. Method \bref{method:EndpointImproved} guarantees that from now on all the distance comparisons  that will be done will be restricted to the sites contained in the rectangular list of sites, namely to the sites located in $S_I[\tilde{y},d(\tilde{y},p_k)]$. This integer rectangle may not be contained entirely in $X$, but as explained in the previous paragraph, from \beqref{eq:alpha2}--\beqref{eq:m1m2} its intersection with $X$ is contained in some integer rectangle $Q'$ which is both contained in $X$ and contains $(1+\lfloor 8\sqrt{2}(\beta+1.01))\rfloor)^2$ buckets.

Summarizing the above-mentioned discussion, starting from Step \beqref{item:RectangularListOfsites}, we need, for the distance comparisons along the considered ray, to consider only sites located in $S_I[\tilde{y},d(\tilde{y},p_k)]$. This rectangle is by itself contained in some integer rectangle $Q'$ which is composed of at most $\alpha:=(1+\lfloor 8\sqrt{2}(\beta+1.01))\rfloor)^2$ buckets. Since we assume that the event $E$  holds,  it follows that $Q'$ contains at most $\gamma \alpha-1$ sites, and since $S_I[\tilde{y},d(\tilde{y},p_k)]$ is contained in $Q'$, also $S_I[\tilde{y},d(\tilde{y},p_k)]$ (namely, the rectangular list of sites) contains at most $\gamma \alpha-1$ sites. Because of the way in which we handle the rectangular list of sites in Method \bref{method:EndpointImproved}, it follows that along the considered ray we perform at most $\gamma\alpha-1$ distance comparisons, starting from Step \beqref{item:RectangularListOfsites}. Since in Step \beqref{item:S_I[y,2 beta s]} we also perform at most $\gamma\alpha-1$ distance comparisons (as proved in the next paragraph), and since $p_k$ and $\theta$ where arbitrary, it follows that indeed the event $E$ implies the  event $G$, as claimed.  

Finally, it remains to show that the number of distance comparisons done in Step \beqref{item:S_I[y,2 beta s]} of Method \bref{method:EndpointImproved} is at most $\gamma \alpha-1$. Indeed, $S_I[y,2\beta s]\cap X$ is a rectangle whose side lengths are composed of at most $4\beta+1$ buckets each, and hence less than $1+\lfloor 8\sqrt{2}(\beta+1.01))\rfloor$ buckets each;  however, each of the side lengths of $X$ is composed of, according to \beqref{eq:alpha2}--\beqref{eq:m1m2}, at least $2(1+\lfloor 8\sqrt{2}(\beta+1.01))\rfloor)+1$ buckets; these considerations imply that even if $S_I[y,2\beta s]$ is not contained entirely in $X$,  its intersection with $X$ is contained in some integer rectangle $Q'$ which is both contained in $X$ and contains $(1+\lfloor 8\sqrt{2}(\beta+1.01))\rfloor)^2$ buckets. Since we assume that the event $E$  holds,  it follows that $Q'$ contains at most $\gamma \alpha-1$ sites, and since $S_I[y,2\beta s]$ is contained in $Q'$, also $S_I[y,2\beta s]$  contains at most $\gamma \alpha-1$ sites. Therefore we perform at most $\gamma\alpha-1$ distance comparisons in Step \beqref{item:S_I[y,2 beta s]}, as claimed.
\end{proof}

\begin{lem}\label{lem:beta_n}
Consider the setting of Lemma \bref{lem:EventEndpoint}. Then with probability  which is at least  the expression on the right-hand side of \beqref{eq:P(F)P(E)},  for all $k\in K$ and all rays shot during the computation of the cell of $p_k$ using Method \bref{method:EndpointImproved},  
one has that $r_k$, namely the maximum number of distance comparisons done 
along the rays, is not greater than $2((1+\lfloor 8\sqrt{2}(\beta+1.01)\rfloor)^2\gamma-1)$. Moreover, with probability  which is also  at least the expression on the right-hand side of  \beqref{eq:P(F)P(E)}, the total number of distance comparisons done in the computation of the whole Voronoi diagram (using Algorithm 1 and Method \bref{method:EndpointImproved}) is not greater than  $(\gamma(1+\lfloor 8\sqrt{2}(\beta+1.01))\rfloor)^2-1)\cdot 22n$. 
\end{lem}

\begin{proof}
Lemma \bref{lem:EventEndpoint} and its proof show that $r_k\leq 2(\gamma(1+\lfloor 8\sqrt{2}(\beta+1.01))\rfloor)^2-1)$ with probability which is not smaller than the expression on the right-hand side of \beqref{eq:P(F)P(E)}. Denote by $\rho_k$ the number of rays shot during the computation of the cell of $p_k$. Then the  total number of distance comparisons 
done in the computation of all the cells is bounded by $\sum_{k=1}^n r_k\rho_k$. 
Since $X$ is a rectangle, \beqref{eq:NumberShotRays} shows that $\sum_{k=1}^n\rho_k\leq 11n$ 
and hence, with probability which is not smaller than the expression on the right-hand side of \beqref{eq:P(F)P(E)}, we have $\sum_{k=1}^n r_k\rho_k\leq 2(\gamma(1+\lfloor 8\sqrt{2}(\beta+1.01))\rfloor)^2-1)\cdot 11n$,  as claimed. 
\end{proof}

It is now possible to prove Theorem   \bref{thm:CorrectnessOfTheAlg}\beqref{item:TimeComplexityUniformDistribution}.  
\begin{proof}[{\bf Proof of Theorem   \bref{thm:CorrectnessOfTheAlg}\beqref{item:TimeComplexityUniformDistribution}}]

It is sufficient to use Lemma \bref{lem:beta_n} with certain values of $\beta$ and $\gamma$ 
and to show that the second term in \beqref{eq:P(F)} is not greater than $\epsilon$.
Let $\gamma:=3$ and let $\beta$ be defined in \beqref{eq:beta}. As explained in Remark \bref{rem:TimeComplexityUniformDistribution}, both  \beqref{eq:ln(n/epsilon)} and \beqref{eq:alpha2m} do hold for all sufficiently large $n$. The relations  \beqref{eq:ln(n/epsilon)}-\beqref{eq:alpha2},  combined with straightforward calculations, show that  $(e/\gamma)^{\gamma}<3/4$, that 
\beqref{eq:alpha-kappa-m} holds with $\alpha:=\alpha_2$, and that 
\begin{equation}\label{eq:alpha1 alpha2}
\alpha_2/\alpha_1>4\sqrt{2}. 
\end{equation}
Therefore 
\begin{equation*}
\left(\frac{e}{\gamma}\right)^{\gamma\alpha_2}\leq \left(\frac{3}{4}\right)^{4\sqrt{2}\alpha_1}<0.2^{\alpha_1}<e^{-\alpha_1}. 
\end{equation*}
Since $n>100$, it follows from \beqref{eq:alpha1} and the inequality $\textnormal{OneTwo}\geq 1$ that $\alpha_1>9$. Hence  \beqref{eq:beta} implies that $\beta>1$. Thus we obtain from \beqref{eq:alpha2} that 
\begin{equation}\label{eq:alpha_2>529}
\alpha_2>(1+\lfloor 8\sqrt{2}\cdot 2\rfloor)^2=23^2=529. 
\end{equation}
We conclude from the previous inequalities that 
\begin{equation*}
\frac{e^{4-0.8\alpha_2}}{\sqrt{\gamma\alpha}}\left(\frac{e}{\gamma}\right)^{\gamma\alpha_2}<\frac{e^{-419}}{\sqrt{1587}}e^{-\alpha_1}<e^{-400}e^{-\alpha_1}.
\end{equation*}
Next we will prove that the second term in the parenthesis in \beqref{eq:P(F)} is smaller than $e^{-400}e^{-\alpha_1}$. Indeed, since \beqref{eq:alpha2m} implies that $(e\alpha_2)/(\lfloor 0.2m\rfloor+1)<e^{-1}$,  since $0.2m<\lfloor 0.2m\rfloor+1$, since \beqref{eq:alpha2} and the fact that $m=(\lfloor\sqrt{n}\rfloor)^2\geq 100$ imply that $15+\alpha_2<m$, and since \beqref{eq:alpha_2>529} and \beqref{eq:alpha1 alpha2} hold, it follows that 
\begin{multline*}
\frac{e^{4+0.2\alpha_2}}{\sqrt{(\lfloor 0.2m\rfloor+1)}}\left(\frac{e\alpha_2}{\lfloor 0.2m\rfloor+1}\right)^{\lfloor 0.2m\rfloor+1}<e^{4+0.2\alpha_2}e^{-0.2m}<e^{4+0.2\cdot (-15\alpha_2)}\\
<e^{4-2\alpha_2}e^{-\alpha_2}<e^{4-2\alpha_2}e^{-\alpha_1}<e^{-400}e^{-\alpha_1}.
\end{multline*}
Consequently, the above-mentioned discussion and \beqref{eq:alpha1} imply that 
\begin{multline*}
n\left(\frac{e^{4-0.8\alpha}}{\sqrt{\gamma\alpha}}
\left(\frac{e}{\gamma}\right)^{\gamma\alpha}
+\frac{e^{4+0.2\alpha}}{\sqrt{(\lfloor 0.2m\rfloor+1)}}\left(\frac{e\alpha}{\lfloor 0.2m\rfloor+1}\right)^{\lfloor 0.2m\rfloor+1}\right)\\
<2ne^{-400}e^{-\alpha_1}<n(1+2e^{-400})e^{-\alpha_1}<\epsilon.
\end{multline*}
Finally, Lemma \bref{lem:beta_n} and \beqref{eq:alpha2} ensure that with probability which is at least $1-\epsilon$, the number of distance comparisons is bounded above by $(3\alpha_2-1)\cdot 22n=O(n\log(n/\epsilon))$. 
\end{proof}

\section*{Acknowledgments}
I want to thank the referee for valuable comments which helped to improve the text. I also want to express my thanks to several people regarding this paper, especially to Omri Azencot to whom I am indebted for  implementing the algorithm so expertly and for helpful discussions, and to Renjie Chen for helpful discussion regarding \cite{ChenGotsman2013jour}. Parts of this work were done during various years, when I have been associated with the following places: The Technion, Haifa, Israel (2010, 2017--2018), the University of Haifa, Haifa, Israel (2010--2011, 2020--2023), the National Institute for Pure and Applied Mathematics  (IMPA), Rio de Janeiro, Brazil (2011--2013), and the Institute of Mathematics and Computational Sciences (ICMC), University of S\~ao Paulo, S\~ao Carlos, Brazil (2014--2015), and it is an opportunity for me to thank the Gurwin Foundation, BSF, and FAPESP. Special thanks are for a special postdoc fellowship (``P\'os-doutorado de Excel\^encia'') given to me when I was at IMPA.

\bibliographystyle{acm}
\bibliography{biblio}

\end{document}